\documentclass[journal]{IEEEtran}
\ifCLASSINFOpdf
\else
\fi

\usepackage{epsfig}\usepackage{epsfig}
\usepackage{amssymb}
\usepackage{amsthm}
\usepackage{amsmath}
\usepackage{cite}
\usepackage{mathrsfs}
\usepackage{cases}
\usepackage{multirow}
\usepackage{makecell}
\usepackage{url}
\usepackage{algorithm}%
\usepackage{algorithmicx}%
\usepackage{algpseudocode}%
\usepackage{subeqnarray}
\usepackage{subfigure}
\graphicspath{ {figure/} }
\usepackage{pifont,bbm,amsfonts,mathrsfs,amsmath,stmaryrd,amssymb}
\usepackage{stfloats}
\usepackage{amscd,graphicx,array,dsfont,texdraw,tikz,footnote,verbatim}%
\usepackage[normalem]{ulem}
\usepackage{mathtools}
\usepackage{bbm,bbding,amssymb,pifont,mathrsfs,amsfonts,graphicx,mathtools,color}%
\usepackage{amscd,array,enumerate,dsfont,texdraw,tikz,multicol,bbm}

\usetikzlibrary{arrows,shapes,snakes,shadows,positioning,automata,patterns}
\usetikzlibrary{trees,decorations.pathmorphing,decorations.markings}
\usepackage{setspace} 
\usepackage{tikz}
\usepackage{pgf}
\usepackage{amsmath}

\usepackage{amsmath, amsfonts, amssymb, amsthm}
\usepackage{pgfplots}
\pgfplotsset{compat=newest}
\usepackage{xxcolor}
\usetikzlibrary{arrows, decorations.markings}
\usetikzlibrary{arrows,shadows,petri}

\newcommand{\beq}{\begin{equation}}
\newcommand{\eeq}{\end{equation}}
\newcommand{\bqa}{\begin{eqnarray}}
\newcommand{\eqa}{\end{eqnarray}}

\definecolor{maroon}{rgb}{0.7,0,0}

\definecolor{ngreen}{rgb}{0.3,0.7,0.3}

\definecolor{golden}{rgb}{0.8,0.6,0.1}

\newtheorem{theorem}{\indent Theorem}
\newtheorem{lemma}{\indent Lemma}
\newtheorem{corollary}{\indent Corollary}
\newtheorem{definition}{\indent Definition}
\newtheorem{assumption}{\indent Assumption}
\newtheorem{myremark}{\indent Remark}

\newenvironment{remark}{\begin{myremark}\normalfont}
{\end{myremark}}

\newcommand{\col}{\mathrm{col}}
\begin{document}
\title{Local Differential Privacy for Distributed Stochastic Aggregative Optimization with Guaranteed Optimality}

\author{Ziqin Chen and Yongqiang Wang, \textit{Senior Member, IEEE}
\thanks{This work was supported by the National Science Foundation under Grant ECCS-1912702, Grant CCF-2106293, Grant CCF-2215088, Grant CNS-2219487, Grant CCF-2334449, and Grant CNS-2422312. (Corresponding author: Yongqiang Wang, email:yongqiw@clemson.edu).}
\vspace{-0.5em}
\thanks{Ziqin Chen and Yongqiang Wang are with the Department of Electrical and Computer Engineering, Clemson University, Clemson, SC 29634 USA.}}

\maketitle
\begin{abstract}
Distributed aggregative optimization underpins many cooperative optimization and multi-agent control systems, where each agent's objective function depends both on its local optimization variable and an aggregate of all agents' optimization variables. Existing distributed aggregative optimization approaches typically require access to accurate gradients of the objective functions, which, however, are often hard to obtain in real-world applications. For example, in machine learning, gradients are commonly contaminated by two main sources of noise: the randomness inherent in sampled data, and the additional variability introduced by mini-batch computations. In addition to the issue of relying on accurate gradients, existing distributed aggregative optimization approaches require agents to share explicit information, which could breach the privacy of participating agents. We propose an algorithm that can solve both problems with existing distributed aggregative optimization approaches: not only can the proposed algorithm guarantee mean-square convergence to an exact optimal solution when the gradients are subject to noise, it also simultaneously ensures rigorous differential privacy, with the cumulative privacy budget guaranteed to be finite even when the number of iterations tends to infinity. To the best of our knowledge, this is the first algorithm able to guarantee both accurate convergence and rigorous differential privacy in distributed aggregative optimization. Besides characterizing the convergence rates under nonconvex/convex/strongly convex conditions, we also rigorously quantify the cost of differential privacy in terms of convergence rates. Experimental results on personalized machine learning using benchmark datasets confirm the efficacy of the proposed algorithm.
\end{abstract}


\IEEEpeerreviewmaketitle

\section{Introduction}
Distributed aggregative optimization addresses cooperative optimization and control problems where each agent's objective function depends on both its local optimization variable and an aggregate of all agents' optimization variables. {\color{blue}Mathematically, the problem of distributed aggregative optimization can be formulated as follows:
\begin{equation}
\min_{\boldsymbol{x}\in\Omega}~ F(\boldsymbol{x})\!=\!\sum_{i=1}^{m}f_{i}(x_{i}, g(\boldsymbol{x})),~~g(\boldsymbol{x})\!=\!\frac{1}{m}\sum_{i=1}^{m}g_{i}(x_{i}).\label{primal}
\end{equation}
Here, $m$ denotes the number of agents, $\boldsymbol{x}=\col(x_{1},\cdots,x_{m})$ denotes the concatenated optimization variable, where $x_{i}\in\Omega_{i}$ with $\Omega_{i}\subseteq\mathbb{R}^{n_{i}}$ denoting the constraint set. $g(\boldsymbol{x}):\Omega\mapsto\mathbb{R}^{r}$ is an aggregative function with $\Omega$ denoting the Cartesian product of all $\Omega_{i}$ (i.e., $\Omega=\Omega_{1}\times\cdots\times\Omega_{m}$) and each $g_{i}(x_{i}):\Omega_{i}\mapsto\mathbb{R}^{r}$ denoting the agent $i$'s private function. $f_{i}(x_{i}, g(\boldsymbol{x})): \Omega_{i}\times\mathbb{R}^{r}\mapsto\mathbb{R}$ is the local objective function of agent $i$.} Problem~\eqref{primal} has found various applications, with typical examples including warehouse placement optimization~\cite{lixiuxian1}, multi-vehicle charging coordination~\cite{EVcharging,truthfulness}, cooperative robot surveillance~\cite{Carnevale1}, and personalized machine learning~\cite{personalized1}.

To date, plenty of distributed aggregative optimization approaches have been proposed, for both time-invariant~\cite{lixiuxian1,yipeng,Carnevale2,wenguanghui2,wang2024momentum,cai2024distributed,mengmin,zhou2025distributed,zhang2025distributed,li2025distributed,huang2025distributed} and time-varying~\cite{Carnevale1,personalized1,yang2024class} objective functions. These approaches typically require participating agents to have access to accurate gradients of the objective functions, which, however, may not be feasible in practical applications. For example, in machine learning applications involving massive datasets, gradients are typically computed using randomly sampled data, which inherently introduces noise. Moreover, the use of mini-batch computations adds further randomness, increasing the variability in the gradient estimates available to the agents~\cite{onlinestochatsic1,onlinestochatsic2,onlinestochatsic3,control}. These challenges necessitate a reevaluation of traditional approaches that assume accurate gradient information, thereby motivating the development of distributed stochastic aggregative optimization algorithms.

However, developing an effective distributed stochastic algorithm for aggregative optimization is nontrivial due to~the unique compositional structure of the objective functions. Specifically, unlike traditional distributed stochastic optimization~\cite{onlinestochatsic1,onlinestochatsic2,onlinestochatsic3,control}, the objective functions in distributed aggregative optimization explicitly depend on {\color{blue}an inner-level aggregative function $g(\boldsymbol{x})$ (see~\eqref{primal})}, which naturally arises in many practical applications. For example, in multi-vehicle charging coordination~\cite{EVcharging}, each user's cost (objective) function typically relies on electricity pricing, which usually varies as a function of the aggregated charging demand. This compositional dependence makes the gradient computation of objective functions inseparable from the estimation of the aggregative function. 
In the presence of noise, this inseparability leads to tangled errors in computed gradients and aggregative function estimates, substantially complicating the analysis. Moreover, in practice, computations must be conducted on sampled data, which introduces additional complexity. This is because at each iteration, both the aggregative function and the objective functions are defined by the sampled data, making them inherently time-varying and interdependent across iterations. Consequently, this creates significant challenges for algorithm design and convergence analysis.
Although some progress has been made to address noise in distributed aggregative optimization~\cite{lixiuxian2,wenguanghui1}, these results only consider the special case where the objective function is stochastic and ignores noise in aggregative function estimation, which significantly simplifies the convergence analysis. In fact, when there is noise in aggregative function estimation, the existing approach in~\cite{lixiuxian2} will be subject to steady-state optimization errors. To the best of our knowledge, no results are available that can address stochasticity in both the objective functions and the aggregative function in distributed aggregative optimization.

Another potential limitation of existing distributed aggregative optimization algorithms in~\cite{lixiuxian1,yipeng,Carnevale2,wenguanghui2,wang2024momentum,cai2024distributed,mengmin,zhou2025distributed,zhang2025distributed,li2025distributed,huang2025distributed,Carnevale1,personalized1,yang2024class,lixiuxian2,wenguanghui1} is that they require participating agents to share explicit variables/gradient estimates at every iteration, which may breach the privacy of participating agents when sensitive information is involved~\cite{leakage1}. In fact, recent studies have shown that even without direct access to raw data, external adversaries can accurately reconstruct raw data from intercepted gradient estimates~\cite{leakage1,leakage2,leakage3}. To address privacy protection in distributed optimization and learning, a variety of privacy mechanisms have been proposed~\cite{homomorphic1,homomorphic2,homomorphic3,noisesinject1,noisesinject2,noisesinject3,stepsizes,statedecomposition,quantization}. Among these mechanisms, differential privacy (DP) has emerged as the de facto standard for privacy protection owing to its rigorous mathematical foundations, ease of implementation, and robustness to postprocessing~\cite{dwork2014}. Plenty of DP-based approaches have been proposed for distributed optimization and learning (see, e.g.,~\cite{huang,nozari2016, dingtie,Tailoring,Assumption3topology,shilin,xie2025,he2025,DOLA,xiong,lu2020,zijiGT,zhangjifeng1,zhangjifeng2}). However, most of these results have to sacrifice convergence accuracy for privacy~\cite{huang,nozari2016, dingtie,Tailoring,Assumption3topology,shilin,xie2025,he2025,DOLA,xiong,lu2020}. Only recently have our works~\cite{Tailoring,Assumption3topology,zijiGT} and others~\cite{zhangjifeng1,zhangjifeng2} achieved both accurate convergence and rigorous DP with a finite cumulative privacy budget in the infinite time horizon. Nevertheless, these methods are tailored for basic distributed optimization without the aggregative function, and therefore, cannot be directly applied to distributed aggregative optimization to ensure both rigorous DP and accurate convergence. This is because, unlike basic distributed optimization, distributed aggregative optimization requires an additional inner-loop tracking operation to estimate the aggregative function. This leads to a faster growth in the privacy budget under commonly used DP-noise models, ultimately resulting in an unbounded cumulative privacy budget as iterations proceed. When greater DP noise is used, the additional tracking operation fast accumulates errors in the aggregative term and leads to significant optimization errors, as evidenced by our experimental results in Fig.~\ref{mnist} and Fig.~\ref{cifar10}. In fact, to the best of our knowledge, no existing results are able to achieve both rigorous DP and accurate convergence in distributed aggregative optimization even in the absence of stochasticity. 

In this article, we propose an approach for distributed stochastic aggregative optimization that can achieve both accurate convergence and rigorous DP with a finite cumulative privacy budget even when the number of iterations tends to infinity. Our basic idea is a new algorithmic design that departs from the framework initially proposed in~\cite{lixiuxian1} and subsequently adopted by all existing distributed aggregative optimization algorithms~\cite{personalized1,Carnevale1,yipeng,Carnevale2,wenguanghui2,wang2024momentum,cai2024distributed,mengmin,zhou2025distributed,zhang2025distributed,li2025distributed,huang2025distributed,lixiuxian2,wenguanghui1,yang2024class}, which was unfortunately vulnerable to DP noise. The main contributions are summarized as follows:

\begin{itemize}
\item By designing a new algorithmic architecture for distributed aggregative optimization, we achieve both rigorous DP and accurate convergence to an exact optimal solution, even when only noisy gradient estimates are available. This is in sharp contrast to existing DP solutions for distributed optimization~\cite{huang,nozari2016, dingtie,shilin,xie2025,he2025,DOLA,xiong,lu2020} that have to trade convergence accuracy for privacy. Our algorithm is also fundamentally different from existing DP approaches in~\cite{Tailoring,Assumption3topology,zijiGT,zhangjifeng1,zhangjifeng2} for conventional distributed optimization, which, while achieving rigorous DP and accurate convergence, are inapplicable to the setting of distributed aggregative optimization due to the additional difficulties brought by the nested loops of gradient estimation and intertwined dynamics of optimization and aggregative function tracking. It is worth noting that, to ensure privacy in a fully distributed setting, we employ the local differential privacy (LDP) framework~\cite{kasiviswanathan2011,zhangjiaojiaoLDP,zijiGT}, in which random perturbations are performed locally by each agent, thereby protecting privacy at the agent level. To the best of our knowledge, this is the first time that both accurate convergence and rigorous LDP are ensured in distributed aggregative optimization.

\item To enhance the robustness against gradient noise and DP noise, we propose a new algorithm structure that fundamentally differs from all existing distributed aggregative optimization algorithms, which are susceptible to noise \cite{lixiuxian1,yipeng,Carnevale2,wenguanghui2,wang2024momentum,cai2024distributed,mengmin,zhou2025distributed,zhang2025distributed,li2025distributed,huang2025distributed,Carnevale1,personalized1,yang2024class,lixiuxian2,wenguanghui1}. Specifically, our approach utilizes two robust tracking operations followed by a dynamic average-consensus-based gradient descent step. This design not only eliminates the accumulation of DP-noise variances in both aggregative function and gradient estimations, but it also ensures diminishing algorithmic sensitivity, thereby ensuring a finite cumulative privacy budget in the infinite time horizon. In the meantime,~by leveraging the unique streaming data patterns and judiciously designing the decaying rates of stepsizes and DP-noise variances, we eliminate the effects of both noisy gradients and DP noise on optimization accuracy, thereby ensuring mean-square convergence to the exact optimal solution.

\item We systematically characterize the convergence rates of our proposed algorithm under nonconvex/convex/strongly convex global objective functions. This is different from existing distributed aggregative optimization results in~\cite{personalized1,Carnevale1,lixiuxian1,wenguanghui2,wang2024momentum,mengmin,zhou2025distributed,li2025distributed,huang2025distributed}, which focus exclusively on strongly convex objective functions and~\cite{yipeng,cai2024distributed,yang2024class,zhang2025distributed,lixiuxian2,wenguanghui1}, which only consider the convex case. Our result also differs from the sole existing nonconvex study for distributed aggregative optimization in~\cite{Carnevale2}, which only proves asymptotic convergence in the continuous-time domain and does not provide explicit convergence rates.

\item Even without taking privacy into consideration, the proposed algorithm and its theoretical derivations are of independent interest. Given that all existing distributed aggregative optimization approaches require a noise-free gradient to ensure accurate convergence, our approach is the first to guarantee convergence in distributed aggregative optimization under stochasticity in both the objective functions and the aggregative function. Note that 
by considering a stochastic aggregative function, our approach broadens the scope of aggregative optimization to a much wider range of practical domains, such as personalized machine learning~\cite{personalized2,personalized3}, which are not addressable by existing methods~\cite{lixiuxian1,yipeng,Carnevale2,wenguanghui2,wang2024momentum,cai2024distributed,mengmin,zhou2025distributed,zhang2025distributed,li2025distributed,huang2025distributed,Carnevale1,personalized1,yang2024class,lixiuxian2,wenguanghui1}. In addition, our result stands in sharp contrast to existing online results in~\cite{lixiuxian2,yang2024class,personalized1,Carnevale1}, which cannot guarantee convergence to the exact solution.
\item We evaluate our algorithm on a distributed personalized machine learning problem using the benchmark ``MNIST" dataset and the ``CIFAR-10" dataset, respectively. The results confirm the efficacy of our approach in real-world applications.
\end{itemize}

The rest of the paper is organized as follows. Sec. \ref{problemstatement}~presents some preliminaries and the problem formulation. Sec. \ref{algorithmdesign} introduces our algorithm. Sec.~\ref{convergenceanalysis} analyzes the optimization accuracy and convergence rates. Sec.~\ref{privacy} establishes a rigorous LDP guarantee. Sec.~\ref{experiment} presents experimental results. Sec.~\ref{conclusion} concludes the paper.

\textit{Notations:} We use $\mathbb{R}^{n}$ to denote the Euclidean space of dimension $n$ and $\mathbb{N}$ ($\mathbb{N}^{+}$) to denote the set of non-negative (positive) integers. We use $\boldsymbol{1}_{m}$ and $\boldsymbol{0}_{m}$ to denote $m$-dimensional column vectors of all ones and all zeros, respectively. We let $W^{\top}$ denote the transpose of a matrix $W$. We use $\langle\cdot,\cdot\rangle$ to denote the inner product of two vectors and $\|\cdot\|$ to denote the Euclidean norm of a vector. We write $\text{col}(x_{1},\cdots,x_{m})$ for the stacked column vector of $x_{1},\cdots,x_{m}$. We denote $\mathbb{P}[\mathcal{A}]$ as the probability of an event $\mathcal{A}$ and $\mathbb{E}[x]$ as the expected value of a random variable $x$. $\Pi_{\Omega}$ denotes the Euclidean projection of a vector $x\in\mathbb{R}^{n}$ onto a compact set $\Omega$, i.e., $\Pi_{\Omega}:= \arg\min_{y\in\Omega}\|x-y\|$. We use $\text{Lap}(\nu_{i})$ to denote the Laplace distribution with a parameter $\nu_{i}>0$, featuring a probability density function $\frac{1}{2\nu_{i}}e^{\frac{-|x|}{\nu_{i}}}$. $\text{Lap}(\nu_{i})$ has a mean of zero and a variance of $2 \nu_{i}^2$. We add an overbar to a letter to denote the average of $m$ agents, e.g., $\bar{y}^{t}=\frac{1}{m}\sum_{i=1}^{m}y_{i}^{t}$, and use bold font to represent stacked
vectors of $m$ agents, e.g., $\boldsymbol{y}^{t}=\col(y_{1}^{t},\cdots,y_{m}^{t})$.
\section{Preliminaries and Problem Formulation}\label{problemstatement}
\subsection{Distributed stochastic aggregative optimization}
{\color{blue}We consider distributed aggregative optimization problem in~\eqref{primal}.} In machine learning applications, the functions $f_{i}$ and $g_{i}$ are typically given by 
\begin{equation}
	\begin{aligned}
		f_{i}(x_{i}, g(\boldsymbol{x}))&=\mathbb{E}_{\varphi_{i}\sim\mathcal{P}_{i}^{f}}[h(x_{i},g(\boldsymbol{x});\varphi_{i})],\\
		g_{i}(x_{i})&=\mathbb{E}_{\xi_{i}\sim \mathcal{P}_{i}^{g}}[l(x_{i};\xi_{i})],\label{primal2}
	\end{aligned}
\end{equation}
respectively. In~\eqref{primal2}, $\varphi_{i}$ and $\xi_{i}$ represent random data points which usually follow unknown and heterogeneous distributions $\mathcal{P}_{i}^{f}$ and $\mathcal{P}_{i}^{g}$, respectively, across different agents. An optimal solution to~\eqref{primal} is denoted as $\boldsymbol{x}^{*}=\col(x_{1}^{*},\cdots,x_{m}^{*})$.

We use the following standard assumptions on $f_{i}$ and $g_{i}$:
\begin{assumption}\label{A1}
	The constraint set $\Omega_{i}$ is nonempty, compact, and convex. In addition, 
	
	(i) $f_{i}(x,y)$ is $L_{f,1}$- and $L_{f,2}$-Lipschitz continuous with respect to $x$ and $y$, respectively. $\nabla f_{i}(x,y)$ is $\bar{L}_{f,1}$- and $\bar{L}_{f,2}$-Lipschitz continuous with respect to $x$ and $y$, respectively; 
	
	(ii) $g_{i}(x)$ and $\nabla g_{i}(x)$ are $L_{g}$- and $\bar{L}_{g}$-Lipschitz continuous, respectively.
\end{assumption}
{\color{blue}In Assumption~\ref{A1}, the Lipschitz continuity of the objective function is a desirable property in machine learning applications, as it often helps to improve the robustness and generalization of the learned model~\cite{robust}. Various techniques have been developed to enforce this property, such as spectral normalization~\cite{PLip1}, spectral norm regularization~\cite{PLip2}, and gradient penalty methods~\cite{PLip3}. Separately, the Lipschitz continuity of gradients is a standard assumption for convergence analysis in existing distributed aggregative optimization works (see, e.g.,~\cite{lixiuxian1,yipeng,Carnevale2,wenguanghui2,wang2024momentum,cai2024distributed,mengmin,zhou2025distributed,zhang2025distributed,li2025distributed,huang2025distributed,Carnevale1,personalized1,yang2024class,lixiuxian2,wenguanghui1}). Furthermore, when the constraint set~$\Omega_{i}$ is compact, many loss functions and their corresponding gradients in machine learning satisfy this assumption, with typical examples including the RoBoSS loss function \cite{RoBoss} and the widely used cross-entropy loss and its variants \cite{Smooth}.} In addition, unlike existing distributed aggregative optimization results, in, e.g.,~\cite{personalized1,Carnevale1,lixiuxian1,wenguanghui2,wang2024momentum,mengmin,zhou2025distributed,li2025distributed,huang2025distributed,yipeng,cai2024distributed,yang2024class,zhang2025distributed,lixiuxian2,wenguanghui1} that require the objective functions to be strongly convex or convex, we allow both $f_{i}$ and $g_{i}$ to be nonconvex.

We also need the following assumption, which is commonly used in distributed stochastic optimization~\cite{onlinestochatsic2,onlinestochatsic3,lixiuxian2}:
\begin{assumption}\label{A2}
\vspace{-0.2em}
The data points $\{\varphi^{t}\}$ and $\{\xi^{t}\}$ are independent and identically distributed across iterations, respectively. In addition, the stochastic oracles $l(x;\xi)$, $\nabla l(x;\xi)$, and $\nabla h(x,y;\varphi)$ are unbiased, with variances bounded by $\sigma_{g,0}^2$, $\sigma_{g,1}^2$, and $\sigma_{f}^2$, respectively.
\end{assumption}

We consider the distributed aggregative optimization problem~\eqref{primal} in a fully distributed setting. Namely, no agent has direct access to the aggregative function $g(\boldsymbol{x})$. Instead, each agent has to construct a local estimate of the aggregate through local interactions with its neighbors. We 
describe the local interaction using a matrix $W=\{w_{ij}\}$, where $w_{ij}>0$ if agents $i$ and $j$ can communicate
with each other, and $w_{ij}\!=\!0$ otherwise. 
For an agent $i\!\in\![m]$, we represent its neighboring set as $\mathcal{N}_{i}=\{j\in[m]|w_{ij}>0\}$ and let $w_{ii}=-\sum_{j\in\mathcal{N}_{i}}w_{ij}$. Furthermore, we make the following assumption on $W$:
{\color{blue}\begin{assumption}\label{A3}
The communication graph is undirected and connected. Moreover, the matrix $I+W$ is doubly stochastic and satisfies $\|I+W-\frac{\boldsymbol{1}_{m}\boldsymbol{1}_{m}^{\top}}{m}\|<1$.
\end{assumption}}
According to Lemma 4 in~\cite{Assumption3topology}, we have $\|I_{m}+W-\frac{\boldsymbol{1}_{m}\boldsymbol{1}_{m}^{\top}}{m}\|<1-|\rho_{2}|$, where $\rho_{2}$ is $W$'s second largest eigenvalue.

Since the local objective functions $f_{i}$ and $g_{i}$ are defined as expectations over random data $\varphi_{i}$ and $\xi_{i}$ sampled from unknown distributions $\mathcal{P}_{i}^{f}$ and $\mathcal{P}_{i}^{g}$, respectively, they are inaccessible in practice and an analytical solution to problem~\eqref{primal} is unattainable. To solve this issue, we focus on addressing the following empirical risk minimization (ERM) problem with sequentially arriving data:
\begin{equation}
\!\!\!\min_{\boldsymbol{x}\in\Omega} F^{t}(\boldsymbol{x})\!=\!\sum_{i=1}^{m}f_{i}^{t}(x_{i}, g^{t}(\boldsymbol{x})),~ g^{t}(\boldsymbol{x})\!=\!\frac{1}{m}\sum_{i=1}^{m}g_{i}^{t}(x_{i}),\label{primal3}
\end{equation}
where $g_{i}^{t}$ and $f_{i}^{t}$ are given by $g_{i}^{t}(x_{i})=\frac{1}{t+1}\sum_{k=0}^{t}l(x_{i};\xi_{i}^{k})$ and $f_{i}^{t}(x_{i},g^{t}(\boldsymbol{x}))=\frac{1}{t+1}\sum_{k=0}^{t}h(x_{i},g^{t}(\boldsymbol{x});\varphi_{i}^{k})$, respectively.

In~\eqref{primal3}, it is clear that sequentially arriving data $\xi_{i}^{t}$ cause the aggregative function $g^{t}(\boldsymbol{x})$ to vary with time. This is different from existing distributed stochastic aggregative optimization results~\cite{lixiuxian2,wenguanghui1}, which only consider time-varying $f_{i}^{t}(x_{i},g(\boldsymbol{x}))$ while assuming $g(\boldsymbol{x})$ to be time-invariant, which makes convergence analysis much easier.

Problem~\eqref{primal3} serves as an approximation to the original problem~\eqref{primal}. According to the law of large numbers~\cite{largelaw}, we have $\lim_{t\rightarrow\infty}g^{t}(\boldsymbol{x})=g(\boldsymbol{x})$ and $\lim_{t\rightarrow\infty}f_{i}^{t}(x_{i},g^{t}(\boldsymbol{x}))=f_{i}(x_{i},g(\boldsymbol{x}))$ (which can be obtained rigorously following the line of reasoning in Lemma 2 of our prior work~\cite{zijiGT} or Section 5.1.2 of~\cite{largelaw}). In fact, this is an intrinsic property of our ERM problem setting, which gradually eliminates the influence of noise in gradient estimations on convergence accuracy as time $t$ tends to infinity.
\vspace{-0.4em}

\subsection{Local differential privacy (LDP)}
We consider the LDP framework to protect each participating agent's privacy in the most severe scenario, where all communication channels
can be compromised and no agents
are trustworthy. To this end, we first denote the dataset $\mathcal{D}_{i}^{t}$ accessible to agent $i$ at each time $t$ as $\mathcal{D}_{i}^{t}=\{\xi_{i}^{1},\cdots,\xi_{i}^{t}\}$. Based on this, the concept of adjacency between the local datasets of agent $i$ is given as follows:
\begin{definition}[Adjacency\cite{zhangjifeng1}]\label{definition1} For any $t\in\mathbb{N}^{+}$ and any agent $i\in[m]$, given two local datasets  $\mathcal{D}_{i}^{t}=\{\xi_{i}^{1},\cdots,\xi_{i}^{k},\cdots,\xi_{i}^{t}\}$ and $\mathcal{D}_{i}^{\prime t}=\{\xi_{i}^{\prime 1},\cdots,\xi_{i}^{\prime k},\cdots,\xi_{i}^{\prime t}\}$, $\mathcal{D}_{i}^{t}$ is said to be adjacent to $\mathcal{D}_{i}^{\prime t}$ if there exists a time instant $k\in[1,t]$ such that $\xi_{i}^{k}\neq \xi_{i}^{\prime k}$ while $\xi_{i}^{p}=\xi_{i}^{\prime p}$ for all $p\in[1,t]$ and $p\neq k$.
\end{definition}

Definition~\ref{definition1} implies that two local datasets $\mathcal{D}_{i}^{t}$ and $\mathcal{D}_{i}^{\prime t}$ are adjacent if and only if they differ in one entry while all other entries are the same. It is worth noting that in problem~\eqref{primal}, each agent $i$ has two local datasets, i.e., $\mathcal{D}_{i,f}^{t}$ and $\mathcal{D}_{i,g}^{t}$, and hence, we allow both $\mathcal{D}_{i,f}^{t}$ and $\mathcal{D}_{i,g}^{t}$ to have one different entry from their respective adjacent datasets for any $t\in\mathbb{N}^{+}$. This introduces additional complexity to privacy analysis and LDP mechanism design compared to existing DP results for conventional distributed optimization~\cite{huang,nozari2016, dingtie,Tailoring,Assumption3topology,shilin,xie2025,he2025,DOLA,xiong,lu2020,zijiGT,zhangjifeng1,zhangjifeng2}, where each agent has only one local dataset. 

Now, we are in a position to present the definition of LDP:

\begin{definition}[Local differential privacy~\cite{dwork2014,kasiviswanathan2011}]\label{D2}
\vspace{-0.2em}
Let $\mathcal{A}_{i}(\mathcal{D}_{i},\theta_{-i})$ be an implementation of a distributed algorithm by agent $i$, which takes agent $i$'s dataset $\mathcal{D}_{i}$ and all received information $\theta_{-i}$ as input. Then, agent $i$'s implementation $\mathcal{A}_{i}$ is $\epsilon_{i}$-locally differentially private if for any adjacent datasets $\mathcal{D}_{i}$ and $\mathcal{D}_{i}^{\prime}$, and the set of all possible observations $\mathcal{O}_{i}$, the following inequality holds:
\begin{equation}
\vspace{-0.2em}
\mathbb{P}[\mathcal{A}_{i}(\mathcal{D}_{i},\theta_{-i})\in \mathcal{O}_{i}] \leq e^{\epsilon_{i}} \mathbb{P}[\mathcal{A}_{i}(\mathcal{D}_{i}^{\prime},\theta_{-i})\in \mathcal{O}_{i}].\nonumber
\end{equation}
\end{definition}
The definition of $\epsilon$-LDP ensures that no third party can infer agent $i$'s private data from shared messages. It can be seen that a smaller privacy budget $\epsilon_{i}$ means a higher level of privacy protection. It is also worth noting that for each agent $i$'s implementation $\mathcal{A}_{i}$, each data access will consume some budget $\epsilon_{i}^{t}$. Therefore, to maintain LDP over $T$ iterations, including the case of $T\rightarrow\infty$, the cumulative privacy budget must not exceed some preset finite value $\epsilon_{i}$, i.e., $\sum_{t=1}^{T}\epsilon_{i}^{t}\leq \epsilon_{i}$. Clearly, if $\sum_{t=1}^{T}\epsilon_{i}^{t}$ grows to infinity as the iteration number $T$ tends to infinity, privacy will be lost eventually in the infinite-time horizon.
\vspace{-0.4em}

\section{Locally Differentially Private Distributed Aggregative Optimization Algorithm}\label{algorithmdesign}
We propose Algorithm~\ref{algorithm1} to solve the stochastic aggregative optimization problem~\eqref{primal} while ensuring rigorous $\epsilon_{i}$-LDP with a finite cumulative privacy budget in the infinite time horizon.
\begin{algorithm}[H]
\caption{LDP design for distributed stochastic aggregative optimization (from agent $i$'s perspective)}
\label{algorithm1} 
\begin{algorithmic}[1]
\State {\bfseries Input:} Random initialization $\mathbf{x}_{i}^{0}\in \Omega$, $y_{i}^{0}\in{\mathbb{R}^{r}}$, and $z_{i}^{0}\in{\mathbb{R}^{r}}$; matrix $R_{i}\!=\![\boldsymbol{0}_{n_{i}\times n_{1}},\cdots, I_{n_{i}},\cdots, \boldsymbol{0}_{n_{i}\times n_{m}}]\!\in\!\mathbb{R}^{n_{i}\times n}$; stepsizes $\lambda_{x}^{t}=\frac{\lambda_{x}^{0}}{(t+1)^{v_{x}}}$,  $\lambda_{y}^{t}=\frac{\lambda_{y}^{0}}{(t+1)^{v_{y}}}$, and $\lambda_{z}^{t}=\frac{\lambda_{z}^{0}}{(t+1)^{v_{z}}}$ with $\lambda_{x}^{0},\lambda_{y}^{0},\lambda_{z}^{0}>0$ and $v_{x},v_{y},v_{z}\in(0,1)$; DP noises $\chi_{i}^{t}\in\mathbb{R}^{r}$, $\zeta_{i}^{t}\in\mathbb{R}^{r}$, and $\boldsymbol{\vartheta}_{i}^{t}\in\mathbb{R}^{n}$ (with $n=\sum_{i=1}^{m}n_{i}$; here, we use bold font to represent stacked vectors of $m$ agents) satisfying Assumption~\ref{A4}. 
\For{$t=0,1,\cdots,T-1$}
\State Acquire current data samples $\varphi_{i}^{t}\sim\mathcal{P}_{i}^{f}$ and $\xi_{i}^{t}\sim\mathcal{P}_{i}^{g}$;
\State $y_{i}^{t+1}=y_{i}^{t}+\sum_{j\in\mathcal{N}_{i}}w_{ij}(y_{j}^{t}+\chi_{j}^{t}-y_{i}^{t})+\lambda_{y}^{t}g_{i}^{t}(R_{i}\mathbf{x}_{i}^{t})$;
\State $z_{i}^{t+1}=z_{i}^{t}+\sum_{j\in\mathcal{N}_{i}}w_{ij}(z_{j}^{t}+\zeta_{j}^{t}-z_{i}^{t})$
\Statex \hspace*{4em} $+\lambda_{z}^{t}\nabla_{y} f_{i}^{t}(R_{i}\mathbf{x}_{i}^{t},\tilde{y}_{i}^{t})$, with $\tilde{y}_{i}^{t}\triangleq\textstyle\frac{1}{\lambda_{y}^{t}}(y_{i}^{t+1}-y_{i}^{t})$;
\State $\mathbf{x}_{i}^{t+1}=\Pi_{\Omega}\Big\{\mathbf{x}_{i}^{t}+\sum_{j\in\mathcal{N}_{i}}w_{ij}(\mathbf{x}_{j}^{t}+\boldsymbol{\vartheta}_{j}^{t}-\mathbf{x}_{i}^{t})$
\Statex \hspace*{4em} $-\lambda_{x}^{t}R_{i}^{\top}\big(\nabla_{x}f_{i}^{t}(R_{i}\mathbf{x}_{i}^{t},\tilde{y}_{i}^{t})+\nabla g_{i}^{t}(R_{i}\mathbf{x}_{i}^{t})\tilde{z}_{i}^{t}\big)\Big\}$,
\Statex \hspace*{4em} with $\tilde{z}_{i}^{t}\triangleq\textstyle \frac{1}{\lambda_{z}^{t}}(z_{i}^{t+1}-z_{i}^{t})$;
\State Send obscured parameters  $y_{i}^{t+1}\!+\!\chi_{i}^{t+1}$, $z_{i}^{t+1}\!+\!\zeta_{i}^{t+1}$, and $\mathbf{x}_{i}^{t+1}\!+\boldsymbol{\vartheta}_{i}^{t+1}$ to its neighbors $j\in\mathcal{N}_{i}$.
\EndFor
\end{algorithmic}
\end{algorithm}
The injected DP noises satisfy the following assumption:
\begin{assumption}\label{A4}
\vspace{-0.5em}
For each agent $i\in[m]$ and at any time $t\in\mathbb{N}$, each element of DP-noise vectors $\boldsymbol{\vartheta}_{i}^{t}$, $\chi_{i}^{t}$, and $\zeta_{i}^{t}$ follows Laplace distributions  $\text{Lap}(\frac{\sigma_{i,x}}{\sqrt{2}(t+1)^{\varsigma_{i,x}}})$, $\text{Lap}(\frac{\sigma_{i,y}}{\sqrt{2}(t+1)^{\varsigma_{i,y}}})$, and $\text{Lap}(\frac{\sigma_{i,z}}{\sqrt{2}(t+1)^{\varsigma_{i,z}}})$, respectively, where  $\sigma_{i,x}$, $\sigma_{i,y}$, and $\sigma_{i,z}$ are positive constants and the rates of noise variances satisfy
{\color{blue}\begin{equation}
\varsigma_{i,x}<v_{x}-v_{z},\quad \varsigma_{i,y}<v_{y},\quad\text{and}\quad\varsigma_{i,z}<v_{z},~\forall i\in[m],\nonumber
\vspace{-0.3em}
\end{equation}
with $v_{x}$, $v_{y}$, and $v_{z}$ being the decaying rates of the stepsizes $\lambda_{x}^{t}$, $\lambda_{y}^{t}$ and $\lambda_{z}^{t}$, respectively, in Algorithm~\ref{algorithm1}.}
\end{assumption}
{\color{blue}\begin{remark}
\vspace{-0.5em}
The stepsizes in Algorithm~\ref{algorithm1}, i.e., $\lambda_{x}^{t}$, $\lambda_{y}^{t}$, and $\lambda_{z}^{t}$, can be hard-coded offline in all agents' programs, and hence, do not require any adjustment
or coordination among agents during algorithmic implementation. Furthermore, since the stepsizes can be hard-coded prior to algorithmic implementation, heterogeneity in stepsizes across agents can be fully avoided and therefore does not to be considered in our design. In addition, we note that all noise parameters in Algorithm~\ref{algorithm1}, i.e.,~$\sigma_{i,x}$, $\sigma_{i,y}$, $\sigma_{i,z}$, $\varsigma_{i,x}$, $\varsigma_{i,y}$, and $\varsigma_{i,z}$, can be locally selected by each agent $i$ based on its own practical privacy needs. Hence, our algorithm is inherently robust to heterogeneous noise parameters across agents.
\vspace{-0.5em}
\end{remark}}
In line 6 of Algorithm \ref{algorithm1}, we use dynamic average-consensus-based gradient descent to enable each agent $i$ to estimate the optimization variables of all other agents, i.e., $\mathbf{x}_{i}^{t}\!=\!\col\{{\rm x}_{i(1)}^{t},\cdots,x_{i}^{t},\cdots,{\rm x}_{i(m)}^{t}\}$, where $x_{i}^{t}$ denotes agent~$i$'s true optimization variable and ${\rm x}_{i(j)}^{t}$ denotes its estimate of agent $j$'s optimization variable for all $j\!\neq\!i$ and $j\!\in\![m]$. {\color{blue}(Here, each agent $i$'s estimate of its own optimization variable coincides with the variable itself, i.e., ${\rm x}_{i(i)}^{t}\!=\!x_{i}^{t}$. Hence, we directly use $x_{i}^{t}$ to represent agent $i$'s estimate of its own optimization variable throughout the paper.)} This estimation of neighbors' optimization variables is an important novelty in our design (absent in all existing distributed aggregative optimization algorithms~\cite{lixiuxian1, personalized1,Carnevale1,yipeng,Carnevale2,wenguanghui2,wang2024momentum,cai2024distributed,mengmin,zhou2025distributed,zhang2025distributed,li2025distributed,huang2025distributed,lixiuxian2,wenguanghui1,yang2024class}), and it is key to ensuring a finite cumulative privacy budget in the infinite time horizon. Specifically, by employing average consensus on $\mathbf{x}_{i}^{t}$, we can ensure a diminishing algorithmic sensitivity (see Eq.~\eqref{4T121}), which, combined with meticulously designed diminishing DP-noise variances, can ensure a finite cumulative privacy budget in the infinite time horizon (see Eq.~\eqref{4T19}). In contrast, in existing algorithms, the difference $\|x_{i}^{t}-x_{i}^{\prime t}\|$ of state trajectories under adjacent datasets will grow over time, which will in turn lead to a growing sensitivity, and hence, an exploding cumulative privacy budget as iterations proceed.

In Algorithm~\ref{algorithm1}, the matrix $R_{i}$ is used to extract the $i$th local variable $x_i^{t}$ from the variable $\mathbf{x}_{i}^{t}$, i.e., $R_{i}\mathbf{x}_{i}^{t}=x_{i}^{t}$. In this case, line 4 enables agent $i$ to track the weighted aggregative function $\frac{\lambda_{y}^{t}}{m}\sum_{i=1}^{m}g_{i}^{t}(x_{i}^{t})$ and line 5 allows agent $i$ to track the weighted average gradient $\frac{\lambda_{z}^{t}}{m}\sum_{i=1}^{m}\nabla_{y}f_{i}^{t}(x_{i}^{t},\tilde{y}_{i}^{t})$. It is clear that we use robust tracking approaches in line 4 and line 5 to update $y_{i}^{t}$ and $z_{i}^{t}$, respectively, which is different from existing distributed aggregative optimization algorithms~\cite{lixiuxian1,personalized1,Carnevale1,yipeng,Carnevale2,wenguanghui2,wang2024momentum,cai2024distributed,mengmin,zhou2025distributed,zhang2025distributed,li2025distributed,huang2025distributed,lixiuxian2,wenguanghui1,yang2024class} which use conventional gradient tracking techniques. This modification is necessary because the injected Laplace noise will accumulate over time in conventional gradient-tracking-based algorithms. {\color{blue}More specifically, we use the classical gradient-tracking-based distributed aggregative optimization algorithm in~\cite{lixiuxian2} as an example to illustrate the idea. In the absence of DP noises, the algorithm in~\cite{lixiuxian2} is given as follows:
\begin{equation}
\begin{aligned}
x_{i}^{t+1}&\textstyle=\Pi_{\Omega_{i}}\Big\{x_{i}^{t}-\alpha_{t}\left(\nabla_{x}f_{i}^{t}(x_{i}^{t},y_{i}^{t})+\nabla g_{i}(x_{i}^{t})z_{i}^{t}\right)\Big\},\\
y_{i}^{t+1}&\textstyle=\sum_{j=1}^{m}a_{ij}y_{j}^{t}+g_{i}(x_{i}^{t+1})-g_{i}(x_{i}^{t}),\\
z_{i}^{t+1}&\textstyle=\sum_{j=1}^{m}a_{ij}z_{j}^{t}+\nabla_{y}f_{i}^{t+1}(x_{i}^{t+1},y_{i}^{t+1})-\nabla_{y}f_{i}^{t}(x_{i}^{t},y_{i}^{t}),\nonumber
\end{aligned}
\end{equation}
where $\sum_{j=1}^{m}a_{ij}=1$ and $\sum_{i=1}^{m}a_{ij}=1$ hold for all $i,j\in[m]$ according to Assumption 2 in~\cite{lixiuxian2} and the stepsize $\alpha_{t}$ is chosen as $\alpha_{t}=\frac{1}{\sqrt{t}}$ for all $t\geq 1$ following Eq. (8) in~\cite{lixiuxian2}.

Using initialization $y_{i}^{0}\!=\!g_{i}(x_{i}^{0})$ and $z_{i}^{0}\!=\!\nabla_{y}f_{i}^{0}(x_{i}^{0},y_{i}^{0})$, we obtain $\bar{y}^{t}=\frac{1}{m}\sum_{i=1}^{m}g_{i}(x_{i}^{t})$ and $\bar{z}^{t}=\frac{1}{m}\sum_{i=1}^{m}\nabla_{y}f_{i}^{t}(x_{i}^{t},y_{i}^{t})$,
which means that ensuring the consensus of $y_{i}^{t}$ and $z_{i}^{t}$, i.e., $y_{i}^{t}=\bar{y}^{t}$ and $z_{i}^{t}=\bar{z}^{t}$, is sufficient for each agent to accurately track the aggregative function and the global gradient, respectively, i.e., achieving $y_{i}^{t}=\bar{y}^{t}=\frac{1}{m}\sum_{i=1}^{m}g_{i}(x_{i}^{t})$ and $z_{i}^{t}=\bar{z}^{t}=\frac{1}{m}\sum_{i=1}^{m}\nabla_{y}f_{i}^{t}(x_{i}^{t},y_{i}^{t})$, respectively.

When noises are added to both shared variables $y_{i}^{t}$ and $z_{i}^{t}$ to ensure $\epsilon_{i}$-LDP, the algorithm in~\cite{lixiuxian2} becomes
\begin{equation}
\vspace{-0.2em}
\begin{aligned}
x_{i}^{t+1}&\textstyle\!\!=\!\Pi_{\Omega_{i}}\Big\{x_{i}^{t}-\alpha_{t}\left(\nabla_{x}f_{i}^{t}(x_{i}^{t},y_{i}^{t})+\nabla g_{i}(x_{i}^{t})z_{i}^{t}\right)\Big\},\\
y_{i}^{t+1}&\textstyle\!\!=\!\sum_{j=1}^{m}a_{ij}(y_{j}^{t}+\chi_{j}^{t})+g_{i}(x_{i}^{t+1})-g_{i}(x_{i}^{t}),\\
z_{i}^{t+1}&\textstyle\!\!=\!\sum_{j=1}^{m}a_{ij}(z_{j}^{t}\!+\!\zeta_{j}^{t})\!+\!\nabla_{y}f_{i}^{t+1}(x_{i}^{t+1}\!,y_{i}^{t+1})\!-\!\nabla_{y}f_{i}^{t}(x_{i}^{t},y_{i}^{t}),\nonumber
\end{aligned}
\end{equation}
where $\chi_{j}^{t}\in\mathbb{R}^{r}$ and $\zeta_{j}^{t}\in\mathbb{R}^{r}$ represent the DP noises injected on variables $y_{j}^{t}$ and $z_{j}^{t}$, respectively.

It can be seen that under the same initialization $y_{i}^{0}=g_{i}(x_{i}^{0})$ and $z_{i}^{0}=\nabla_{y}f_{i}^{0}(x_{i}^{0},y_{i}^{0})$, the following relations hold true by induction:
\begin{equation}
\vspace{-0.2em}
\begin{aligned}
\bar{y}^{t}&\textstyle=\frac{1}{m}\sum_{i=1}^{m}\left(g_{i}(x_{i}^{t})+\sum_{k=0}^{t-1}\chi_{i}^{k}\right);\\
\bar{z}^{t}&\textstyle=\frac{1}{m}\sum_{i=1}^{m}\left(\nabla_{y}f_{i}^{t}(x_{i}^{t},y_{i}^{t})+\sum_{k=0}^{t-1}\zeta_{i}^{k}\right),\label{accumu}
\end{aligned}
\vspace{-0.2em}
\end{equation}
which implies that the DP noises accumulate over time in the estimates of both the aggregative function and the global gradient. Therefore, when the variables $y_{i}^{t}$ and $z_{i}^{t}$ are directly fed into the update of the model parameter $x_{i}^{t+1}$, learning accuracy will be compromised.} In fact, our experimental results in Fig. \ref{mnista} and Fig.~\ref{cifara} have shown that existing distributed aggregative optimization algorithms in~\cite{Carnevale1} and~\cite{lixiuxian2} (both of which rely on conventional gradient tracking techniques) exhibit divergent behaviors under Laplace DP noise.

{\color{blue}In contrast, in Algorithm~\ref{algorithm1}, we use the increments $y_{i}^{t+1}-y_{i}^{t}$ and $z_{i}^{t+1}-z_{i}^{t}$ in the update of the model parameter $x_{i}^{t+1}$. Because the two increments can be verified to satisfy (by using Line 4 and Line 5 of Algorithm~\ref{algorithm1} and Assumption~\ref{A3}):
\begin{equation}
\vspace{-0.2em}
\begin{aligned}
\bar{y}^{t+1}-\bar{y}^{t}&\textstyle=\frac{\lambda_{y}^{t}}{m}\sum_{i=1}^{m}\left(g_{i}^{t}(R_{i}\mathbf{x}_{i}^{t})+\sum_{j\in\mathcal{N}_{i}}w_{ij}\chi_{j}^{t}\right);\\
\bar{z}^{t+1}-\bar{z}^{t}&\textstyle=\frac{\lambda_{z}^{t}}{m}\sum_{i=1}^{m}\left(\nabla_{y}f_{i}^{t}(R_{i}\mathbf{x}_{i}^{t},\tilde{y}_{i}^{t})+\sum_{j\in\mathcal{N}_{i}}w_{ij}\zeta_{j}^{t}\right),\nonumber
\end{aligned}
\end{equation}
this design effectively circumvents the noise accumulation problem discussed above in the existing algorithm of~\cite{lixiuxian2}.} In addition, it is worth noting that since we use robust tracking approaches, our algorithmic structure (i.e., two robust tracking steps followed by a dynamic average-consensus gradient descent) is fundamentally different from that of existing distributed aggregative optimization algorithms in \cite{lixiuxian1,personalized1,Carnevale1,yipeng,Carnevale2,wenguanghui2,wang2024momentum,cai2024distributed,mengmin,zhou2025distributed,zhang2025distributed,li2025distributed,huang2025distributed,lixiuxian2,wenguanghui1,yang2024class}, which makes it impossible to use proof techniques for existing results.

Another key reason for our algorithm to achieve both accurate convergence and rigorous $\epsilon_{i}$-LDP is the use of all data available at time $t$ to compute the gradient. This strategy can effectively reduce the variance in the gradient estimates, and hence, leads to accurate convergence to an optimal solution. This is in sharp contrast to the existing stochastic solution for distributed aggregative optimization in~\cite{lixiuxian2}, which cannot ensure convergence to an exact solution.
{\color{blue}\begin{remark}
In Algorithm~\ref{algorithm1}, although agent $i$ estimates the optimization variables of other agents, it cannot reconstruct the private data of others. Specifically, for any finite number of iterations, i.e., $T\!<\!\infty$, each agent $i$ does not have access to the exact optimization variables of others and only receives noise-corrupted variables from neighboring agents, which prevent accurate recovery of the exact optimization variables or inference of the private data of others. Furthermore, in the asymptotic case, i.e., $T\!=\!\infty$, each agent $i$ can indeed accurately know the optimal solutions of other agents. However, the optimal solution, i.e.,  $\mathbf{x}_{i}^{\infty}\!=\!\boldsymbol{x}^{*}$, is only a single point in the decision space $\mathbb{R}^{n}$ and is insufficient to reconstruct the objective functions or the underlying private data of other agents. In fact, revealing such a single point $\boldsymbol{x}^{*}$ corresponds to disclosing an event of measure zero with respect to the data-generating distribution, which does not affect the probabilities of outputs and therefore does not violate the $\epsilon_{i}$-LDP guarantees established based on output probabilities.
\end{remark}
\begin{remark}
\vspace{-0.5em}
Our algorithm can, in principle, be extended to directed graphs by leveraging the push-sum method~\cite{direct1,direct2,direct3}, which introduces normalization factors to compensate for network imbalances. Since the updates of these normalization factors are independent of the optimization variables, integrating the push-sum mechanism with Algorithm~\ref{algorithm1} is conceptually feasible. However, rigorously incorporating these normalization factors would require introducing additional communication variables and substantially more involved technical derivations, leading to a significantly lengthier convergence analysis and privacy analysis. We leave a systematic treatment of this extension to future work.
\end{remark}
\begin{remark}
\vspace{-0.5em}
Our work has many natural applications in practical control systems. For example, \cite{applicationDC} shows that in DC microgrid control systems, the objective of voltage regulation under current-sharing and economic operation requirements can be naturally formulated as a distributed aggregative optimization problem, thereby explicitly casting the control design task within an aggregative optimization framework (see Eq. (3) in \cite{applicationDC}). Similarly, \cite{huang2025distributed} shows that multi-agent formation control problems, e.g., the target surrounding problem,  
can be formulated as a distributed aggregative optimization problem, making Algorithm~\ref{algorithm1} applicable to these types of control tasks.
\end{remark}}
\vspace{-0.5em}
\section{Convergence Accuracy and Rate Analysis}\label{convergenceanalysis}
In this section, we systematically analyze the convergence rates of Algorithm~\ref{algorithm1} under nonconvex, general convex, and strongly convex global objective functions, respectively. 

For the sake of notational simplicity, we denote the lowest decaying rates of DP-noise variances among agents as $\varsigma_{x}\!=\!\min_{i\in[m]}\{\varsigma_{i,x}\}$, $\varsigma_{y}\!=\!\min_{i\in[m]}\{\varsigma_{i,y}\}$, and $\varsigma_{z}\!=\!\min_{i\in[m]}\{\varsigma_{i,z}\}$.

\begin{theorem}\label{T1}
	Under Assumptions~\ref{A1}-\ref{A4}, if the decaying rates of stepsizes satisfy $1>v_{x}>v_{z}$ and $\frac{1}{2}>v_{z}>v_{y}>0$, then for any $T\in\mathbb{N}^{+}$, we have the following results for Algorithm~\ref{algorithm1}:
	
	(i) if $F(\boldsymbol{x})$ is strongly convex with the rates of DP-noise variances satisfying $\varsigma_{x}>\max\{v_{y}-\varsigma_{y}, v_{z}-\varsigma_{z}, \frac{v_{x}}{2}\}$, $\varsigma_{y}>v_{y}-\frac{1}{2}$, and $\varsigma_{z}>v_{z}-\frac{1}{2}$, then we have 
	\begin{equation}
		\mathbb{E}[\|\mathbf{x}_{i}^{T}-\boldsymbol{x}^*\|^2]\leq \mathcal{O}(T^{-\beta}),\label{Theorem1results1}
	\end{equation}
	with $\beta=\min\{2\varsigma_{x}-v_{x},\frac{1}{2}-v_{y}+\varsigma_{y},\frac{1}{2}-v_{z}+\varsigma_{z}\};$
	
	(ii) if $F(\boldsymbol{x})$ is general convex with the rates of DP-noise variances satisfying $\varsigma_{x}>\frac{1}{2}$, $\varsigma_{y}>v_{y}-\frac{1}{2}+(1-v_{x})$, and $\varsigma_{z}>v_{z}-\frac{1}{2}+(1-v_{x})$, then we have 
	\begin{equation}
		\textstyle\frac{1}{T+1}\sum_{t=0}^{T}\mathbb{E}[F(\mathbf{x}_{i}^{t})-F(\boldsymbol{x}^{*})]\leq \mathcal{O}(T^{-(1-v_{x})});\label{Theorem1results2}
	\end{equation}
	
	(iii) if $F(\boldsymbol{x})$ is nonconvex with $\Omega=\mathbb{R}^{n}$ and the rates of DP-noise variances satisfy $\varsigma_{x}>\frac{1}{2}$, $\varsigma_{y}\!>\!v_{y}-\frac{1}{2}+(1-v_{x})$, and $\varsigma_{z}\!>\!v_{z}-\frac{1}{2}+(1-v_{x})$, then we have 
	\vspace{-0.2em}
	\begin{equation}
		\textstyle\frac{1}{T+1}\sum_{t=0}^{T}\mathbb{E}[\|\nabla F(\mathbf{x}_{i}^{t})\|^2]\leq\mathcal{O}(T^{-(1-v_{x})}).\label{Theorem1results3}
	\end{equation}
\end{theorem}
\begin{proof}
	See Appendix B.
	\vspace{-0.5em}
\end{proof}
Theorem~\ref{T1} proves that the optimization error of Algorithm~\ref{algorithm1} decreases over time at rates of $\mathcal{O}(T^{-\beta})$, $\mathcal{O}(T^{-(1-v_{x})})$, and $\mathcal{O}(T^{-(1-v_{x})})$, respectively, for strongly convex, general convex, and nonconvex $F(\boldsymbol{x})$. It is more comprehensive than existing distributed aggregative optimization results, in, e.g.,~\cite{lixiuxian1,yipeng,wenguanghui2,wang2024momentum,cai2024distributed,mengmin,zhou2025distributed,zhang2025distributed,li2025distributed,huang2025distributed,Carnevale1,personalized1,yang2024class,lixiuxian2,wenguanghui1}, which only consider strongly convex or convex objective functions. Moreover, it is also different from the only existing nonconvex result in~\cite{Carnevale2} for distributed aggregative optimization, which proves asymptotic convergence in the continuous-time domain and does not give explicit convergence rates.  

Next, we provide an example of parameter selection that satisfies the conditions given in Theorem~\ref{T1}:
\vspace{-0.5em}
\begin{corollary}\label{values}
Under the conditions in Theorem~\ref{T1}, we have the following results for Algorithm~\ref{algorithm1}:

(i) For a strongly convex $F(\boldsymbol{x})$ and an arbitrarily small $\delta>0$, if we set 
$v_{x}=\frac{1}{2}+7\delta$, $v_{z}=3\delta$, $v_{y}=2\delta$, $\varsigma_{x}=\frac{1}{2}+3\delta$, $\varsigma_{y}=\delta$, and $\varsigma_{z}=2\delta$, then we have $\beta=\frac{1}{2}-\delta$, which implies that the convergence rate of Algorithm~\ref{algorithm1} is arbitrarily close to $\mathcal{O}(T^{-0.5})$.

(ii) For a convex $F(\boldsymbol{x})$ and an arbitrarily small $\delta>0$, if we set 
$v_{x}=\frac{1}{2}+5\delta$, $v_{z}=3\delta$, $v_{y}=2\delta$, $\varsigma_{x}=\frac{1}{2}+\delta$, $\varsigma_{y}=\delta$, and $\varsigma_{z}=2\delta$, then we have $1-v_{x}=\frac{1}{2}-5\delta$, which implies that the convergence rate of Algorithm~\ref{algorithm1} is arbitrarily close to $\mathcal{O}(T^{-0.5})$.

(iii) For a nonconvex $F(\boldsymbol{x})$ and an arbitrarily small $\delta>0$, if we set 
$v_{x}=\frac{1}{2}+5\delta$, $v_{z}=3\delta$, $v_{y}=2\delta$, $\varsigma_{x}=\frac{1}{2}+\delta$, $\varsigma_{y}=\delta$, and $\varsigma_{z}=2\delta$, then we have $1-v_{x}=\frac{1}{2}-5\delta$, which implies that the convergence rate of Algorithm~\ref{algorithm1} is arbitrarily close to $\mathcal{O}(T^{-0.5})$.
\end{corollary}
\begin{proof}
The results follow naturally from Theorem~\ref{T1}.
\end{proof}
Corollary~\ref{values} proves that the convergence rate of our algorithm is arbitrarily close to $\mathcal{O}(T^{-0.5})$ for nonconvex/convex/strongly convex objective functions even in the presence of Laplace noise. This is consistent with the existing result in~\cite{wenguanghui1}, which established a convergence rate of $\mathcal{O}(T^{-0.5})$ for the special case of convex $F(\boldsymbol{x})$ and noise-free aggregative function $g(\boldsymbol{x})$. {\color{blue}Moreover, Corollary~\ref{values} provides a simple way to select parameters that satisfy the conditions in Theorem~\ref{T1}. Specifically, it implies that for any $\delta\in(0,\frac{1}{14})$, the parameters $v_{x}$, $v_{y}$, $v_{z}$, $\varsigma_{x}$, $\varsigma_{y}$, and $\varsigma_{z}$, chosen according to the specified rules, will satisfy the conditions in Theorem 1. Hence, in real-world applications, we do not need to tune six parameters $v_{x}$, $v_{y}$, $v_{z}$, $\varsigma_{x}$, $\varsigma_{y}$, and $\varsigma_{z}$, but instead only need to choose one parameter $\delta$.
	
\begin{corollary}\label{complexity}
	Under Assumption~\ref{A1}, for any $T\in\mathbb{N}^{+}$ and $i\in[m]$, the communication and computational complexities of Algorithm~\ref{algorithm1} are on the order of $\mathcal{O}((2r+n)d_{i}T)$ and $\mathcal{O}(n_{i}rT^2)$, respectively, where $d_{i}$ denotes the number of agent~$i$'s neighbors.
\end{corollary}
\begin{proof}
	See Appendix C.
\end{proof}
Compared with the conventional distributed aggregative optimization algorithms that do not consider privacy in, e.g.,~\cite{lixiuxian1,lixiuxian2}, where the dimension of exchanged messages is $2r$, Algorithm~\ref{algorithm1} requires exchanging messages of dimension $2r+n$. This increased communication overhead arises from the imposed LDP constraints. Specifically, the noise added to satisfy the LDP requirements inevitably reduces the amount of useful information conveyed in each shared variable. As a result, more variable exchanges are required to transmit sufficient information to ensure accurate convergence, compared with the conventional distributed aggregative optimization algorithms that ignore privacy. We regard this increased communication overhead as a necessary price for achieving accurate convergence under rigorous $\epsilon_{i}$-LDP.

In addition, we note that using historical data to compute gradients increases the computational complexity compared with the strategy of using one data-point per iteration. However, for continuous loss functions, our prior work~\cite[Section III-B]{LDPziji1} shows that, by employing Lagrange interpolation, the increased computational complexity at each iteration can be made independent of the iteration number. As a result, the computational complexity of Algorithm 1 over $T$ iterations can be reduced from $\mathcal{O}(n_{i}rT^2)$ to $\mathcal{O}(n_{i}rT)$, which is consistent with the order of computational complexity of noise-free distributed aggregative optimization algorithms in, e.g.,~\cite{lixiuxian2}.}

\section{Privacy analysis}\label{privacy}
In this section, we prove that besides accurate convergence, Algorithm~\ref{algorithm1} can simultaneously ensure rigorous $\epsilon_{i}$-LDP for each agent $i\!\in\![m]$, with a finite cumulative privacy budget even in the infinite time horizon. To this end, we first introduce the definition of sensitivity for agent $i$'s implementation at time $t$:

\begin{definition}\label{D3}
\vspace{-0.2em}
(Sensitivity~\cite{zhangjifeng1}) We denote the adjacency relationship between $\mathcal{D}_{i}^{t}$ and $\mathcal{D}_{i}^{\prime t}$ by $\text{Adj}(\mathcal{D}_{i}^{t},\mathcal{D}_{i}^{\prime t})$. The sensitivity of each agent $i$'s implementation $\mathcal{A}_{i}$ at time $t$ is defined as
\vspace{-0.1em}
\begin{equation}
\Delta_{i}^{t+1}=\max_{\text{Adj}(\mathcal{D}_{i}^{t},\mathcal{D}_{i}^{\prime t})}\|\mathcal{A}_{i}(\mathcal{D}_{i}^{t},\theta_{-i}^{t})-\mathcal{A}_{i}(\mathcal{D}_{i}^{\prime t},\theta_{-i}^{t})\|_{1},\nonumber
\vspace{-0.3em}
\end{equation}
where $\theta_{-i}^{t}$ represents all received information acquired by agent $i$ at time $t$.
\vspace{-0.2em}
\end{definition}
According to Definition~\ref{D3}, under Algorithm~\ref{algorithm1}, agent $i$'s implementation involves three sensitivities: $\Delta_{i,x}^{t}$, $\Delta_{i,y}^{t}$, and $\Delta_{i,z}^{t}$, which correspond to the three shared variables $\mathbf{x}_{i}^{t}+\boldsymbol{\vartheta}_{i}^{t}$, $y_{i}^{t}+\chi_{i}^{t}$, and
$z_{i}^{t}+\zeta_{i}^{t}$, respectively. With this understanding, we present the following lemma, which delineates a sufficient condition for $\epsilon_{i}$-LDP over any time horizon $T\in\mathbb{N}^{+}$:
\begin{lemma}[Lemma 2 in~\cite{huang}]\label{LDPlemma1}
For any given $T\in\mathbb{N}^{+}$ (which includes the case of $T\rightarrow\infty$), if agent $i$ injects to each of its shared variables $\mathbf{x}_{i}^{t}$,
$y_{i}^{t}$, and $z_{i}^{t}$ at each iteration $t\in\{1,\cdots, T\}$ noise vectors $\boldsymbol{\vartheta}_{i}^{t}$, $\chi_{i}^{t}$, and $\zeta_{i}^{t}$ consisting of $n$, $r$, and $r$ independent Laplace noises with parameters $\nu_{i,x}^{t}$, $\nu_{i,y}^{t}$, and $\nu_{i,z}^{t}$, respectively, then agent $i$'s implementation $\mathcal{A}_{i}$ is $\epsilon_{i}$-locally differentially private with the cumulative privacy budget from $t=1$ to $t=T$ upper bounded by $\sum_{t=1}^{T}(\frac{\Delta_{i,x}^{t}}{\nu_{i,x}^{t}}+\frac{\Delta_{i,y}^{t}}{\nu_{i,y}^{t}}+\frac{\Delta_{i,z}^{t}}{\nu_{i,z}^{t}})$.
\end{lemma}
For privacy analysis, we also need the following assumption, which is commonly used in DP design for server-assisted stochastic learning with an aggregative term~\cite{smooth2,smooth3}.
\begin{assumption}\label{A5}
The stochastic oracles 
$h(x,y;\varphi_{i})$, $l(x;\xi_{i})$, $\nabla h(x,y;\varphi_{i})$, and $\nabla l(x;\xi_{i})$ are $L_{h}$-, $L_{l}$-, $\bar{L}_{h}$-, and $\bar{L}_{l}$-Lipschitz continuous, respectively.
\end{assumption}

Without loss of generality, we consider adjacent datasets $\mathcal{D}_{i,f}^{t}$ (resp. $\mathcal{D}_{i,g}^{t}$) and $\mathcal{D}_{i,f}^{\prime t}$ (resp. $\mathcal{D}_{i,g}^{\prime t}$) that differ in the $k$th element. For the sake of clarity, the variables corresponding to $\mathcal{D}_{i,f}^{t}$ and $\mathcal{D}_{i,g}^{t}$ are denoted as $\mathbf{x}_{i}^{t}$, $y_{i}^{t}$, and $z_{i}^{t}$. Similarly, the variables corresponding to $\mathcal{D}_{i,f}^{\prime t}$ and $\mathcal{D}_{i,g}^{\prime t}$ are denoted as $\mathbf{x}_{i}^{\prime t}$, $y_{i}^{\prime t}$, and $z_{i}^{\prime t}$. 

\begin{theorem}\label{T4}
Under Assumptions~\ref{A2}-\ref{A5} and the conditions in Theorem~\ref{T1}, $\mathbf{x}_{i}^{t}$ (resp. $F(\mathbf{x}_{i}^{t})$ in the general convex case and $\nabla F(\mathbf{x}_{i}^{t})$ in the nonconvex case) in Algorithm~\ref{algorithm1} converges in mean-square to the optimal solution $\boldsymbol{x}^*$ to problem~\eqref{primal} (resp. in mean to $F(\boldsymbol{x}^*)$ and in mean-square to zero). Furthermore,

(i) for any $T\in\mathbb{N}^{+}$ (which includes the case of $T\rightarrow\infty$), agent $i$'s implementation of Algorithm~\ref{algorithm1} is locally differentially private with a cumulative privacy budget bounded by $\epsilon_{i}=\epsilon_{i,x}+\epsilon_{i,y}+\epsilon_{i,z}$, where $\epsilon_{i,x}$, $\epsilon_{i,y}$, and $\epsilon_{i,z}$ are given by $\epsilon_{i,x}\leq \sum_{t=1}^{T}\frac{\sqrt{2}C_{i,x}}{\sigma_{i,x}(t+1)^{1+v_{x}-v_{z}-\varsigma_{i,x}}}$, $\epsilon_{i,y}\leq\sum_{t=1}^{T}\frac{\sqrt{2}C_{i,y}}{\sigma_{i,y}(t+1)^{1+v_{y}-\varsigma_{i,y}}}$, and $\epsilon_{i,z}\leq\sum_{t=1}^{T}\frac{\sqrt{2}C_{i,z}}{\sigma_{i,z}(t+1)^{1+v_{z}-\varsigma_{i,z}}}$ with the positive constants $C_{i,x}$, $C_{i,y}$, and $C_{i,z}$ given in~\eqref{4T16},~\eqref{4T14}, and~\eqref{4T17}, respectively;

(ii) the cumulative privacy budget $\epsilon_{i}$ is finite even when the number of iterations $T$ tends to infinity.
\end{theorem}
\begin{proof}
The convergence results follow naturally from Theorem~\ref{T1}.

(i) To prove the statement on privacy, we first analyze the sensitivities of agent $i$'s implementation under Algorithm~\ref{algorithm1}. 

According to Definition~\ref{D3}, we have 
$\mathbf{x}_{j}^{t}+\boldsymbol{\vartheta}_{j}^{t}=\mathbf{x}_{j}^{\prime t}+\boldsymbol{\vartheta}_{j}^{\prime t}$,
$y_{j}^{t}+\chi_{j}^{t}=y_{j}^{\prime t}+\chi_{j}^{\prime t}$, and $z_{j}^{t}+\zeta_{j}^{t}=z_{j}^{\prime t}+\zeta_{j}^{\prime t}$ for all $t\in\mathbb{N}$ and $j\in{\mathcal{N}_{i}}$. Since we assume that only the $k$th data point is different between
$\mathcal{D}_{i,f}^{t}$ (resp. $\mathcal{D}_{i,g}^{t}$) and $\mathcal{D}_{i,f}^{\prime t}$ (resp. $\mathcal{D}_{i,g}^{\prime t}$), we have $\mathbf{x}_{i}^{p}=\mathbf{x}_{i}^{\prime p}$, $y_{i}^{p}=y_{i}^{\prime p}$, and $z_{i}^{p}=z_{i}^{\prime p}$ for all $p\leq k\leq t$. However, since the difference in loss functions kicks in at iteration $k$, i.e., $h(x,y;\varphi_{i}^{k})\neq h(x,y;\varphi_{i}^{\prime k})$ and $l(x;\xi_{i}^{k})\neq l(x;\xi_{i}^{\prime k})$, we have
$\mathbf{x}_{i}^{t}\neq \mathbf{x}_{i}^{\prime t}$, $y_{i}^{t}\neq y_{i}^{\prime t}$, and $z_{i}^{t}\neq z_{i}^{\prime t}$ for all $t>k$. Hence, for the dynamics of $y_{i}^{t}$ in Algorithm~\ref{algorithm1}, we have
\begin{equation}
\begin{aligned}
&\textstyle\|y_{i}^{t+1}-y_{i}^{\prime t+1}\|_{1}=\|(1+w_{ii})(y_{i}^{t}-y_{i}^{\prime t})\\
&\textstyle\quad+\lambda_{y}^{t}(g_{i}^{t}(R_{i}\mathbf{x}_{i}^{t})-g_{i}^{\prime t}(R_{i}\mathbf{x}_{i}^{\prime t}))\|_{1},\nonumber
\end{aligned}
\end{equation}
where we have used the definition $w_{ii}=-\sum_{j\in\mathcal{N}_{i}}w_{ij}$.

By using the relationship $\xi_{i}^{p}=\xi_{i}^{\prime p}$ for any $p\neq k$, the sensitivity $\Delta_{i,y}^{t+1}$ satisfies
\begin{equation}
\begin{aligned}
\textstyle\Delta_{i,y}^{t+1}&\textstyle\leq (1-|w_{ii}|)\Delta_{i,y}^{t}+\frac{\lambda_{y}^{t}}{t+1}\|l(R_{i}\mathbf{x}_{i}^{t};\xi_{i}^{k})-l(R_{i}\mathbf{x}_{i}^{\prime t};\xi_{i}^{\prime k})\|_{1}\nonumber\\
&\textstyle\quad+\frac{\sqrt{r}\lambda_{y}^{t}}{t+1}\sum_{p\neq k}^{t}\|l(R_{i}\mathbf{x}_{i}^{t};\xi_{i}^{p})-l(R_{i}\mathbf{x}_{i}^{\prime t};\xi_{i}^{p})\|_{2},\nonumber
\end{aligned}
\end{equation}
where we have used $\|x\|_{2}\leq \|x\|_{1}\leq \sqrt{r}\|x\|_{2}$ for any $x\in\mathbb{R}^{r}$.

Based on the fact $\|R_{i}\|=1$ and the $L_{l}$-Lipschitz continuity of $l(x;\xi_{i})$ from Assumption~\ref{A5}, we have
\begin{equation}
\textstyle\Delta_{i,y}^{t+1}\leq (1-|w_{ii}|)\Delta_{i,y}^{t}+L_{l}\sqrt{r}\lambda_{y}^{t}\Delta_{i,x}^{t}+\frac{2d_{i,l}\lambda_{y}^{t}}{t+1},\label{4T3}
\end{equation}
where we have used the fact that 
the optimization variable $\mathbf{x}_{i}^{t}$ is always constrained in the compact set $\Omega$, i.e., $\max_{\mathbf{x}_{i}\in\Omega}\{\|l(R_{i}\mathbf{x}_{i};\xi_{i})\|_{1}\}\leq d_{i,l}$ holds for some $d_{i,l}>0$.

Furthermore, we use the projection inequality from Lemma 1 in~\cite{projectionlemma} and the dynamics of $\mathbf{x}_{i}^{t}$ to obtain
\begin{flalign}
&\textstyle\|\mathbf{x}_{i}^{t+1}-\mathbf{x}_{i}^{\prime t+1}\|_{1}\leq (1-|w_{ii}|)\|\mathbf{x}_{i}^{t}-\mathbf{x}_{i}^{\prime t}\|_{1}\nonumber\\
&\textstyle+\frac{\lambda_{x}^{t}}{t+1}\sum_{p=0}^{t}\|\nabla_{x}h(R_{i}\mathbf{x}_{i}^{t},\tilde{y}_{i}^{t};\varphi_{i}^{p})-\nabla_{x}h(R_{i}\mathbf{x}_{i}^{\prime t},\tilde{y}_{i}^{\prime t};\varphi_{i}^{\prime p})\|_{1}\nonumber\\
&\textstyle+\frac{\lambda_{x}^{t}}{t+1}\sum_{p=0}^{t}\|\nabla l(R_{i}\mathbf{x}_{i}^{t};\xi_{i}^{p})\tilde{z}_{i}^{t}-\nabla l(R_{i}\mathbf{x}_{i}^{\prime t};\xi_{i}^{\prime p})\tilde{z}_{i}^{\prime t}\|_{1}.\label{4T5}
\end{flalign}
Since the relationship $\varphi_{i}^{p}=\varphi_{i}^{\prime p}$ holds for all $p\neq k$, the second term on the right hand side of~\eqref{4T5} can be decomposed as:
\begin{flalign}
&\textstyle\frac{\lambda_{x}^{t}}{t+1}\sum_{p=0}^{t}\|\nabla_{x}h(R_{i}\mathbf{x}_{i}^{t},\tilde{y}_{i}^{t};\varphi_{i}^{p})-\nabla_{x}h(R_{i}\mathbf{x}_{i}^{\prime t},\tilde{y}_{i}^{\prime t};\varphi_{i}^{\prime p})\|_{1}\nonumber\\
&\textstyle\leq\!\! \sqrt{n_{i}}\bar{L}_{h}\lambda_{x}^{t}\Delta_{i,x}^{t}\!+\!\frac{\sqrt{n_{i}}\bar{L}_{h}\lambda_{x}^{t}}{\lambda_{y}^{t}}(\Delta_{i,y}^{t+1}\!+\!\Delta_{i,y}^{t})\!+\!\frac{2\sqrt{n_{i}}L_{h}\lambda_{x}^{t}}{t+1},\label{4T6}
\end{flalign}
where in the derivation we have used the relationship $\|\tilde{y}_{i}^{t}-\tilde{y}_{i}^{\prime t}\|_{2}\leq \frac{1}{\lambda_{y}^{t}}(\|y_{i}^{t+1}-y_{i}^{\prime t+1}\|_{1}+\|y_{i}^{t}-y_{i}^{\prime t}\|_{1})$.

The third term on the right hand side of~\eqref{4T5} satisfies
\begin{equation}
\begin{aligned}
&\textstyle\frac{\lambda_{x}^{t}}{t+1}\sum_{p=0}^{t}\|\nabla l(R_{i}\mathbf{x}_{i}^{t};\xi_{i}^{p})\tilde{z}_{i}^{t}-\nabla l(R_{i}\mathbf{x}_{i}^{\prime t};\xi_{i}^{\prime p})\tilde{z}_{i}^{\prime t}\|_{1}\\
&\textstyle\leq \frac{\sqrt{n_{i}}\lambda_{x}^{t}}{t+1}\sum_{p=0}^{t}\|\nabla l(R_{i}\mathbf{x}_{i}^{t};\xi_{i}^{p})\|_{2}\|\tilde{z}_{i}^{t}-\tilde{z}_{i}^{\prime t}\|_{1}\\
&\textstyle\quad+\frac{\sqrt{n_{i}}\lambda_{x}^{t}}{t+1}\sum_{p\neq k}^{t}\|\nabla l(R_{i}\mathbf{x}_{i}^{t};\xi_{i}^{p})-\nabla l(R_{i}\mathbf{x}_{i}^{\prime t};\xi_{i}^{p})\|_{2}\|\tilde{z}_{i}^{\prime t}\|_{2}\\
&\textstyle\quad+\frac{\sqrt{n_{i}}\lambda_{x}^{t}}{t+1}\|\nabla l(R_{i}\mathbf{x}_{i}^{t};\xi_{i}^{k})-\nabla l(R_{i}\mathbf{x}_{i}^{\prime t};\xi_{i}^{\prime k})\|_{2}\|\tilde{z}_{i}^{\prime t}\|_{2}\\
&\textstyle\leq \frac{\sqrt{n_{i}}L_{l}\lambda_{x}^{t}}{\lambda_{z}^{t}}(\Delta_{i,z}^{t+1}+\Delta_{i,z}^{t})\!+\!\frac{\sqrt{n_{i}}\bar{L}_{l}d_{i,z}\lambda_{x}^{t}}{\lambda_{z}^{t}}\Delta_{i,x}^{t}\!+\!\frac{2\sqrt{n_{i}}L_{l}d_{i,z}\lambda_{x}^{t}}{\lambda_{z}^{t}(t+1)},\label{xxx}
\end{aligned}
\end{equation}
where we have used the fact that the convergence of Algorithm~\ref{algorithm1} ensures $\max_{t>0}\{\|z_{i}^{t}\|_{2}\}\leq d_{i,z}$ for some $d_{i,z}>0$.

Substituting~\eqref{4T6} and~\eqref{xxx} into~\eqref{4T5}, we arrive at
\begin{flalign}
&\textstyle\Delta_{i,x}^{t+1}\leq \left(1-|w_{ii}|+\sqrt{n_{i}}\bar{L}_{h}\lambda_{x}^{t}+\frac{\sqrt{n_{i}}\bar{L}_{l}d_{i,z}\lambda_{x}^{t}}{\lambda_{z}^{t}}\right)\Delta_{i,x}^{t}\nonumber\\
&\textstyle\quad+\frac{\sqrt{n_{i}}\bar{L}_{h}\lambda_{x}^{t}}{\lambda_{y}^{t}}(\Delta_{i,y}^{t+1}+\Delta_{i,y}^{t})+\frac{\sqrt{n_{i}}L_{l}\lambda_{x}^{t}}{\lambda_{z}^{t}}(\Delta_{i,z}^{t+1}+\Delta_{i,z}^{t})\nonumber\\
&\textstyle\quad+\frac{2\sqrt{n_{i}}L_{h}\lambda_{x}^{t}}{t+1}+\frac{2\sqrt{n_{i}}d_{i,z}L_{l}\lambda_{x}^{t}}{\lambda_{z}^{t}(t+1)}.\label{4T9}
\end{flalign}

Next, we characterize the sensitivity $\Delta_{i,z}^{t}$. According to the dynamics of $z_{i}^{t}$, we have
\begin{flalign}
&\|z_{i}^{t+1}-z_{i}^{\prime t+1}\|_{1}\leq(1-|w_{ii}|)\|z_{i}^{t}-z_{i}^{\prime t}\|_{1}\nonumber\\
&\textstyle+\frac{\sqrt{r}\lambda_{z}^{t}}{t+1}\sum_{p\neq k}^{t}\|\nabla_{y}h(R_{i}\mathbf{x}_{i}^{t},\tilde{y}_{i}^{t};\varphi_{i}^{p})\!-\!\nabla_{y}h(R_{i}\mathbf{x}_{i}^{\prime t},\tilde{y}_{i}^{\prime t};\varphi_{i}^{p})\|_{2}\nonumber\\
&\textstyle+\frac{\sqrt{r}\lambda_{z}^{t}}{t+1}\|\nabla_{y}h(R_{i}\mathbf{x}_{i}^{t},\tilde{y}_{i}^{t};\varphi_{i}^{k})\!-\!\nabla_{y}h(R_{i}\mathbf{x}_{i}^{\prime t},\tilde{y}_{i}^{\prime t};\varphi_{i}^{\prime k})\|_{2}.\nonumber
\end{flalign}
Using an argument similar to the derivation of~\eqref{4T6} yields
\begin{equation}
\begin{aligned}
&\Delta_{i,z}^{t+1}\textstyle\leq (1-|w_{ii}|)\Delta_{i,z}^{t}+\sqrt{r}\bar{L}_{h}\lambda_{z}^{t}\Delta_{i,x}^{t}\\
&\textstyle\quad+\frac{\sqrt{r}\bar{L}_{h}\lambda_{z}^{t}}{\lambda_{y}^{t}}(\Delta_{i,y}^{t+1}+\Delta_{i,y}^{t})+\frac{2\sqrt{r}L_{h}\lambda_{z}^{t}}{t+1}.\label{4T11}
\end{aligned}
\end{equation}

Summing both sides of~\eqref{4T3},~\eqref{4T9}, and~\eqref{4T11}, we obtain
\begin{flalign}
&\textstyle\Delta_{i,x}^{t+1}\!+\!\left(1\!-\!\frac{\sqrt{n_{i}}\bar{L}_{h}\lambda_{x}^{t}}{\lambda_{y}^{t}}\!-\!\frac{\sqrt{r}\bar{L}_{h}\lambda_{z}^{t}}{\lambda_{y}^{t}}\right)\Delta_{i,y}^{t+1}\!+\!\left(1\!-\!\frac{\sqrt{n_{i}}L_{l}\lambda_{x}^{t}}{\lambda_{z}^{t}}\right)\Delta_{i,z}^{t+1}\nonumber\\
&\textstyle\leq \left(1\!-\!\bar{w}\!+\!\sqrt{n_{i}}\bar{L}_{h}\lambda_{x}^{t}+\frac{\sqrt{n_{i}}\bar{L}_{l}d_{i,z}\lambda_{x}^{t}}{\lambda_{z}^{t}}\!+\!\sqrt{r}(L_{l}\lambda_{y}^{t}\!+\!\bar{L}_{h}\lambda_{z}^{t})\right)\Delta_{i,x}^{t}\nonumber\\
&\textstyle+\left(1\!-\!\bar{w}\!+\!\frac{\sqrt{n_{i}}\bar{L}_{h}\lambda_{x}^{t}}{\lambda_{y}^{t}}\!+\!\frac{\sqrt{r}\bar{L}_{h}\lambda_{z}^{t}}{\lambda_{y}^{t}}\right)\!\!\Delta_{i,y}^{t}\!+\!\left(1\!-\!\bar{w}\!+\!\frac{\sqrt{n_{i}}L_{l}\lambda_{x}^{t}}{\lambda_{z}^{t}}\right)\!\Delta_{i,z}^{t}\nonumber\\
&\textstyle+\frac{2d_{i,l}\lambda_{y}^{t}}{t+1}+\frac{2\sqrt{n_{i}}L_{h}\lambda_{x}^{t}}{t+1}+\frac{2\sqrt{n_{i}}d_{i,z}L_{l}\lambda_{x}^{t}}{\lambda_{z}^{t}(t+1)}+\frac{2\sqrt{r}L_{h}\lambda_{z}^{t}}{t+1}.
\end{flalign}
Since the decaying rates of stepsizes $\lambda_{x}^{t}$, $\lambda_{y}^{t}$, and $\lambda_{z}^{t}$ satisfy $1\!>\!v_{x}\!>\!v_{z}\!>\!v_{y}\!>\!0$, we can choose proper initial stepsizes such that the following inequality always holds:
\begin{flalign}
\textstyle\Delta_{i,x}^{t+1}+\Delta_{i,y}^{t+1}+\Delta_{i,z}^{t+1}&\textstyle\leq \big(1-\frac{|w_{ii}|}{2}\big)\left(\Delta_{i,x}^{t}+\Delta_{i,y}^{t}+\Delta_{i,z}^{t}\right)\nonumber\\
&\textstyle\quad+\frac{C_{i,0}}{(t+1)^{\min\{1+v_x-v_{z},1+v_{y}\}}},\label{tt1}
\end{flalign}
with $C_{i,0}\!=\!2(d_{i,l}\lambda_{y}^{0}+\sqrt{n_{i}}\lambda_{x}^{0}(L_{h}+d_{i,z}L_{l}(\lambda_{z}^{0})^{-1})+\sqrt{r}L_{h}\lambda_{z}^{0})$.

By applying Lemma 11 in~\cite{zijiGT} to~\eqref{tt1} ({\color{blue}here, in~\eqref{tt1}, $\Delta_{i,x}^{t}+\Delta_{i,y}^{t}+\Delta_{i,z}^{t}$ corresponds to $\psi_{t}$ in Lemma 11 in~\cite{zijiGT}}) and using the relation $\Delta_{i,x}^{0}\!=\!\Delta_{i,y}^{0}\!=\!\Delta_{i,z}^{0}\!=\!0$, we have
\begin{equation}
\textstyle\Delta_{i,x}^{t}+\Delta_{i,y}^{t}+\Delta_{i,z}^{t}\leq \frac{C_{i,1}}{(t+1)^{\min\{1+v_x-v_{z},1+v_{y}\}}},\label{4T121}
\vspace{-0.2em}
\end{equation}
with $C_{i,1}\!=\!\frac{4C_{i,0}}{|w_{ii}|}(\frac{4\min\{1+v_{x}-v_{z},1+v_{y}\}}{e\ln(\frac{4}{4-|w_{ii}|})})^{\min\{1+v_{x}-v_{z},1+v_{y}\}}$.

Substituting $\Delta_{i,x}^{t}\leq C_{i,1}(t+1)^{-\min\{1+v_x-v_{z},1+v_{y}\}}$ into~\eqref{4T3} and using Lemma 11 in~\cite{zijiGT}, we obtain
\begin{equation}
\textstyle \Delta_{i,y}^{t}\leq \frac{C_{i,y}}{(t+1)^{1+v_{y}}}~\text{with}~{\color{blue}C_{i,y}=\frac{8C_{i,2}}{|w_{ii}|}(\frac{1+v_{y}}{e\ln(\frac{2}{2-|w_{ii}|})})^{1+v_{y}}},\label{4T14}
\vspace{-0.2em} 
\end{equation}
where $C_{i,2}=(C_{i,1}L_{l}\sqrt{r}+2d_{i,l})\lambda_{y}^{0}$ and $C_{i,1}$ is given in~\eqref{4T121}.

Substituting $\Delta_{i,z}^{t}\!\leq\! C_{i,1}(t+1)^{-\min\{1+v_x-v_{z},1+v_{y}\}}$ and \eqref{4T14} into~\eqref{4T9}, we have
\begin{equation}
\textstyle\Delta_{i,x}^{t+1}\leq\big(1-\frac{|w_{ii}|}{2}\big)\Delta_{i,x}^{t}+\frac{C_{i,3}}{(t+1)^{1+v_{x}-v_{z}}},\label{tt2}
\vspace{-0.2em}
\end{equation}
with $C_{i,3}=\frac{2C_{i,y}\sqrt{n_{i}}\bar{L}_{h}\lambda_{x}^{0}}{\lambda_{y}^{0}}+\frac{2(C_{i,1}+d_{i,z})\sqrt{n_{i}}L_{l}\lambda_{x}^{0}}{\lambda_{z}^{0}}+2\sqrt{n_{i}}L_{h}\lambda_{x}^{0}$. 

By applying Lemma 11 in~\cite{zijiGT} to~\eqref{tt2}, we obtain
\begin{equation}
\vspace{-0.2em}
\textstyle\Delta_{i,x}^{t}\!\leq\!\! \frac{C_{i,x}}{(t+1)^{1+v_{x}-v_{z}}}~\text{with}~{\color{blue}C_{i,x}\!=\!\frac{16C_{i,3}}{|w_{ii}|}\!(\frac{1+v_{x}-v_{z}}{e\ln(\frac{4}{4-|w_{ii}|})}\!)^{1+v_{x}-v_{z}}}\!.\label{4T16}
\vspace{-0.2em}
\end{equation}

Similarly, substituting~\eqref{4T14} and~\eqref{4T16} into~\eqref{4T11} and using Lemma 11 in~\cite{zijiGT}, we have
\begin{equation}
\vspace{-0.2em}
\textstyle \Delta_{i,z}^{t}\leq \frac{C_{i,z}}{(t+1)^{1+v_{z}}}~\text{with}~{\color{blue}C_{i,z}=\frac{8C_{i,4}}{|w_{ii}|}(\frac{1+v_{z}}{e\ln(\frac{2}{2-|w_{ii}|})})^{1+v_{z}}},\label{4T17}
\end{equation}
with $C_{i,4}=C_{i,x}\sqrt{r}\bar{L}_{h}\lambda_{z}^{0}+2C_{i,y}\sqrt{r}\bar{L}_{h}\lambda_{z}^{0}(\lambda_{y}^{0})^{-1}+2\sqrt{r}L_{h}\lambda_{z}^{0}$.

Combining~\eqref{4T14},~\eqref{4T16}, and~\eqref{4T17} with Lemma~\ref{LDPlemma1}, we arrive at
\begin{flalign}
&\textstyle\sum_{t=1}^{T}\left(\frac{\Delta_{i,x}^{t}}{\nu_{i,x}}+\frac{\Delta_{i,y}^{t}}{\nu_{i,y}}+\frac{\Delta_{i,z}^{t}}{\nu_{i,z}}\right)\leq \sum_{t=1}^{T}\left(\frac{\sqrt{2}C_{i,x}}{\sigma_{i,x}(t+1)^{1+v_{x}-v_{z}-\varsigma_{i,x}}}\right.\nonumber\\
&\left.\textstyle\quad+\frac{\sqrt{2}C_{i,y}}{\sigma_{i,y}(t+1)^{1+v_{y}-\varsigma_{i,y}}}+\frac{\sqrt{2}C_{i,z}}{\sigma_{i,z}(t+1)^{1+v_{z}-\varsigma_{i,z}}}\right),\label{4T19}
\end{flalign}
where $\nu_{i,x}$, $\nu_{i,y}$, and $\nu_{i,z}$ represent Laplace-noise variances, which are given in Assumption~\ref{A4}.

(ii) Inequality~\eqref{4T19} implies that the cumulative privacy budget $\epsilon_{i}$ is finite even when $T\rightarrow\infty$ since $v_{x}-v_{z}>\varsigma_{i,x}$, $v_{y}>\varsigma_{i,y}$, and $v_{z}>\varsigma_{i,z}$ always hold based on Assumption~\ref{A4}.
\vspace{-0.5em}
\end{proof}
{\color{blue}Note that for any given $\epsilon_{i}>0$, the required initial noise of agent $i$ depends on its weighted degree $|w_{ii}|\!=\sum_{j\in\mathcal{N}_{i}}w_{ij}$. Specifically, according to our analysis in~\eqref{TC2}, for any given $\epsilon_{i}\!>\!0$, the initial DP-noises of agent $i$ can be chosen as $\sigma_{i,x}\!=\!\frac{3\sqrt{2}C_{x}}{(v_{x}-v_{z}-\varsigma_{i,x})\epsilon_{i}}$, $\sigma_{i,y}\!=\!\frac{3\sqrt{2}C_{y}}{(v_{y}-\varsigma_{i,y})\epsilon_{i}}$,~and $\sigma_{i,z}\!=\!\frac{3\sqrt{2}C_{z}}{(v_{z}-\varsigma_{i,z})\epsilon_{i}}$, where $v_{x}$, $v_{y}$, $v_{z}$, $\varsigma_{i,x}$, $\varsigma_{i,y}$, and $\varsigma_{i,z}$ are preset parameters based on Corollary~\ref{values}, and $C_{i,x}$, $C_{i,y}$, and $C_{i,z}$ are given in~\eqref{4T16},~\eqref{4T14}, and~\eqref{4T17}, respectively. According to~\eqref{4T16},~\eqref{4T14}, and~\eqref{4T17}, the constants $C_{i,x}$, $C_{i,y}$, and $C_{i,z}$ are inversely proportional to the weighted degree $|w_{ii}|$. Consequently, a smaller weighted degree implies that agent $i$ requires larger initial noise to ensure a desired cumulative privacy budget $\epsilon_{i}$.}

We would like to point out that although ensuring accurate convergence, the achievement of rigorous $\epsilon_{i}$-LDP of Algorithm~\ref{algorithm1} does pay a price in convergence rates. We quantify the tradeoff in the following corollary:
\begin{corollary}\label{tradeoff}
For any given cumulative privacy budget $\epsilon_{i}>0$ and $i\in[m]$, the convergence rates of Algorithm~\ref{algorithm1} are $\mathcal{O}\left(\frac{T^{-\beta}}{\min_{i\in[m]}\{\epsilon_{i}^2\}}\right)$ for strongly convex $F(\boldsymbol{x})$ and $\mathcal{O}\left(\frac{T^{-(1-v_{x})}}{\min_{i\in[m]}\{\epsilon_{i}^2\}}\right)$ for general convex or nonconvex $F(\boldsymbol{x})$.
\end{corollary}
\begin{proof}
See Appendix D.
\end{proof}
From Corollary~\ref{tradeoff}, it is clear that a higher level of local differential privacy, i.e., a smaller cumulative privacy budget $\epsilon_{i}$, leads to a reduced convergence rate.
\begin{figure*}
\centering
\subfigure[Comparison of different algorithms]{\label{mnista}
\includegraphics[width=0.24\linewidth]{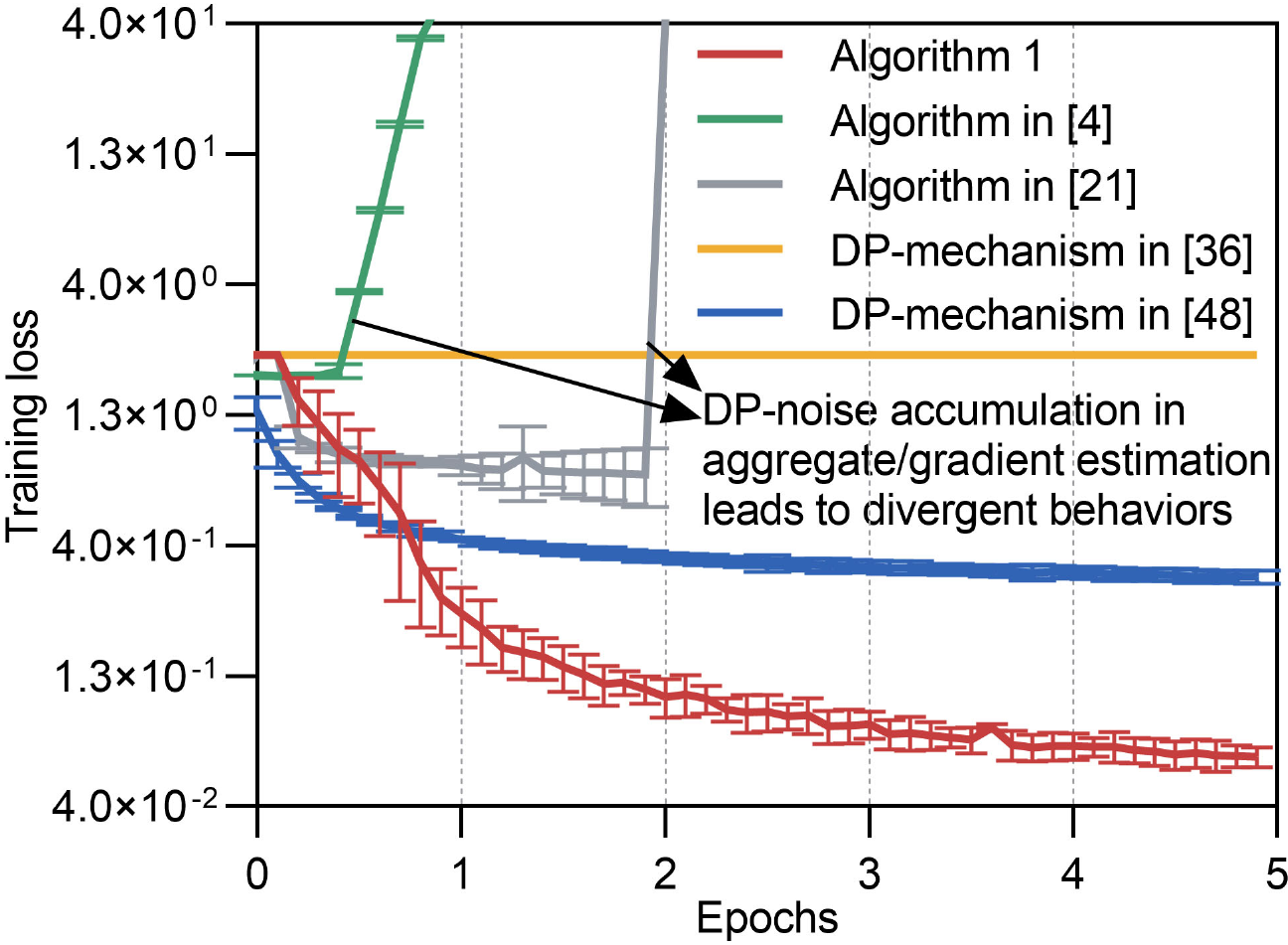}}\!
\subfigure[Cumulative privacy budgets]{\label{mnistb}
\includegraphics[width=0.23\linewidth]{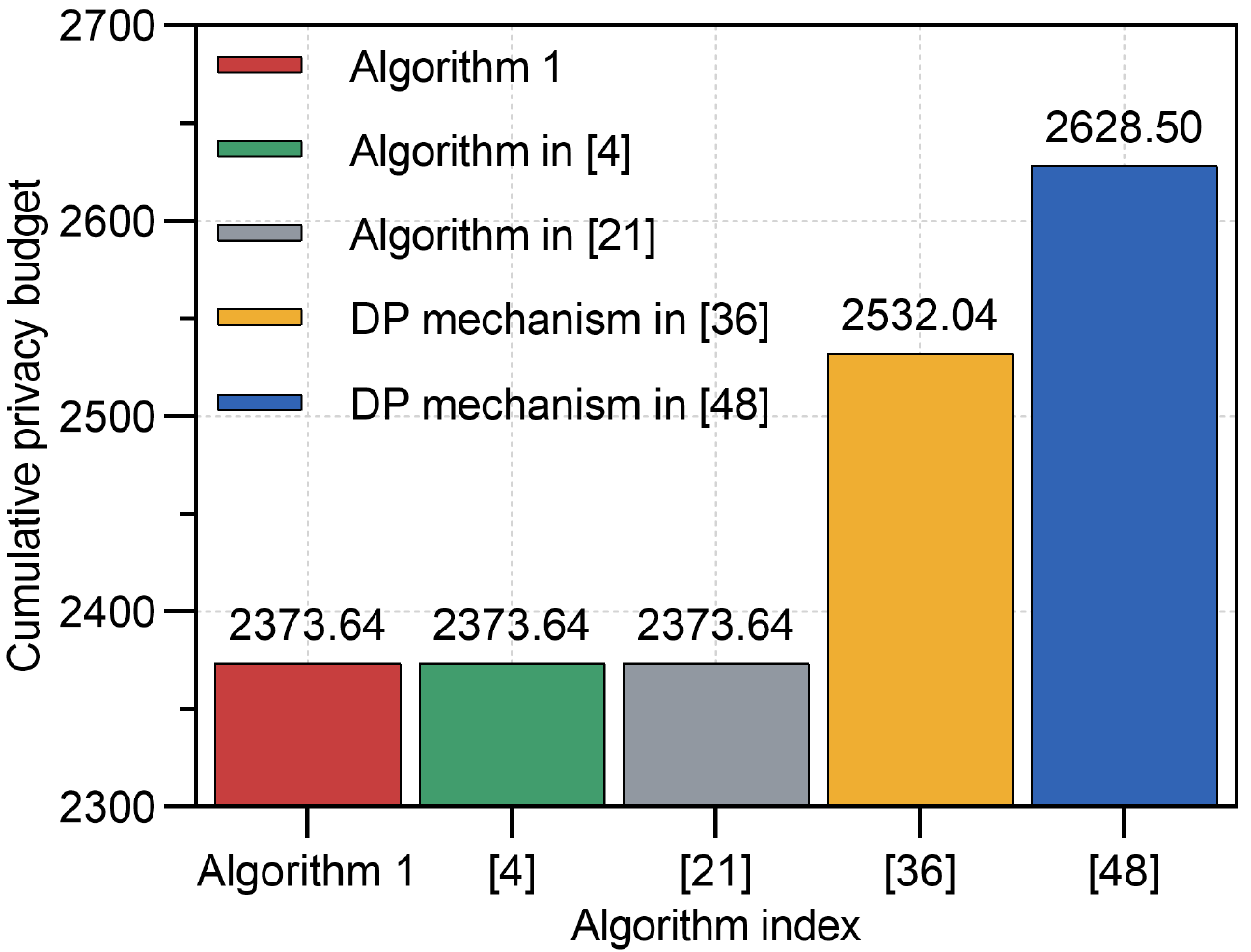}}\!
\subfigure[{\color{blue}Comparison under different $\delta$}]{\label{mnistc}
\includegraphics[width=0.24\linewidth]{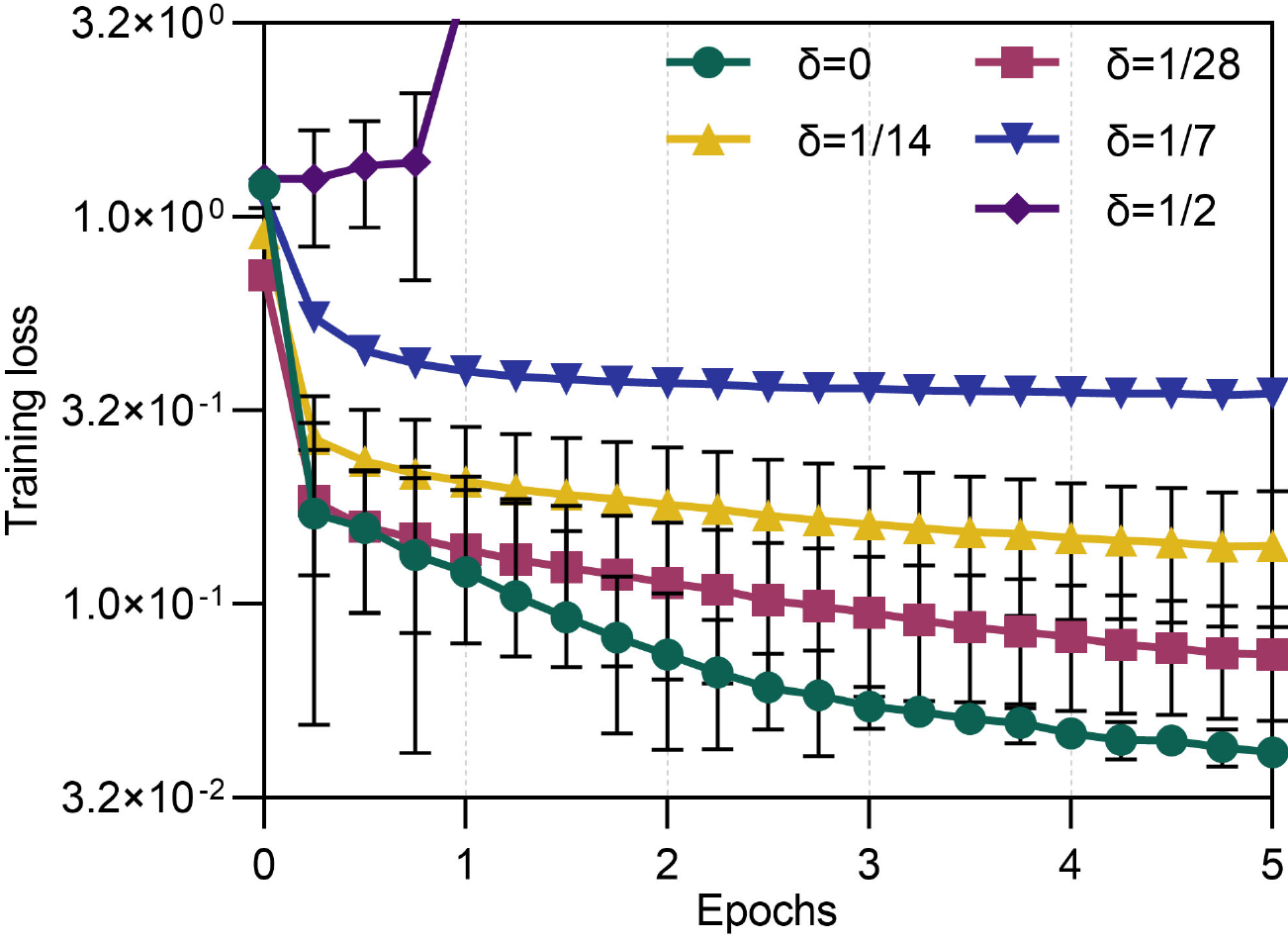}}\!
\subfigure[{\color{blue}Comparison under different $\lambda$}]{\label{mnistd}
\includegraphics[width=0.24\linewidth]{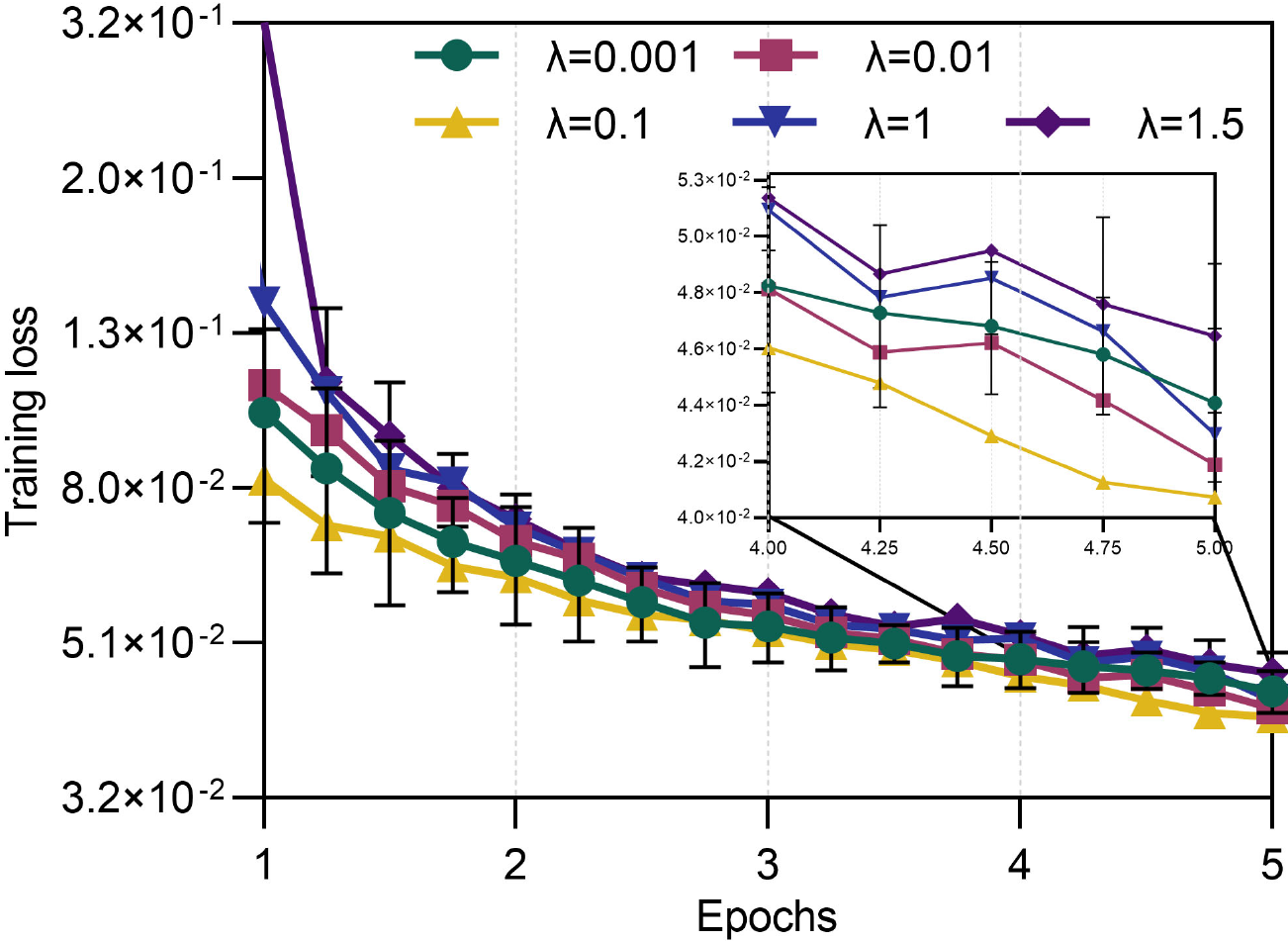}}
\caption{{\color{blue}Experimental results on CNN training using the ``MNIST" dataset.} (a) and (b) Training loss and cumulative privacy budget comparison of Algorithm~\ref{algorithm1} with existing distributed aggregative optimization algorithms in~\cite{Carnevale1} and~\cite{lixiuxian2} under the same cumulative privacy budget, and existing DP approaches in~\cite{huang}  and~\cite{zhangjifeng1} for distributed optimization. {\color{blue}(c) Training losses of Algorithm~\ref{algorithm1} under different $\delta$ ($\delta$ is given in Corollary~\ref{values}). (d) Training losses of Algorithm~\ref{algorithm1} under different $\lambda$ ($\lambda$ is defined in~\eqref{experimentloss}).} The error bars represent the standard deviation.}
\label{mnist}
\vspace{-0.8em}
\end{figure*}
\section{Numerical Experiments}\label{experiment}
\subsection{Experimental setups for personalized machine learning}
We considered the applications of Algorithm~\ref{algorithm1} in distributed personalized machine learning using the ``MNIST" and ``CIFAR-10" datasets, respectively. Following~\cite{personalized2,personalized3}, the personalized machine learning problem can be formulated as follows:
\begin{flalign}
	\vspace{-0.2em}
	\min_{\boldsymbol{x}\in\Omega}~&\textstyle \sum_{i=1}^{m}F_{i}^{t}(x_{i},G^{t}(\boldsymbol{x})),~G^{t}(\boldsymbol{x})=\frac{1}{m}\sum_{i=1}^{m}G_{i}^{t}(x_{i}),\nonumber\\
	\text{where}~&\textstyle ~~F_{i}^{t}(x_{i},G^{t}(\boldsymbol{x}))=\frac{1}{|\mathcal{D}_{i}^{t}|}\sum_{j=1}^{|\mathcal{D}_{i}^{t}|}h(x_{i},G^{t}(\boldsymbol{x});\xi_{j}),\nonumber\\
	&\textstyle h(x_{i},G^{t}(\boldsymbol{x});\xi_{j})=\mathcal{L}(x_{i};\xi_{j})\!+\!\lambda\|\mathcal{L}(x_{i};\xi_{j})\!-\!G^{t}(\boldsymbol{x})\|^2,\nonumber\\
	&\textstyle\quad\quad\quad~~ G_{i}^{t}(x_{i})=\frac{1}{|\mathcal{D}_{i}^{t}|}\sum_{j=1}^{|\mathcal{D}_{i}^{t}|}\mathcal{L}(x_{i};\xi_{j}).\label{experimentloss}
	\vspace{-0.2em}
\end{flalign}
Here, $\lambda\geq 0$ is a penalty parameter, $x_{i}\in \Omega_{i}$ is the local model of agent $i$, $\boldsymbol{x}=\col(x_{1},\cdots,x_{m})\in\Omega$ denotes the collection of local models from all agents, and $\mathcal{L}(\cdot)$ denotes the cross entropy loss. Moreover, $\mathcal{D}_{i}^{t}$ represents the training dataset available to agent $i$ at iteration $t$ with $|\mathcal{D}_{i}^{t}|$ denoting its size, and 
the data point $\xi_{j}$ is randomly sampled from $\mathcal{D}_{i}^{t}$. The goal of personalized machine learning is to obtain a collection of optimal local models $\boldsymbol{x}^*=\col(x_{1}^{*},\cdots,x_{m}^{*})$, where each $x_i^*$ is optimized to balance fidelity to agent $i$'s local dataset and the aggregation of information from all agents~\cite{personalized2,personalized3}.

According to~\eqref{experimentloss}, it can be seen that when $\lambda=0$, problem~\eqref{experimentloss} reduces to $\min_{\boldsymbol{x}\in\Omega} \sum_{i=1}^{m}F_{i}^{t}(x_{i})$ with $F_{i}^{t}(x_{i})=\frac{1}{|\mathcal{D}_{i}^{t}|}\sum_{j=1}^{|\mathcal{D}_{i}^{t}|}\mathcal{L}(x_{i};\xi_{j})$. In this case, each agent $i$ independently learns an optimal local model $x_{i}^{*}$ based solely on its own dataset $\mathcal{D}_{i}$. However, the local dataset $\mathcal{D}_{i}$ is typically not rich enough for the resulting model to be useful~\cite{personalized2}. In order to learn a better model, it is necessary to take into account information from other agents. This is achieved by introducing a penalty term $\lambda\|\cdot\|$ in~\eqref{experimentloss}, which promotes collaboration and enhances the consistency of local models across agents~\cite{personalized2}.

In each experiment, we conducted two types of comparisons: 1) we compared our algorithm with the state-of-the-art distributed aggregative optimization algorithms in~\cite{Carnevale1} and~\cite{lixiuxian2} under the same cumulative privacy budget; and 2) to compare with existing DP approaches for distributed optimization, we also implemented the DP approach PDOP in~\cite{huang} (which uses geometrically
decreasing stepsizes and DP-noise variances to ensure a finite cumulative
privacy budget) and the DP approach in~\cite{zhangjifeng1} (which uses time-varying sample sizes to ensure a finite cumulative privacy budget) within our algorithmic framework. {\color{blue}Note that the first type of comparison evaluates the performance of our novel distributed aggregative optimization algorithm structure, whereas the second type of comparison assesses the performance of our proposed LDP mechanism.} For both ``MNIST" and ``CIFAR-10" datasets, we considered heterogeneous data distributions, where $60\%$ of the data from the $i$th and the $(2i)$th classes was assigned to agent $i$, while the remaining $40\%$ was evenly distributed among the other agents. {\color{blue}The parameter $\lambda$ in~\eqref{experimentloss} was set to $\lambda\!=\!0.5$ for the ``MNIST" experiment and to $\lambda\!=\!0.1$ for the ``CIFAR-10" experiment.} The projection set in Algorithm~\ref{algorithm1} and the comparison algorithms was set to $\Omega_{i}=\{x\in\mathbb{R}^{n}|10^{-6}\!\leq\! x_{j}\!\leq\! 10^{6}, \forall j\!=\!1,\cdots,n\}$ (e.g., with $n=784$ in the ``MNIST" experiment). The interaction
pattern contains $5$ agents connected in a circle, where each
agent can only communicate with its two immediate neighbors. For the weight matrix $W$, we set $w_{ij}\!=\!0.3$ if agents $i$ and $j$ are neighbors, and $w_{ij}\!=\!0$ otherwise. 
\begin{figure*}
	\centering
	\subfigure[Comparison of different algorithms]{\label{cifara}
		\includegraphics[width=0.24\linewidth]{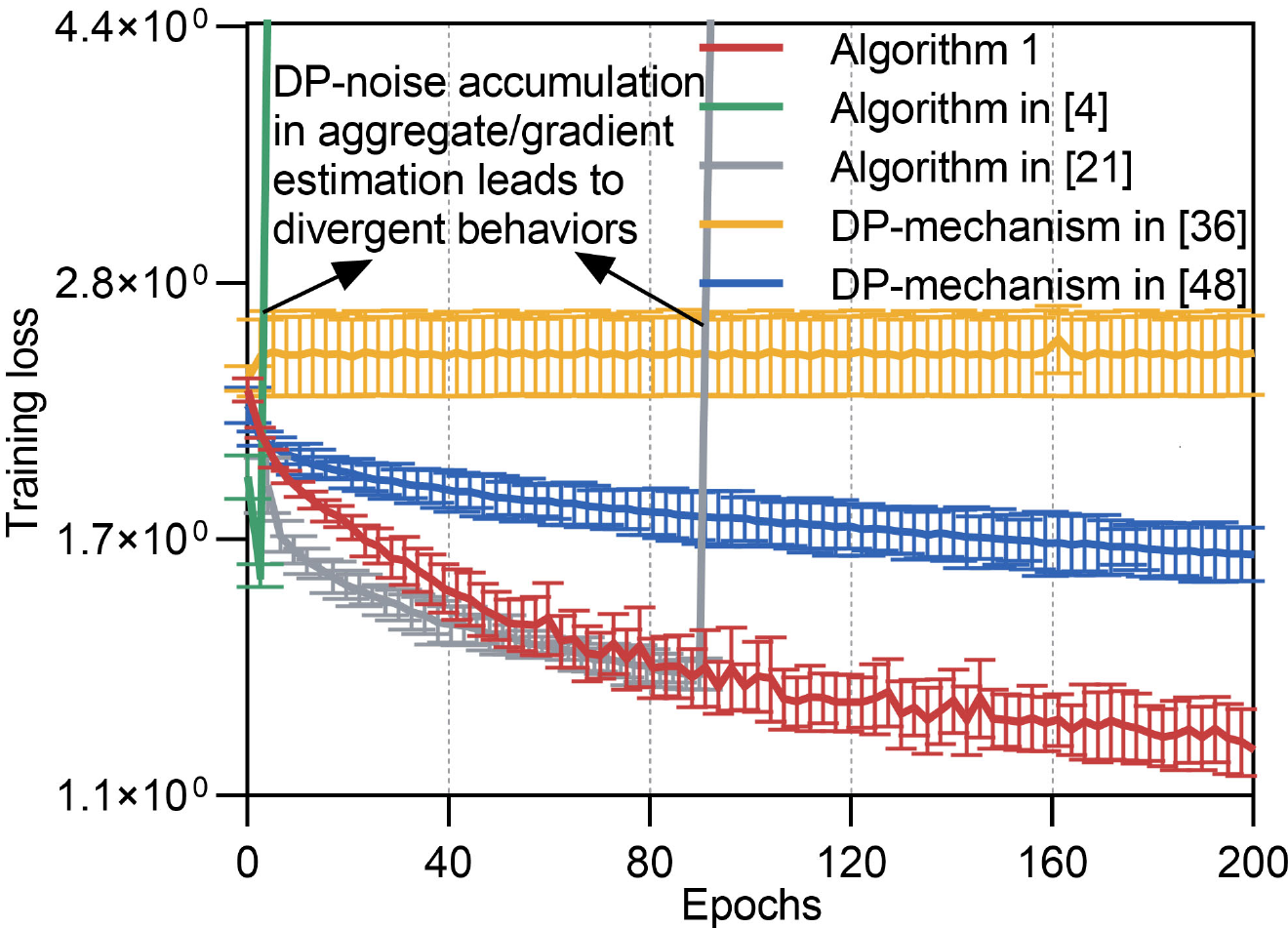}}\!
	\subfigure[Cumulative privacy budgets]{\label{cifarb}
		\includegraphics[width=0.23\linewidth]{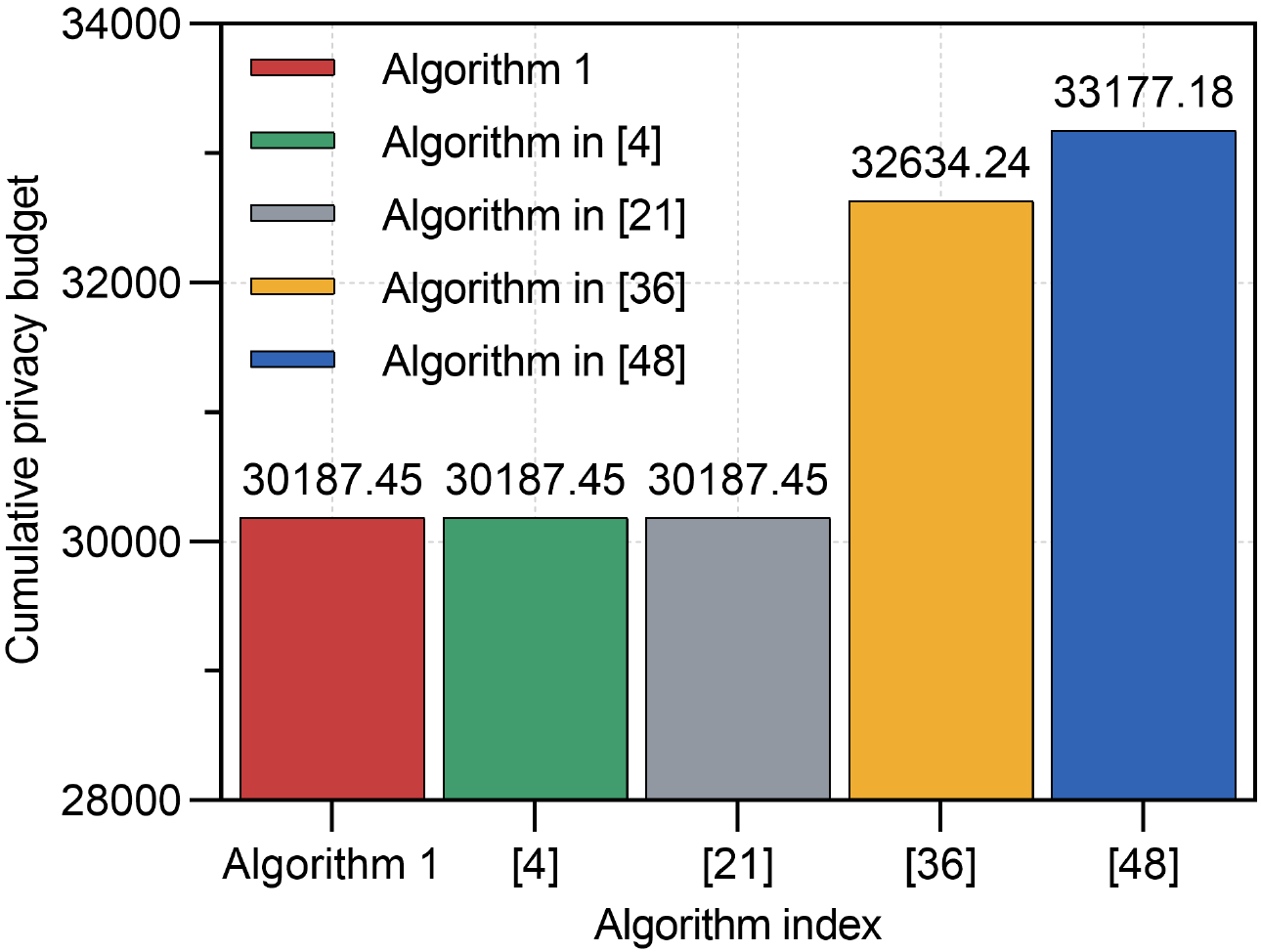}}\!
	\subfigure[{\color{blue}Comparison under different $\delta$}]{\label{cifarc}
		\includegraphics[width=0.24\linewidth]{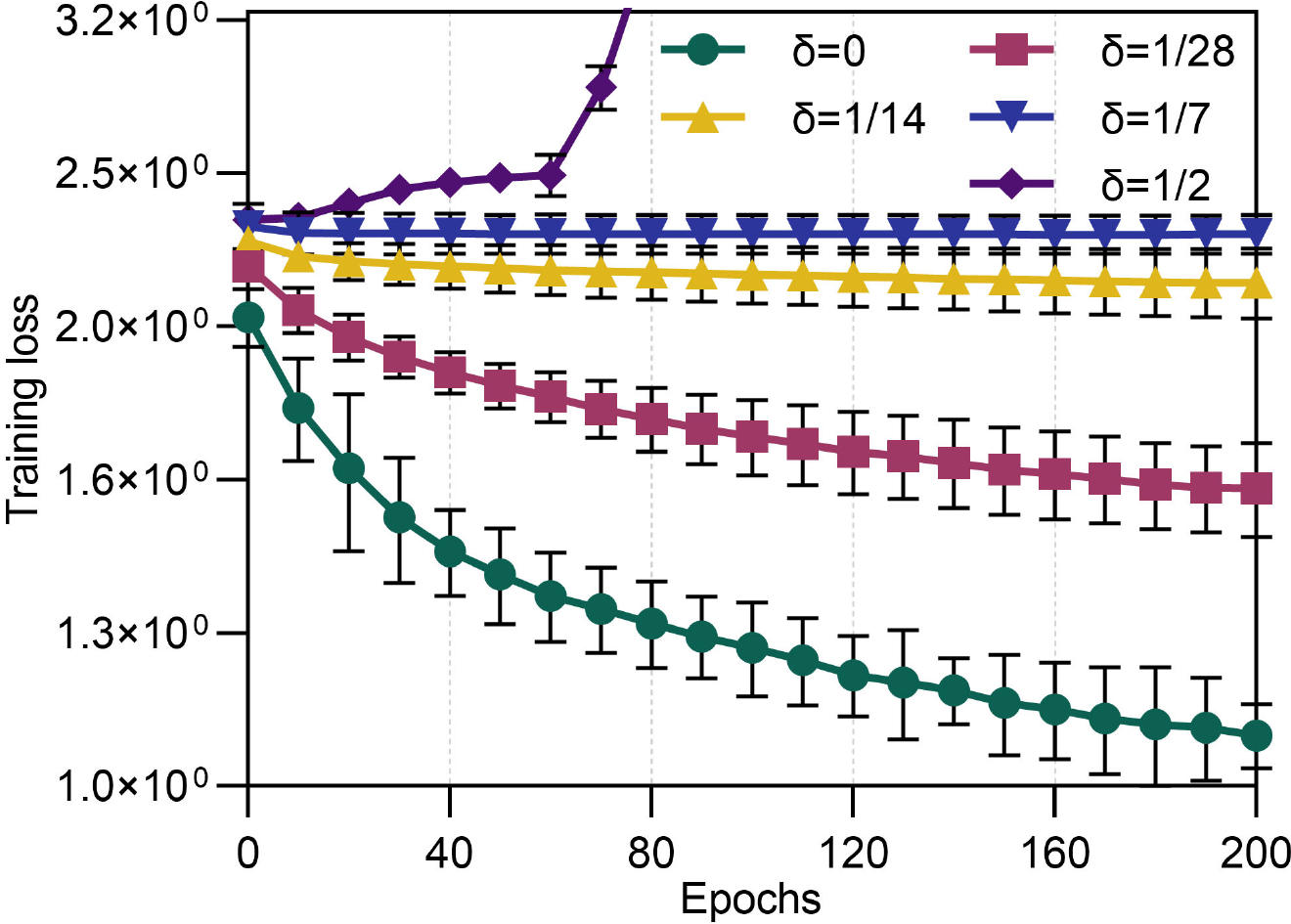}}\!
	\subfigure[{\color{blue}Comparison under different $\lambda$}]{\label{cifard}
		\includegraphics[width=0.24\linewidth]{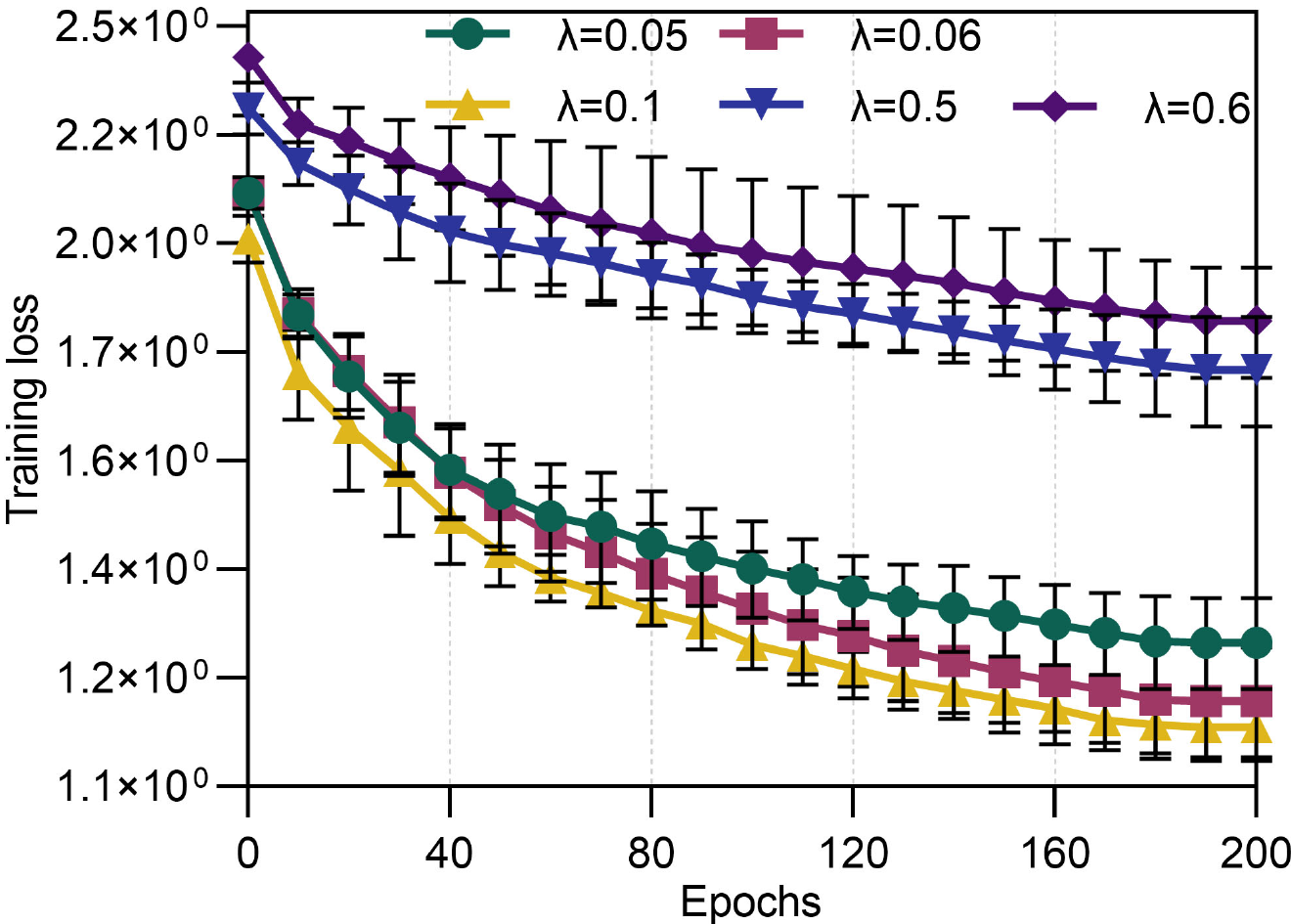}}
	\caption{{\color{blue}Experimental results on CNN training using the ``CIFAR-10" dataset.} (a) and (b) Training loss and cumulative privacy budget comparison of Algorithm~\ref{algorithm1} with existing distributed aggregative optimization algorithms in~\cite{Carnevale1} and~\cite{lixiuxian2} under the same cumulative privacy budget, and existing DP approaches in~\cite{huang}  and~\cite{zhangjifeng1} for distributed optimization. {\color{blue}(c) Training losses of Algorithm~\ref{algorithm1} under different $\delta$ ($\delta$ is given in Corollary~\ref{values}). (d) Training losses of Algorithm~\ref{algorithm1} under different $\lambda$ ($\lambda$ is defined in~\eqref{experimentloss}).} The error bars represent the standard deviation.}
	\label{cifar10}
	\vspace{-0.5em}
\end{figure*}

\subsection{Evaluation results on the ``MNIST" dataset}
In this experiment, we tested the effectiveness of Algorithm~\ref{algorithm1} using the ``MNIST" dataset~\cite{MNIST}, which contains $60,000$ training images and $10,000$ testing images of handwritten digits. Each image is a $28 \times 28$ grayscale image representing a single digit from $0$ to $9$. All agents were equipped with the same two-layer convolutional neural network (CNN), consisting of two convolutional layers with $16$ and $32$ filters, with each followed by a max pooling layer. Finally, a fully connected dense layer maps the extracted features to $10$ output classes. We set the stepsizes
and DP-noise variances in our algorithm as $\lambda_{x}^{t}=(t+1)^{-0.55}$, $\lambda_{y}^{t}=(t+1)^{-0.02}$, $\lambda_{z}^{t}=(t+1)^{-0.03}$, $\sigma_{i,x}^{t}=0.01(t+1)^{-0.505+0.001i}$, $\sigma_{i,y}^{t}=0.01(t+1)^{-0.005+0.001i}$, and $\sigma_{i,z}^{t}=0.01(t+1)^{-0.015+0.001i}$ for $i=1,\cdots,5$, respectively, which satisfy all conditions in Theorem~\ref{algorithm1}. In our comparison, for the algorithm in~\cite{Carnevale1}, we set its stepsize as $\alpha=1$ and its parameter as $\delta=0.5$, and for the algorithm in~\cite{lixiuxian2}, we set its stepsize as $\lambda_{x}^{t}=(t+1)^{-0.5}$, in line with the guidelines provided in~\cite{Carnevale1} and~\cite{lixiuxian2}, respectively (see Section IV in~\cite{Carnevale1} and Eq. (8) in~\cite{lixiuxian2} for details). Furthermore, the stepsizes and DP-noise variances for the DP approach PDOP in~\cite{huang} were set to $\lambda_{x}^{t}=0.1(0.94)^{t}$, $\lambda_{y}^{t}=0.1(0.98)^{t}$, $\lambda_{z}^{t}=0.1(0.97)^{t}$, $\sigma_{x}^{t}=0.07(0.99)^{t}$, $\sigma_{y}^{t}=0.008(0.99)^{t}$ and $\sigma_{z}^{t}=0.008(0.99)^{t}$, respectively. For the DP mechanism in~\cite{zhangjifeng1}, the stepsizes, DP-noise variances, decaying factor $\beta^{t}$, and the time-varying sampling number $\gamma_t$ were set to $\lambda_{x}^{t}=\frac{0.05}{(t+1)^{0.76}}$, $\lambda_{y}^{t}=\frac{0.05}{(t+1)^{0.11}}$, $\lambda_{z}^{t}=\frac{0.05}{(t+1)^{0.65}}$, $\sigma_{x}^{t}=0.1(t+1)^{0.01}$, $\sigma_{y}^{t}=0.01(t+1)^{0.01}$, $\sigma_{z}^{t}=0.01(t+1)^{0.01}$, $\beta^{t}=\frac{0.05}{(t+1)^{0.51}}$, and $\gamma^{t}=\lceil (t+1)^{1.4} \rceil$, respectively. The parameter settings for~\cite{huang} and~\cite{zhangjifeng1} lead to a larger cumulative privacy budget (and hence, weaker privacy protection) than our algorithm, making the optimization-accuracy comparison favorable to these baselines (see the yellow and blue blocks in Fig.~\ref{mnistb}).

The experimental results are summarized in Fig.~\ref{mnist}. It is clear
that the proposed algorithm has much lower training loss than existing counterparts under LDP constraints. In particular, Fig.~\ref{mnista} shows that Laplace noise indeed accumulates
in the conventional gradient-tracking-based distributed aggregative optimization algorithms in~\cite{Carnevale1} and~\cite{lixiuxian2}, leading to
divergent behaviors (see the green and gray curves in Fig.~\ref{mnista}). Furthermore, the yellow curve in Fig.~\ref{mnista} indicate that PDOP, which uses geometrically decaying
stepsizes and DP-noise variances to ensure a finite cumulative privacy budget in the infinite time horizon, is unable to train the complex CNN model. These comparisons corroborate the advantage of the proposed Algorithm~\ref{algorithm1}. {\color{blue} Fig.~\ref{mnistc} shows that our algorithm is robust even when $\delta$ slightly deviates from the range $\delta\in(0,\frac{1}{14})$ derived in Corollary~\ref{values}. Moreover, a smaller $\delta$ leads to lower training loss. Fig.~\ref{mnistd} shows that an overly small or large $\lambda$ results in high training loss, whereas a moderate $\lambda$ improves performance.}
\vspace{-0.5em}

\subsection{Evaluation results on the ``CIFAR-10" dataset}
In the second experiment, we evaluated the performance of Algorithm~\ref{algorithm1} by training a four-layer CNN on the ``CIFAR-10" dataset~\cite{CIFAR10}, which offers greater diversity and complexity compared to the ``MNIST" dataset. The dataset contains $50,000$ training images and $10,000$ testing images, with each image being a $32\!\times\!32$ color image belonging to one of $10$ different classes. The CNN consists of four convolutional layers. Max pooling is applied after the second and fourth layers to reduce spatial dimensions. After global average pooling, the features are fed into a fully connected layer for $10$-class prediction. In this experiment, the stepsizes in our algorithm were set to be the same as those used in the previous ``MNIST" experiment, and the DP-noise variances were set to $\sigma_{i,x}^{t}=0.001(t+1)^{-0.505+0.001i}$, $\sigma_{i,y}^{t}=0.001(t+1)^{-0.005+0.001i}$, and $\sigma_{i,z}^{t}=0.001(t+1)^{-0.015+0.001i}$ for all agents, which satisfy all the conditions in Theorem~\ref{T1}. For the comparison algorithms, the noise variances were reduced by a factor of $10$ compared to those used in the ``MNIST'' experiment, while all other parameters remained unchanged.

The results are summarized in Fig.~\ref{cifar10}. It can be seen that the proposed algorithm has much better robustness to DP noise compared to existing distributed aggregative optimization algorithms in~\cite{Carnevale1} and~\cite{lixiuxian2}. Moreover, compared with the existing DP approaches in~\cite{huang} and~\cite{zhangjifeng1} for distributed optimization, our Algorithm~\ref{algorithm1} achieves lower training loss even under a stronger privacy guarantee. {\color{blue}The trends in the sensitivity analysis shown in Fig.~\ref{cifarc} and Fig.~\ref{cifard} closely resemble those in Fig.~\ref{mnistc} and Fig.~\ref{mnistd}, respectively.}

\section{Conclusions}\label{conclusion}
In this paper, we proposed a distributed stochastic aggregative optimization algorithm under the constraints of local differential privacy and noisy gradients. We proved that the proposed algorithm achieves accurate convergence with explicit rates under nonconvex, general convex, and strongly convex global objective functions, respectively. Simultaneously, we also proved that our algorithm guarantees rigorous $\epsilon_{i}$-local differential privacy with a finite cumulative privacy budget even when the number of iterations tends to infinity. To the best of our knowledge, this is the first algorithm to simultaneously achieve both accurate convergence and rigorous local differential privacy, with a guaranteed finite cumulative privacy budget, in distributed aggregative optimization. This is significant since even in simpler deterministic settings where accurate gradients are accessible to agents, such results have not been reported to the best of our knowledge. Experimental results on personalized machine learning using benchmark datasets
confirmed the efficiency of
our algorithm.

\section*{Appendix}
To simplify notations, we use~$\bar{\mathrm{x}}_{(i)}^{t}=\frac{1}{m}(\sum_{l\neq i}\mathrm{x}_{l(i)}^{t}+x_{i}^{t})$ to denote the average of agent $i$'s optimization variable $x_{i}^{t}$ and other agents' estimates $\mathrm{x}_{l(i)}^{t}$ of this variable for all $l\neq i$. We further define the assembly of the $i$th optimization variable as $\mathbf{x}_{(i)}^{t}\!\!=\!\!\col(\mathrm{x}_{1(i)}^{t},\cdots,\mathrm{x}_{m(i)}^{t})\!\in\!\mathbb{R}^{mn_{i}}$. To further simplify the notations, we define $f(\boldsymbol{x},\boldsymbol{y})=\sum_{i=1}^{m}f_{i}(x_{i},y_{i})\in\mathbb{R}$, $g(\boldsymbol{x})=\frac{1}{m}\sum_{i=1}^{m}g_{i}(x_{i})\in\mathbb{R}^{r}$, $\nabla_{x}\boldsymbol{f}(\boldsymbol{x},\boldsymbol{y})\!=\!\col(\nabla_{x}f_{1}(x_{1},y_{1}),\cdots,\nabla_{x}f_{m}(x_{m},y_{m}))\in\mathbb{R}^{n}$, $\nabla_{y}\boldsymbol{f}(\boldsymbol{x},\boldsymbol{y})\!=\!\col(\nabla_{y}f_{1}(x_{1},y_{1}),\cdots,\nabla_{y}f_{m}(x_{m},y_{m}))\in\mathbb{R}^{mr}$, $\boldsymbol{g}(\boldsymbol{x})\!=\!\col(g_{1}(x_{1}),\cdots,g_{m}(x_{m}))\in\mathbb{R}^{mr}$, and $\nabla \boldsymbol{g}(\boldsymbol{x})=\text{blkdiag}(\nabla g_{1}(x_{1}),\cdots,\nabla g_{m}(x_{m}))\in\mathbb{R}^{n\times mr}$. Moreover, we introduce some additional notations: $u_{i}^{t}\!=\!\nabla_{x}f_{i}^{t}(x_{i}^{t},\tilde{y}_{i}^{t})+\nabla g_{i}^{t}(x_{i}^{t})\tilde{z}_{i}^{t}$, $\hat{x}_{(i)}^{t}=x_{i}^{t}-\bar{\mathrm{x}}_{(i)}^{t}$, $\hat{y}_{i}^{t}=y_{i}^{t}-\bar{y}^{t}$, $\hat{z}_{i}^{t}=z_{i}^{t}-\bar{z}^{t}$, $L_{f}=\max\{L_{f,1},L_{f,2}\}$, and $\bar{L}_{f}=\max\{\bar{L}_{f,1},\bar{L}_{f,2}\}$.

\subsection{Auxiliary lemmas}
\begin{lemma}[Lemma 2  in~\cite{truthfulness}]\label{lemmaLF}
Under Assumption~\ref{A1}, $\nabla F(\boldsymbol{x})$ is $L_{F}$-Lipschitz continuous, i.e., for any $\boldsymbol{x}_{1},\boldsymbol{x}_{2}\in\Omega$, we have $\|\nabla F(\boldsymbol{x}_{1})-\nabla F(\boldsymbol{x}_{2})\|\leq L_{F}\|\boldsymbol{x}_{1}-\boldsymbol{x}_{2}\|,$
with $L_{F}=\bar{L}_{f,1}+L_{f}\bar{L}_{g}+L_{g}(\bar{L}_{f,1}+\bar{L}_{f,2}+L_{g}\bar{L}_{f,2})$.
\end{lemma}
\begin{lemma}[Lemma 7 in~\cite{zijiGT}]\label{lemmae}
The relation $a\gamma^{t}\leq \frac{1}{t^2}$ always holds for all $t\in\mathbb{N}^{+}$ and $\gamma\in(0,1)$, with $a=4^{-1}(\ln(\gamma)e)^2$.
\end{lemma}
\begin{lemma}\label{lemmaexp}
The relationship $\sum_{p=0}^t\frac{\gamma^{t-p}}{(p+1)^v}\leq \frac{b}{(t+1)^v}$ always holds for all $t\in\mathbb{N}^{+}$, $\gamma\in(0,1)$, and $v\in(0,2)$, where the constant $b$ is given by $b=\frac{2^{v+2}}{(\ln(\sqrt{\gamma})e)^2(1-\sqrt{\gamma})}$.
\end{lemma}
\begin{proof}
{\color{blue}For real numbers $a,b,c,d>0$ satisfying $\frac{d}{b}\leq\frac{c}{a}$, the inequality $\frac{d}{b}\leq\frac{c+d}{a+b}\leq\frac{c}{a}$ always holds. Hence, for any $t\in\mathbb{N}^{+}$ and $p\in(0,t)$, we have the inequality $\frac{1}{(t-p)^2}\leq (\frac{p+1}{t})^2$ by setting $a=c=p$, $b=t-p$, and $d=1$.}

By using Lemma~\ref{lemmae} and
the relationship $\frac{1}{(t-p)^2}\leq (\frac{p+1}{t})^2<(\frac{p+1}{t})^{v}$ for any $p\in(0,t)$ and $v\in(0,2)$, we obtain
\begin{equation}
\textstyle\frac{\sqrt{\gamma}^{t-p}}{(p+1)^v}\textstyle\leq \frac{1}{a(p+1)^v(t-p)^2}<\frac{1}{a(p+1)^v}\left(\frac{p+1}{t}\right)^{v}=\frac{1}{at^v},\label{L41}
\end{equation}
with $a\!=\!4^{-1}(\ln(\sqrt{\gamma})e)^2$. Note that~\eqref{L41} still holds when $p\!=\!0$.

Using~\eqref{L41} and the geometric series formula, we arrive at $\sum_{p=0}^{t} \frac{\gamma^{t-p}}{(p+1)^v}\leq \sum_{p=0}^t\sqrt{\gamma}^{t-p}\frac{1}{at^v}\leq\frac{4\times 2^{v}}{(\ln(\sqrt{\gamma})e)^2(1-\sqrt{\gamma})(t+1)^{v}}.$
\end{proof}
\begin{lemma}\label{AppLemma1}
Under Assumptions~\ref{A1}-\ref{A4}, for any $t\in\mathbb{N}^{+}$, we have the following inequality for Algorithm~\ref{algorithm1}:
	\begin{flalign}
		&\textstyle\sum_{i=1}^{m}\mathbb{E}[\|\mathbf{x}_{(i)}^{t+1}-\boldsymbol{1}_{m}\otimes \bar{\mathrm{x}}_{(i)}^{t+1}\|^2]\leq\frac{c_{x1}(\lambda_{x}^{t})^2}{(\lambda_{z}^{t})^2}\mathbb{E}[\|\boldsymbol{\hat{z}}^{t+1}-\boldsymbol{\hat{z}}^{t}\|^2]\nonumber\\
		&\textstyle+(1-|\rho_{2}|)\sum_{i=1}^{m}\mathbb{E}[\|\mathbf{x}_{(i)}^{t}-\boldsymbol{1}_{m}\otimes\bar{\mathrm{x}}_{(i)}^{t}\|^2]\!+\!c_{x2}(\sigma_{x}^{t})^{2},
		\label{AppLemma1results}
		\vspace{-0.2em}
	\end{flalign}
	where
	$c_{x1}\!=\!\frac{c_{u1}}{|\rho_{2}|}$ and $c_{x2}\!=\!\frac{c_{u1}mr(\lambda_{x}^{0})^2\sigma_{z}^2}{|\rho_{2}|(\lambda_{z}^{0})^2}+\frac{mc_{u2}(\lambda_{x}^{0})^2}{|\rho_{2}|}+mn$ with $c_{u1}=4(\sigma_{g,1}^2+L_{g}^2)$ and $c_{u2}=2(\sigma_{f}^2+L_{f}^2)(1+2\sigma_{g,1}^2+2L_{g}^2)$.
\end{lemma}
\begin{proof}
According to the dynamics of $\mathbf{x}_{i}^{t}$ and the projection inequality, for any $i\in[m]$, $j\neq i$ and $\bar{\mathrm{x}}_{(i)}^{t}\in\Omega_{i}$, we have
\begin{equation}
\vspace{-0.2em}
\begin{aligned}
\|x_{i}^{t+1}-\bar{\mathrm{x}}_{(i)}^{t}\|&\textstyle\leq\|x_{i}^{t}+\sum_{j\in\mathcal{N}_{i}}w_{ij}(\mathrm{x}_{j(i)}^{t}+\vartheta_{j(i)}^{t}-x_{i}^{t})\\
&\textstyle\quad-\lambda_{x}^{t}u_{i}^{t}-\bar{\mathrm{x}}_{(i)}^{t}\|,\\
\|\mathrm{x}_{j(i)}^{t+1}-\bar{\mathrm{x}}_{(i)}^{t}\|&\textstyle\leq\|\mathrm{x}_{j(i)}^{t}+\sum_{l\in\mathcal{N}_{j}}w_{jl}(\mathrm{x}_{l(i)}^{t}+\vartheta_{l(i)}^{t}-\mathrm{x}_{j(i)}^{t})\\
&\textstyle\quad-\bar{\mathrm{x}}_{(i)}^{t}\|,~\forall j \neq i. \nonumber
\end{aligned}
\end{equation}
The compact form of the preceding inequalities is
\begin{equation}
\begin{aligned}
&\textstyle\|\mathbf{x}_{(i)}^{t+1}-\boldsymbol{1}_{m}\otimes \bar{\mathrm{x}}_{(i)}^{t}\|\leq\|((I_{m}+W)\otimes I_{n_{i}})\mathbf{x}_{(i)}^{t}\\
&\textstyle\quad+(W^{0}\otimes I_{n_{i}})\boldsymbol{\vartheta}_{(i)}^{t}-\lambda_{x}^{t}\boldsymbol{u}_{(i)}^{t}-\boldsymbol{1}_{m}\otimes \bar{\mathrm{x}}_{(i)}^{t}\|,\label{xdynamics}
\end{aligned}
\vspace{-0.2em}
\end{equation}
where $\boldsymbol{u}_{(i)}^{t}$ is given by $\boldsymbol{u}_{(i)}^{t}=\col(\boldsymbol{0}_{n_{i}},\cdots,u_{i},\cdots,\boldsymbol{0}_{n_{i}})$, with the $i$th block being $u_{i}$ and all others being $0$, and $W^{0}$ is the same as $W$, except that all of its diagonal entries are zero.

By defining $\tilde{W}=I_{m}+W-\frac{\boldsymbol{1}_{m}\boldsymbol{1}_{m}^{T}}{m}$ and using the relation $\bar{\mathrm{x}}_{(i)}^{t}=\frac{\boldsymbol{1}_{m}^{T}\otimes I_{n_{i}}}{m}\mathbf{x}_{(i)}^{t}$, we can rewrite~\eqref{xdynamics} as
\begin{flalign}
&\textstyle\|\mathbf{x}_{(i)}^{t+1}-\boldsymbol{1}_{m}\otimes \bar{\mathrm{x}}_{(i)}^{t}\|^2
\leq\|(W^{0}\otimes I_{n_{i}})\boldsymbol{\vartheta}_{(i)}^{t}\|^2\nonumber\\
&\textstyle\quad+\|(\tilde{W}\otimes I_{n_{i}})(\mathbf{x}_{(i)}^{t}-\boldsymbol{1}_{m}\otimes\bar{\mathrm{x}}_{(i)}^{t})-\lambda_{x}^{t}\boldsymbol{u}_{(i)}^{t}\|^2\nonumber\\
&\textstyle\quad+2\langle(\tilde{W}\otimes I_{n_{i}})\mathbf{x}_{(i)}^{t}-\lambda_{x}^{t}\boldsymbol{u}_{(i)}^{t},(W^{0}\otimes I_{n_{i}})\boldsymbol{\vartheta}_{(i)}^{t}\rangle,
\label{1AppL1}
\end{flalign}
where we have subtracted $(\tilde{W}\otimes I_{n_{i}})(\boldsymbol{1}_{m}\otimes\bar{\mathrm{x}}_{(i)}^{t})$ from the right hand side of~\eqref{xdynamics} due to $(\tilde{W}\otimes I_{n_{i}})(\boldsymbol{1}_{m}\otimes \bar{\mathrm{x}}_{(i)}^{t})=\boldsymbol{0}_{mn_{i}}$.

Based on $\|\tilde{W}\|<1-|\rho_{2}|$ from Assumption~\ref{A3}, the second term on the right hand side of~\eqref{1AppL1} satisfies
\begin{equation}
\begin{aligned}
&\textstyle\|(\tilde{W}\otimes I_{n_{i}})(\mathbf{x}_{(i)}^{t}-\boldsymbol{1}_{m}\otimes\bar{\mathrm{x}}_{(i)}^{t})-\lambda_{x}^{t}\boldsymbol{u}_{(i)}^{t}\|^2\\
&\textstyle< (1-|\rho_{2}|)\|\mathbf{x}_{(i)}^{t}-\boldsymbol{1}_{m}\otimes\bar{\mathrm{x}}_{(i)}^{t}\|^2+\frac{(\lambda_{x}^{t})^2}{|\rho_{2}|}\|\boldsymbol{u}_{(i)}^{t}\|^2,\label{1AppL2}
\end{aligned}
\end{equation}
where in the derivation we have used the inequality $(a+b)^2\leq(1+\varepsilon)a^2+(1+\varepsilon^{-1})b^2$ for any scalars $a$, $b$, and $\varepsilon>0$ (by setting $\varepsilon=\frac{|\rho_{2}|}{1-|\rho_{2}|}$, i.e., $1+\varepsilon=\frac{1}{1-|\rho_{2}|}$ and $1+\varepsilon^{-1}=\frac{1}{|\rho_{2}|}$).

Substituting~\eqref{1AppL2} into~\eqref{1AppL1} and using $\|W^{0}\otimes I_{n_{i}}\|\leq 1$ yield
\begin{flalign}
&\textstyle\|\mathbf{x}_{(i)}^{t+1}\!-\!\boldsymbol{1}_{m}\otimes \bar{\mathrm{x}}_{(i)}^{t}\|^2\!<\!(1\!-\!|\rho_{2}|)\|\mathbf{x}_{(i)}^{t}\!-\!\boldsymbol{1}_{m}\otimes\bar{\mathrm{x}}_{(i)}^{t}\|^2\!+\!\|\boldsymbol{\vartheta}_{(i)}^{t}\|^2\nonumber\\
&\textstyle\quad\!+\!\frac{(\lambda_{x}^{t})^2}{|\rho_{2}|}\|\boldsymbol{u}_{(i)}^{t}\|^2\!+\!2\langle(\tilde{W}\!\otimes\! I_{n_{i}})\mathbf{x}_{(i)}^{t}\!-\!\lambda_{x}^{t}\boldsymbol{u}_{(i)}^{t},(W^{0}\!\otimes\! I_{n_{i}})\boldsymbol{\vartheta}_{(i)}^{t}\rangle.\nonumber
\end{flalign}

Note that when the variables $\mathrm{x}_{l(i)}^{t+1}$ are constrained in a convex set $\Omega_{i}$, their average $\bar{\mathrm{x}}_{(i)}^{t+1}$ also lies in the set $\Omega_{i}$, implying that the average $\bar{\mathrm{x}}_{(i)}^{t+1}$ is the solution to the constrained problem of minimizing $\sum_{l=1}^{m}\|\mathrm{x}_{l(i)}^{t+1}-x\|$ over $x\in\Omega_{i}$. Using this and the relation 
$\sum_{l=1}^{m}\|\mathrm{x}_{l(i)}^{t+1}-x\|^2=\|\mathbf{x}_{(i)}^{t+1}-\boldsymbol{1}_{m}\otimes x\|^2$, we have $\|\mathbf{x}_{(i)}^{t+1}-\boldsymbol{1}_{m}\otimes \bar{\mathrm{x}}_{(i)}^{t+1}\|^2\leq \|\mathbf{x}_{(i)}^{t+1}-\boldsymbol{1}_{m}\otimes \bar{\mathrm{x}}_{(i)}^{t}\|^2$, which implies
\begin{flalign}
&\textstyle\|\mathbf{x}_{(i)}^{t+1}\!-\!\boldsymbol{1}_{m}\!\otimes\! \bar{\mathrm{x}}_{(i)}^{t+1}\|^2\!\leq\!(1\!-\!|\rho_{2}|)\|\mathbf{x}_{(i)}^{t}\!-\!\boldsymbol{1}_{m}\!\otimes\!\bar{\mathrm{x}}_{(i)}^{t}\|^2\!+\!\|\boldsymbol{\vartheta}_{(i)}^{t}\|^2\nonumber\\
&\textstyle+\!\frac{(\lambda_{x}^{t})^2}{|\rho_{2}|}\|\boldsymbol{u}_{(i)}^{t}\|^2\!+\!2\langle(\tilde{W}\!\otimes\! I_{n_{i}})\mathbf{x}_{(i)}^{t}\!-\!\lambda_{x}^{t}\boldsymbol{u}_{(i)}^{t},(W^{0}\!\otimes\! I_{n_{i}})\boldsymbol{\vartheta}_{(i)}^{t}\rangle.\nonumber
\end{flalign}

Taking the expectation on both sides of the preceding inequality and using  $\mathbb{E}[\boldsymbol{\vartheta}_{(i)}^{t}]=0$ from Assumption~\ref{A4} yield
\begin{flalign}
&\textstyle\mathbb{E}[\|\mathbf{x}_{(i)}^{t+1}\!-\!\boldsymbol{1}_{m}\otimes \bar{\mathrm{x}}_{(i)}^{t+1}\|^2]\!\leq\!(1-|\rho_{2}|)\mathbb{E}[\|\mathbf{x}_{(i)}^{t}\!-\!\boldsymbol{1}_{m}\otimes\bar{\mathrm{x}}_{(i)}^{t}\|^2]\nonumber\\
&\textstyle\quad+\frac{(\lambda_{x}^{t})^2}{|\rho_{2}|}\mathbb{E}[\|\boldsymbol{u}_{(i)}^{t}\|^2]+mn_{i}(\sigma_{x}^{t})^{2},
\label{1AppL6}
\end{flalign}
where the DP-noise variance $\sigma_{x}^{t}$ is given by $\sigma_{x}^{t}=\frac{\sigma_{x}}{(t+1)^{\varsigma_{x}}}$ with $\sigma_{x}=\max_{i\in[m]}\{\sigma_{i,x}\}$ and $\varsigma_{x}=\min_{i\in[m]}\{\varsigma_{i,x}\}$.

According to the definition of $\boldsymbol{u}_{(i)}^{t}$ given in~\eqref{xdynamics} and the definition $u_{i}^{t}=\nabla_{x}f_{i}^{t}(x_{i}^{t},\tilde{y}_{i}^{t})-\nabla g_{i}^{t}(x_{i}^{t})\tilde{z}_{i}^{t}$, we have
\begin{equation}
\vspace{-0.2em}
\textstyle\mathbb{E}[\|\boldsymbol{u}_{(i)}^{t}\|^2]\leq 2(\sigma_{f}^2+L_{f}^2)+2(\sigma_{g,1}^2+L_{g}^2)\mathbb{E}[\|\tilde{z}_{i}^{t}\|^2],\label{1AppL7}
\end{equation}
where in the derivation we have used
Assumptions~\ref{A1} and~\ref{A2}.

We proceed to analyze the term $\mathbb{E}[\|\tilde{z}_{i}^{t}\|^2]$ in~\eqref{1AppL7}. Based on the update rule of $z_{i}^{t}$, we have $\bar{z}^{t+1}=\bar{z}^{t}+\bar{\zeta}_{w}^{t}+\lambda_{z}^{t}\nabla_{y}\bar{f}^{t}(\boldsymbol{x}^{t},\boldsymbol{\tilde{y}}^{t})$ with $\bar{\zeta}_{w}^{t}\triangleq\frac{1}{m}\sum_{i=1}^{m}\sum_{j\in\mathcal{N}_{i}}w_{ij}\zeta_{j}^{t}$.

Recalling the definitions $\tilde{z}_{i}^{t}\!=\!\frac{1}{\lambda_{z}^{t}}(z_{i}^{t+1}-z_{i}^{t})$ and $\hat{z}_{i}^{t}=z_{i}^{t}-\bar{z}^{t}$, we have $\tilde{z}_{i}^{t}\!=\!\frac{1}{\lambda_{z}^{t}}(\hat{z}_{i}^{t+1}-\hat{z}_{i}^{t}+\bar{z}^{t+1}-\bar{z}^{t})$, which implies
\begin{equation}
	\textstyle\mathbb{E}[\|\tilde{z}_{i}^{t}\|^2]\leq \frac{2}{(\lambda_{z}^{t})^2}\mathbb{E}[\|\hat{z}_{i}^{t+1}-\hat{z}_{i}^{t}\|^2]+\frac{2r(\sigma_{z}^{t})^2}{(\lambda_{z}^{t})^2}+2(\sigma_{f}^2+L_{f}^2),\nonumber
\end{equation}
where the DP-noise variance $\sigma_{z}^{t}$ is given by $\sigma_{z}^{t}=\frac{\sigma_{z}}{(t+1)^{\varsigma_{z}}}$ with $\sigma_{z}=\max_{i\in[m]}\{\sigma_{i,z}\}$ and $\varsigma_{z}=\min_{i\in[m]}\{\varsigma_{i,z}\}$.

Substituting the preceding inequality into~\eqref{1AppL7}, we obtain
\begin{equation}
\textstyle\mathbb{E}[\|\boldsymbol{u}_{(i)}^{t}\|^2]\leq \frac{c_{u1}}{(\lambda_{z}^{t})^2}\mathbb{E}[\|\hat{z}_{i}^{t+1}-\hat{z}_{i}^{t}\|^2]+\frac{c_{u1}r(\sigma_{z}^{t})^2}{(\lambda_{z}^{t})^2}+c_{u2},\label{1AppL9}
\end{equation}
with $c_{u1}$ and $c_{u2}$ given in the lemma statement.

Further incorporating~\eqref{1AppL9} into~\eqref{1AppL6}, we arrive at~\eqref{AppLemma1results}.
\vspace{-0.2em}
\end{proof}
\begin{lemma}\label{AppLemma2}
Under Assumptions~\ref{A1}-\ref{A4}, for any $t\in\mathbb{N}^{+}$, we have the following inequality for Algorithm~\ref{algorithm1}:
\begin{equation}
\begin{aligned}
&\textstyle \mathbb{E}[\|\boldsymbol{\hat{y}}^{t+1}-\boldsymbol{\hat{y}}^{t}\|^2]\leq(1-|\rho_{2}|)\mathbb{E}[\|\boldsymbol{\hat{y}}^{t}-\boldsymbol{\hat{y}}^{t-1}\|^2]\\
&\textstyle\quad+c_{y1}(\lambda_{y}^{t})^{2}\sum_{i=1}^{m}\mathbb{E}[\|\mathbf{x}_{(i)}^{t-1}-\boldsymbol{1}_{m}\otimes \bar{\mathrm{x}}_{(i)}^{t-1}\|^2]\\
&\textstyle\quad+\frac{c_{y2}(\lambda_{y}^{t})^{2}(\lambda_{x}^{t-1})^2}{(\lambda_{z}^{t-1})^2}\mathbb{E}[\|\boldsymbol{\hat{z}}^{t}-\boldsymbol{\hat{z}}^{t-1}\|^2]+c_{y3}(\sigma_{y}^{t-1})^2,\label{AppLemma2results}
\end{aligned}
\end{equation}
where $c_{y1}\!=\!4L_{g}^2|\rho_{2}|^{-1}$,$c_{y2}\!=\!c_{y1}c_{u1}$, and~$c_{y3}\!=\!2mr+2(\sigma_{g,0}^2+d_{g}^2)mv_{y}^{2}(\lambda_{y}^{0})^2|\rho_{2}|^{-1}+4m\sigma_{g,0}^2(\lambda_{y}^{1})^{2}|\rho_{2}|^{-1}+\frac{1}{2}c_{y1}n(\lambda_{y}^{1})^{2}\sigma_{x}^2+mc_{y1}c_{u1}r(\lambda_{y}^{1})^{2}(\lambda_{x}^{0})^2\sigma_{z}^2(\lambda_{z}^{0})^{-2}+mc_{y1}c_{u2}(\lambda_{y}^{1})^{2}(\lambda_{x}^{0})^2$, with $c_{u1}$ and $c_{u2}$ given in the statement of Lemma~\ref{AppLemma1}.
\end{lemma}
\begin{proof}
According to the update rule of $y_{i}^{t}$ in line 4 of Algorithm~\ref{algorithm1}, we have
\begin{equation}
\vspace{-0.2em}
\begin{aligned}
&\textstyle \boldsymbol{y}^{t+1}-\boldsymbol{1}_{m}\otimes \bar{y}^{t+1}=(\tilde{W}\otimes I_{r})(\boldsymbol{y}^{t}-\boldsymbol{1}_{m}\otimes \bar{y}^{t})\\
&\textstyle\quad+(\Gamma\otimes I_{r})(W^{0}\otimes I_{r})\boldsymbol{\chi}^{t}+\lambda_{y}^{t}(\Gamma\otimes I_{r})\boldsymbol{g}^{t}(\boldsymbol{x}^{t}),\label{2LApp5}
\end{aligned}
\vspace{-0.2em}
\end{equation}
with $\tilde{W}\triangleq I_{m}+W-\frac{\boldsymbol{1}\boldsymbol{1}^{T}}{m}$ and $\Gamma\triangleq I_{m}-\frac{\boldsymbol{1}_{m}\boldsymbol{1}_{m}^{T}}{m}$.

Recalling the definition $\boldsymbol{\hat{y}}^{t}=\boldsymbol{y}^{t}-\boldsymbol{1}_{m}\otimes \bar{y}^{t}$ and using an argument similar to the derivation of~\eqref{1AppL6} yield
\begin{equation}
\vspace{-0.2em}
\begin{aligned}
&\textstyle \mathbb{E}[\|\boldsymbol{\hat{y}}^{t+1}-\boldsymbol{\hat{y}}^{t}\|^2]\leq(1-|\rho_{2}|)\mathbb{E}[\|\boldsymbol{\hat{y}}^{t}-\boldsymbol{\hat{y}}^{t-1}\|^2]\\
&\textstyle\quad+\frac{1}{|\rho_{2}|}\mathbb{E}[\|\lambda_{y}^{t}\boldsymbol{g}^{t}(\boldsymbol{x}^{t})-\lambda_{y}^{t-1}\boldsymbol{g}^{t-1}(\boldsymbol{x}^{t-1})\|^2]+2mr(\sigma_{y}^{t-1})^2,\label{2LApp8}
\end{aligned}
\vspace{-0.2em}
\end{equation}
where the noise variance $\sigma_{y}^{t}$ is given by $\sigma_{y}^{t}=\frac{\sigma_{y}}{(t+1)^{\varsigma_{y}}}$ with $\sigma_{y}=\max_{i\in[m]}\{\sigma_{i,y}\}$ and $\varsigma_{y}=\min_{i\in[m]}\{\varsigma_{i,y}\}$.

The second term on the right hand side of~\eqref{2LApp8} satisfies
\begin{equation}
\vspace{-0.2em}
\begin{aligned}
&\textstyle\mathbb{E}[\|\lambda_{y}^{t}\boldsymbol{g}^{t}(\boldsymbol{x}^{t})-\lambda_{y}^{t-1}\boldsymbol{g}^{t-1}(\boldsymbol{x}^{t-1})\|^2]\\
&\textstyle\leq 2(\lambda_{y}^{t})^2\mathbb{E}[\|\boldsymbol{g}^{t}(\boldsymbol{x}^{t})-\boldsymbol{g}^{t-1}(\boldsymbol{x}^{t-1})\|^2]\\
&\textstyle\quad+2(\lambda_{y}^{t}-\lambda_{y}^{t-1})^2\mathbb{E}[\|\boldsymbol{g}^{t-1}(\boldsymbol{x}^{t-1})\|^2].\label{2LApp9}
\end{aligned}
\vspace{-0.2em}
\end{equation}

To characterize the first term on the right-hand side of~\eqref{2LApp9}, we note that Assumptions~\ref{A1} and~\ref{A2} imply
\begin{equation}
\begin{aligned}
&\textstyle2(\lambda_{y}^{t})^2\mathbb{E}[\|\boldsymbol{g}^{t}(\boldsymbol{x}^{t})-\boldsymbol{g}^{t-1}(\boldsymbol{x}^{t-1})\|^2]\\
&\textstyle\leq 4m\sigma_{g,0}^2(\lambda_{y}^{t})^{2}t^{-1}+2L_{g}^2(\lambda_{y}^{t})^{2}\mathbb{E}[\|\boldsymbol{x}^{t}-\boldsymbol{x}^{t-1}\|^2].\label{2LApp10}
\end{aligned}
\end{equation}
By defining 
$\vartheta_{w(i)}^{t}\triangleq\sum_{j\in\mathcal{N}_{i}}w_{ij}\vartheta_{j(i)}^{t}$ and and
using~\eqref{1AppL9}, the last term on the right hand side of~\eqref{2LApp10} satisfies
\begin{flalign}
&\textstyle\mathbb{E}[\|\boldsymbol{x}^{t}-\boldsymbol{x}^{t-1}\|^2]\leq 2\sum_{i=1}^{m}\mathbb{E}[\|\mathbf{x}_{(i)}^{t-1}-\boldsymbol{1}_{m}\otimes \bar{\mathrm{x}}_{(i)}^{t-1}\|^2]\nonumber\\
&\textstyle\quad+\frac{2c_{u1}(\lambda_{x}^{t-1})^2}{(\lambda_{z}^{t-1})^2}\mathbb{E}[\|\boldsymbol{\hat{z}}^{t}-\boldsymbol{\hat{z}}^{t-1}\|^2]+\frac{2mc_{u1}r(\lambda_{x}^{t-1})^2(\sigma_{z}^{t-1})^2}{(\lambda_{z}^{t-1})^2}\nonumber\\
&\textstyle\quad+2mc_{u2}(\lambda_{x}^{t-1})^2+n(\sigma_{x}^{t-1})^2.\label{xsubx}
\end{flalign}
Substituting~\eqref{xsubx} into~\eqref{2LApp10}, we obtain that the first term on the right hand side of~\eqref{2LApp9} satisfies
\begin{flalign}
&\textstyle2(\lambda_{y}^{t})^2\mathbb{E}[\|\boldsymbol{g}^{t}(\boldsymbol{x}^{t})-\boldsymbol{g}^{t-1}(\boldsymbol{x}^{t-1})\|^2]\leq 4m\sigma_{g,0}^2(\lambda_{y}^{t})^{2}t^{-1}\nonumber\\
&\textstyle\!\!+4L_{g}^2(\lambda_{y}^{t})^{2}\sum_{i=1}^{m}\mathbb{E}[\|\mathbf{x}_{(i)}^{t-1}-\boldsymbol{1}_{m}\otimes \bar{\mathrm{x}}_{(i)}^{t-1}\|^2]\nonumber\\
&\textstyle\!\!+\frac{4L_{g}^2c_{u1}(\lambda_{y}^{t})^{2}(\lambda_{x}^{t-1})^2}{(\lambda_{z}^{t-1})^2}\mathbb{E}[\|\boldsymbol{\hat{z}}^{t}\!-\!\boldsymbol{\hat{z}}^{t-1}\|^2]\!+\!2nL_{g}^2(\lambda_{y}^{t})^{2}(\sigma_{x}^{t-1})^2\nonumber\\
&\textstyle\!\!+\left(\frac{4mL_{g}^2c_{u1}r(\sigma_{z}^{t-1})^2}{(\lambda_{z}^{t-1})^2}\!+\!4mL_{g}^2c_{u2}\right)(\lambda_{y}^{t})^{2}(\lambda_{x}^{t-1})^2.\label{2LApp13}
\end{flalign}

Using the relation $(\lambda_{y}^{t}-\lambda_{y}^{t-1})^2\leq \left(\frac{v_{y}\lambda_{y}^{0}}{t^{1+v_{y}}}\right)^2$ for some $\varepsilon\in(t, t+1)$, the second term on the right hand side of~\eqref{2LApp9} satisfies
\begin{equation}
	\textstyle2(\lambda_{y}^{t}-\lambda_{y}^{t-1})^2\mathbb{E}[\|\boldsymbol{g}^{t-1}(\boldsymbol{x}^{t-1})\|^2]\leq \frac{2(\sigma_{g,0}^2+d_{g}^2)mv_{y}^{2}(\lambda_{y}^{0})^2}{t^{2+2v_{y}}},\label{2LApp14}
\end{equation}
where we have used $\max_{i\in[m]}\sup_{x_{i}\in\Omega_{i}}\{\|g_{i}(x_{i})\|\}\leq d_{g}$ for some $d_{g}>0$. Substituting~\eqref{2LApp13} and~\eqref{2LApp14} into~\eqref{2LApp9}, and further substituting~\eqref{2LApp9} into~\eqref{2LApp8}, we arrive at~\eqref{AppLemma2results}.
\end{proof}
\begin{lemma}\label{AppLemma3}
Under Assumptions~\ref{A1}-\ref{A4}, for any $t\in\mathbb{N}^{+}$, the following inequality always holds for Algorithm~\ref{algorithm1}:
\begin{equation}
\begin{aligned}
&\textstyle \mathbb{E}[\|\boldsymbol{\hat{z}}^{t+1}-\boldsymbol{\hat{z}}^{t}\|^2]\leq(1-|\rho_{2}|)\mathbb{E}[\|\boldsymbol{\hat{z}}^{t}-\boldsymbol{\hat{z}}^{t-1}\|^2]\\
&\textstyle+c_{z1}(\lambda_{z}^{t})^{2}\sum_{i=1}^{m}\mathbb{E}[\|\mathbf{x}_{(i)}^{t-1}-\boldsymbol{1}_{m}\otimes \bar{\mathrm{x}}_{(i)}^{t-1}\|^2]\\
&\textstyle+\frac{c_{z2}(\lambda_{z}^{t})^{2}}{(\lambda_{y}^{t})^2}\mathbb{E}[\|\boldsymbol{\hat{y}}^{t+1}\!-\!\boldsymbol{\hat{y}}^{t}\|^2]\!+\!\frac{c_{z2}(\lambda_{z}^{t})^{2}}{(\lambda_{y}^{t-1})^2}\mathbb{E}[\|\boldsymbol{\hat{y}}^{t}-\boldsymbol{\hat{y}}^{t-1}\|^2]\\
&\textstyle+c_{z3}(\lambda_{x}^{t-1})^2\mathbb{E}[\|\boldsymbol{\hat{z}}^{t}\!-\!\boldsymbol{\hat{z}}^{t-1}\|^2]\!+\!c_{z4}\frac{(\lambda_{z}^{t})^{2}(\sigma_{y}^{t})^2}{(\lambda_{y}^{t})^2}\!+\!2mr(\sigma_{z}^{t-1})^2,\label{AppLemma3results}
\end{aligned}
\end{equation}
where $c_{z1}=4\bar{L}_{f}^2|\rho_{2}|^{-1}$, $c_{z2}=2c_{z1}$, $c_{z3}=c_{z1}c_{u1}$, and $c_{z4}=c_{z1}(\frac{1}{2}n(\lambda_{z}^{1})^{2}(\sigma_{x}^{0})^{2}+mc_{u1}r(\lambda_{x}^{1})^2\sigma_{z}^2+4mr+mc_{u2}(\lambda_{z}^{1})^{2}(\lambda_{x}^{0})^2\!+\!4m(\sigma_{g,0}^2+d_{g}^2)(\lambda_{z}^{1})^{2})+2(\sigma_{f}^2+L_{f}^2)mv_{z}^{2}(\lambda_{z}^{0})^2|\rho_{2}|^{-1}$.
\end{lemma}
\begin{proof}
Since the dynamics of $z_{i}^{t}$ are similar to that of $y_{i}^{t}$, we use an argument similar to the derivation of~\eqref{2LApp8} to obtain
\begin{equation}
\vspace{-0.1em}
\begin{aligned}
&\textstyle \mathbb{E}[\|\boldsymbol{\hat{z}}^{t+1}-\boldsymbol{\hat{z}}^{t}\|^2]\leq(1-|\rho_{2}|)\mathbb{E}[\|\boldsymbol{\hat{z}}^{t}-\boldsymbol{\hat{z}}^{t-1}\|^2]+2mr(\sigma_{z}^{t-1})^2\\
&\textstyle\quad+\frac{1}{|\rho_{2}|}\mathbb{E}[\|\lambda_{z}^{t}\nabla_{y}\boldsymbol{f}^{t}(\boldsymbol{x}^{t},\boldsymbol{\tilde{y}}^{t})\!-\!\lambda_{z}^{t-1}\nabla_{y}\boldsymbol{f}^{t-1}(\boldsymbol{x}^{t-1},\boldsymbol{\tilde{y}}^{t-1})\|^2].\label{3LApp1}
\end{aligned}
\vspace{-0.1em}
\end{equation}

Following an argument similar to the derivation of~\eqref{2LApp10}, we can arrive at
\begin{flalign}
&\textstyle\mathbb{E}[\|\nabla_{y}\boldsymbol{f}^{t}(\boldsymbol{x}^{t},\boldsymbol{\tilde{y}}^{t})\!-\!\nabla_{y}\boldsymbol{f}^{t-1}(\boldsymbol{x}^{t-1},\boldsymbol{\tilde{y}}^{t-1})\|^2]\leq 2m\sigma_{f}^2t^{-1}\nonumber\\
&\textstyle\quad+\bar{L}_{f}^2(\mathbb{E}[\|\boldsymbol{x}^{t}-\boldsymbol{x}^{t-1}\|^2]+\mathbb{E}[\|\boldsymbol{\tilde{y}}^{t}-\boldsymbol{\tilde{y}}^{t-1}\|^2]).\label{3LApp3}
\end{flalign}

Using an argument similar to the derivation of~\eqref{2LApp14} yields
\vspace{-0.2em}
\begin{flalign}
&\textstyle\mathbb{E}[\|\boldsymbol{\tilde{y}}^{t}-\boldsymbol{\tilde{y}}^{t-1}\|^2]\leq8m(\sigma_{g,0}^2+d_{g}^2)\!+\! \frac{4}{(\lambda_{y}^{t})^2}\mathbb{E}[\|\boldsymbol{\hat{y}}^{t+1}-\boldsymbol{\hat{y}}^{t}\|^2]\nonumber\\
&\textstyle\!+\!\frac{4}{(\lambda_{y}^{t-1})^2}\mathbb{E}[\|\boldsymbol{\hat{y}}^{t}-\boldsymbol{\hat{y}}^{t-1}\|^2]\!+\!4mr\left(\frac{(\sigma_{y}^{t})^2}{(\lambda_{y}^{t})^2}+\frac{(\sigma_{y}^{t-1})^2}{(\lambda_{y}^{t-1})^2}\right).\label{3LApp5}
\end{flalign}

Substituting~\eqref{xsubx} and~\eqref{3LApp5} into~\eqref{3LApp3}, we can obtain
\vspace{-0.2em}
\begin{flalign}
&\textstyle2(\lambda_{z}^{t})^2\mathbb{E}[\|\nabla_{y}\boldsymbol{f}^{t}(\boldsymbol{x}^{t},\boldsymbol{\tilde{y}}^{t})\!-\!\nabla_{y}\boldsymbol{f}^{t-1}(\boldsymbol{x}^{t-1},\boldsymbol{\tilde{y}}^{t-1})\|^2]\nonumber\\
&\textstyle\leq4m\sigma_{f}^2(\lambda_{z}^{t})^{2}t^{-1}+4\bar{L}_{f}^2(\lambda_{z}^{t})^{2}\sum_{i=1}^{m}\mathbb{E}[\|\mathbf{x}_{(i)}^{t-1}-\boldsymbol{1}_{m}\otimes \bar{\mathrm{x}}_{(i)}^{t-1}\|^2]\nonumber\\
&\textstyle\quad+\frac{8\bar{L}_{f}^2(\lambda_{z}^{t})^{2}}{(\lambda_{y}^{t})^2}\mathbb{E}[\|\boldsymbol{\hat{y}}^{t+1}-\boldsymbol{\hat{y}}^{t}\|^2]\!+\!\frac{8\bar{L}_{f}^2(\lambda_{z}^{t})^{2}}{(\lambda_{y}^{t-1})^2}\mathbb{E}[\|\boldsymbol{\hat{y}}^{t}\!-\!\boldsymbol{\hat{y}}^{t-1}\|^2]\nonumber\\
&\textstyle\quad+4\bar{L}_{f}^2c_{u1}(\lambda_{x}^{t-1})^2\mathbb{E}[\|\boldsymbol{\hat{z}}^{t}-\boldsymbol{\hat{z}}^{t-1}\|^2]+2n\bar{L}_{f}^2(\lambda_{z}^{t})^{2}(\sigma_{x}^{t-1})^2\nonumber\\
&\textstyle\quad+4m\bar{L}_{f}^2c_{u1}r(\lambda_{x}^{t-1})^2(\sigma_{z}^{t-1})^2+16mr\bar{L}_{f}^2\frac{(\lambda_{z}^{t})^{2}(\sigma_{y}^{t})^2}{(\lambda_{y}^{t})^2}\nonumber\\
&\textstyle\quad+4m\bar{L}_{f}^2c_{u2}(\lambda_{z}^{t})^{2}(\lambda_{x}^{t-1})^2+16m(\sigma_{g,0}^2+d_{g}^2)\bar{L}_{f}^2(\lambda_{z}^{t})^{2}.\label{3LApp6}
\end{flalign}

Using an argument similar to the derivation of~\eqref{2LApp14} yields
\begin{equation}
\textstyle(\lambda_{z}^{t}\!-\!\lambda_{z}^{t-1})^2\mathbb{E}[\|\nabla_{y}\boldsymbol{f}^{t-1}(\boldsymbol{x}^{t-1},\boldsymbol{\tilde{y}}^{t-1})\|^2]\!\leq\! \frac{(\sigma_{f}^2+L_{f}^2)mv_{z}^{2}(\lambda_{z}^{0})^2}{t^{2+2v_{z}}}.\nonumber
\end{equation}

Following an argument similar to the derivation of~\eqref{2LApp9}, we  incorporate the preceding inequality and~\eqref{3LApp6} into the second term on the right hand side of~\eqref{3LApp1} to establish Lemma~\ref{AppLemma3}.
\vspace{-0.2em}
\end{proof}
\begin{lemma}\label{consensuslemma}
Under Assumptions~\ref{A1}-\ref{A4}, if the decaying rates of stepsizes satisfy $1>v_{x}>v_{z}>v_{y}>0$, then we have
\begin{flalign}
\textstyle\sum_{i=1}^{m}\mathbb{E}[\|\mathbf{x}_{(i)}^{t}\!-\!\boldsymbol{1}_{m}\otimes \bar{\mathrm{x}}_{(i)}^{t}\|^2]&\textstyle\leq \frac{\bar{C}_{x}}{(t+1)^{2\varsigma_{x}}},\label{ratex}\\
\textstyle\mathbb{E}[\|\boldsymbol{y}^{t+1}-\boldsymbol{y}^{t}-\boldsymbol{1}_{m}\otimes(\bar{y}^{t+1}-\bar{y}^{t})\|^2]&\textstyle\leq\frac{\bar{C}_{y}}{(t+1)^{2\varsigma_{y}}},\label{ratey}\\
\textstyle\mathbb{E}[\|\boldsymbol{z}^{t+1}-\boldsymbol{z}^{t}-\boldsymbol{1}_{m}\otimes(\bar{z}^{t+1}-\bar{z}^{t})\|^2]&\textstyle\leq\frac{\bar{C}_{z}}{(t+1)^{\beta_{z}}},\label{ratez}
\end{flalign}
where $\beta_{z}$ is given by $\beta_{z}=\min\{2v_{z}+2\varsigma_{y}-2v_{y},2\varsigma_{z}\}$ and $\bar{C}_{x}$, $\bar{C}_{y}$, and $\bar{C}_{z}$ are given in~\eqref{4LApp2},~\eqref{4LApp3}, and~\eqref{4LApp4}, respectively.
\vspace{-0.3em}
\end{lemma}
\begin{proof}
We sum both sides of~\eqref{AppLemma1results},~\eqref{AppLemma2results}, and~\eqref{AppLemma3results} to obtain
\begin{equation}
\vspace{-0.2em}
\begin{aligned}
&\textstyle\sum_{i=1}^{m}\!\!\left(\mathbb{E}[\|\mathbf{x}_{(i)}^{t+1}\!-\!\boldsymbol{1}_{m}\otimes \bar{\mathrm{x}}_{(i)}^{t+1}\|^2]\!+\!\frac{|\rho_{2}|}{2}\!\mathbb{E}[\|\mathbf{x}_{(i)}^{t}\!-\!\boldsymbol{1}_{m}\otimes \bar{\mathrm{x}}_{(i)}^{t}\|^2]\right)\nonumber\\
&\textstyle+\left(1-\frac{c_{z2}(\lambda_{z}^{t})^{2}}{(\lambda_{y}^{t})^2}\right)\mathbb{E}[\|\boldsymbol{\hat{y}}^{t+1}-\boldsymbol{\hat{y}}^{t}\|^2]\nonumber\\
&\textstyle+\left(1-\frac{c_{x1}(\lambda_{x}^{t})^2}{(\lambda_{z}^{t})^2}\right)\mathbb{E}[\|\boldsymbol{\hat{z}}^{t+1}-\boldsymbol{\hat{z}}^{t}\|^2]\nonumber\\
&\textstyle\leq \left(1-\frac{|\rho_{2}|}{2}\right)\sum_{i=1}^{m}\mathbb{E}[\|\mathbf{x}_{(i)}^{t}-\boldsymbol{1}_{m}\otimes\bar{\mathrm{x}}_{(i)}^{t}\|^2]\nonumber\\
&\textstyle+(c_{y1}(\lambda_{y}^{t})^{2}+c_{z1}(\lambda_{z}^{t})^{2})\sum_{i=1}^{m}\mathbb{E}[\|\mathbf{x}_{(i)}^{t-1}-\boldsymbol{1}_{m}\otimes \bar{\mathrm{x}}_{(i)}^{t-1}\|^2]\nonumber\\
&\textstyle+\left(1-|\rho_{2}|+\frac{c_{z2}(\lambda_{z}^{t})^{2}}{(\lambda_{y}^{t-1})^2}\right)\mathbb{E}[\|\boldsymbol{\hat{y}}^{t}-\boldsymbol{\hat{y}}^{t-1}\|^2]\nonumber\\
&\textstyle+\left(1\!-\!|\rho_{2}|\!+\!\frac{c_{y2}(\lambda_{y}^{t})^{2}(\lambda_{x}^{t-1})^2}{(\lambda_{z}^{t-1})^2}\!+\!c_{z3}(\lambda_{x}^{t-1})^2\right)\mathbb{E}[\|\boldsymbol{\hat{z}}^{t}-\boldsymbol{\hat{z}}^{t-1}\|^2]\nonumber\\
&\textstyle+c_{x2}(\sigma_{x}^{t})^{2}+c_{y3}(\sigma_{y}^{t-1})^2+c_{z4}\frac{(\lambda_{z}^{t})^{2}(\sigma_{y}^{t})^2}{(\lambda_{y}^{t})^2}+2mr(\sigma_{z}^{t-1})^2.\nonumber
\end{aligned}
\vspace{-0.2em}
\end{equation}

Given that $\lambda_{x}^{t}$, $\lambda_{y}^{t}$, and $\lambda_{z}^{t}$ are decaying sequences and their decaying rates satisfy $v_{x}>v_{z}>v_{y}$, we can choose proper initial stepsizes such that the following inequality holds:
\begin{flalign}
&\textstyle\sum_{i=1}^{m}\!\!\left(\mathbb{E}[\|\mathbf{x}_{(i)}^{t+1}\!-\!\boldsymbol{1}_{m}\otimes \bar{\mathrm{x}}_{(i)}^{t+1}\|^2]\!+\!\frac{|\rho_{2}|}{2}\!\mathbb{E}[\|\mathbf{x}_{(i)}^{t}\!-\!\boldsymbol{1}_{m}\otimes \bar{\mathrm{x}}_{(i)}^{t}\|^2]\right)\nonumber\\
&\textstyle\quad+\left(1-\frac{|\rho_{2}|}{4}\right)\left(\mathbb{E}[\|\boldsymbol{\hat{y}}^{t+1}-\boldsymbol{\hat{y}}^{t}\|^2]+\mathbb{E}[\|\boldsymbol{\hat{z}}^{t+1}-\boldsymbol{\hat{z}}^{t}\|^2]\right)\nonumber\\
&\textstyle\leq \left(1-\frac{|\rho_{2}|}{2}\right)\Big[\sum_{i=1}^{m}\mathbb{E}[\|\mathbf{x}_{(i)}^{t}-\boldsymbol{1}_{m}\otimes \bar{\mathrm{x}}_{(i)}^{t}\|^2]\nonumber\\
&\textstyle\quad+\frac{|\rho_{2}|}{2}\!\sum_{i=1}^{m}\mathbb{E}[\|\mathbf{x}_{(i)}^{t}-\boldsymbol{1}_{m}\otimes \bar{\mathrm{x}}_{(i)}^{t}\|^2]\nonumber\\
&\textstyle\quad+\left(1-\frac{|\rho_{2}|}{4}\right)\left(\mathbb{E}[\|\boldsymbol{\hat{y}}^{t}-\boldsymbol{\hat{y}}^{t-1}\|^2]+\mathbb{E}[\|\boldsymbol{\hat{z}}^{t}-\boldsymbol{\hat{z}}^{t-1}\|^2]\right)\Big]\nonumber\\
&\textstyle \quad+\frac{c_{x2}+c_{y3}+c_{z4}+2mr}{(t+1)^{\min\{2v_{z}+2\varsigma_{y}-2v_{y},2\varsigma_{x},2\varsigma_{y},2\varsigma_{z}\}}}.\label{bili}
\end{flalign}
Applying Lemma 11 in~\cite{zijiGT} to~\eqref{bili}, for any $t>0$, we have
\begin{equation}
\vspace{-0.2em}
\begin{aligned}
&\textstyle\sum_{i=1}^{m}\!\!\left(\mathbb{E}[\|\mathbf{x}_{(i)}^{t+1}\!-\!\boldsymbol{1}_{m}\otimes \bar{\mathrm{x}}_{(i)}^{t+1}\|^2]\!+\!\frac{|\rho_{2}|}{2}\!\mathbb{E}[\|\mathbf{x}_{(i)}^{t}\!-\!\boldsymbol{1}_{m}\otimes \bar{\mathrm{x}}_{(i)}^{t}\|^2]\right)\\
&\textstyle+\left(1\!-\!\frac{|\rho_{2}|}{4}\right)\left(\mathbb{E}[\|\boldsymbol{\hat{y}}^{t+1}-\boldsymbol{\hat{y}}^{t}\|^2]\!+\!\mathbb{E}[\|\boldsymbol{\hat{z}}^{t+1}-\boldsymbol{\hat{z}}^{t}\|^2]\right)\!\leq\! \frac{c_{1}}{(t+1)^{\beta_{0}}},\label{4LApp0}
\end{aligned}
\end{equation}
where $\beta_{0}=\min\{2v_{z}+2\varsigma_{y}-2v_{y}, 2\varsigma_{x},2\varsigma_{y}, 2\varsigma_{z}\}$ and $c_{1}\!=\!b_{0}(\frac{4\beta_{0}}{e\ln(\frac{4}{4-|\rho_{2}|})})^{\beta_{0}}(\frac{a_{0}(2-|\rho_{2}|)}{2b_{0}}\!+\!\frac{4}{|\rho_{2}|})$ with $a_{0}\triangleq\sum_{i=1}^{m}\!(\mathbb{E}[\|\mathbf{x}_{(i)}^{1}\!-\!\boldsymbol{1}_{m}\otimes \bar{\mathrm{x}}_{(i)}^{1}\|^2]+\frac{|\rho_{2}|}{2}\mathbb{E}[\|\mathbf{x}_{(i)}^{0}\!-\!\boldsymbol{1}_{m}\otimes \bar{\mathrm{x}}_{(i)}^{0}\|^2])+(1-\frac{|\rho_{2}|}{4})(\mathbb{E}[\|\boldsymbol{\hat{y}}^{1}-\boldsymbol{\hat{y}}^{0}\|^2]+\mathbb{E}[\|\boldsymbol{\hat{z}}^{1}-\boldsymbol{\hat{z}}^{0}\|^2])$ and $b_{0}=c_{x2}+c_{y3}+c_{z4}+2mr$.

Inequality~\eqref{4LApp0} implies $\mathbb{E}[\|\boldsymbol{\hat{z}}^{t+1}-\boldsymbol{\hat{z}}^{t}\|^2]\leq \frac{4c_{1}}{(4-|\rho_{2}|)(t+1)^{\beta_{0}}}$. Substituting this inequality into~\eqref{AppLemma1results} yields
\begin{equation}
\begin{aligned}
&\textstyle\sum_{i=1}^{m}\mathbb{E}[\|\mathbf{x}_{(i)}^{t+1}-\boldsymbol{1}_{m}\otimes \bar{\mathrm{x}}_{(i)}^{t+1}\|^2]\\
&\textstyle\leq(1-|\rho_{2}|)\sum_{i=1}^{m}\mathbb{E}[\|\mathbf{x}_{(i)}^{t}-\boldsymbol{1}_{m}\otimes\bar{\mathrm{x}}_{(i)}^{t}\|^2]\\
&\textstyle\quad+\frac{4c_{1}c_{x1}(\lambda_{x}^{0})^2}{(4-|\rho_{2}|)(\lambda_{z}^{0})^2(t+1)^{\beta_{0}+2v_{x}-2v_{z}}}+\frac{c_{x2}\sigma_{x}^2}{(t+1)^{2\varsigma_{x}}}.
\label{4LApp1}
\end{aligned}
\end{equation}
Given $v_{x}\!>\!v_{z}\!>\!v_{y}$ from the lemma statement and $v_{x}-v_{z}>\varsigma_{x}$ from Assumption~\ref{A4}, we have $\min\{\beta_{0}+2v_{x}-2v_{z},2\varsigma_{x}\}=2\varsigma_{x}$.

Applying Lemma 11 in~\cite{zijiGT} to~\eqref{4LApp1}, we arrive at
\begin{equation}
\textstyle\sum_{i=1}^{m}\mathbb{E}[\|\mathbf{x}_{(i)}^{t}\!-\!\boldsymbol{1}_{m}\otimes \bar{\mathrm{x}}_{(i)}^{t}\|^2]\leq \frac{\bar{C}_{x}}{(t+1)^{2\varsigma_{x}}},\label{4LApp2}
\end{equation}
where $\bar{C}_{x}$ is given by $\bar{C}_{x}\!=\!b_{1}(\frac{8\varsigma_{x}}{e\ln(\frac{2}{2-|\rho_{2}|})})^{2\varsigma_{x}}(\frac{a_{1}(1-|\rho_{2}|)}{b_{1}}+\frac{2}{|\rho_{2}|})$ with $a_{1}\!=\!\sum_{i=1}^{m}\mathbb{E}[\|\mathbf{x}_{(i)}^{0}-\boldsymbol{1}_{m}\otimes \bar{\mathrm{x}}_{(i)}^{0}\|^2]$, $b_{1}\!=\!\frac{4c_{1}c_{x1}(\lambda_{x}^{0})^2}{(4-|\rho_{2}|)(\lambda_{z}^{0})^2}+c_{x2}\sigma_{x}^{2}$, $c_{1}$ given in~\eqref{4LApp0}, and $c_{x1}$ and $c_{x2}$ given in~\eqref{AppLemma1results}.

By substituting~\eqref{4LApp2} and $\mathbb{E}[\|\boldsymbol{\hat{z}}^{t}-\boldsymbol{\hat{z}}^{t-1}\|^2]\!\leq\! \frac{4c_{1}}{(4-|\rho_{2}|)t^{\beta_{0}}}$ into~\eqref{AppLemma2results} and using Lemma 11 in~\cite{zijiGT}, we have
\begin{equation}
\textstyle \mathbb{E}[\|\boldsymbol{\hat{y}}^{t+1}-\boldsymbol{\hat{y}}^{t}\|^2]\leq\frac{\bar{C}_{y}}{(t+1)^{2\varsigma_{y}}},\label{4LApp3}
\end{equation}
with $\bar{C}_{y}=b_{2}(\frac{8\varsigma_{y}}{e\ln(\frac{2}{2-|\rho_{2}|})})^{2\varsigma_{y}}(\frac{a_{2}(1-|\rho_{2}|)}{b_{2}}+\frac{2}{|\rho_{2}|})$, $a_{2}\!=\!\textstyle \mathbb{E}[\|\boldsymbol{\hat{y}}^{1}-\boldsymbol{\hat{y}}^{0}\|^2]$, $b_{2}\!=\!\bar{C}_{x}c_{y1}(\lambda_{y}^{0})^{2}+c_{y3}\sigma_{y}^2+\frac{4c_{1}c_{y2}(\lambda_{y}^{0})^{2}(\lambda_{x}^{0})^2}{(4-|\rho_{2}|)(\lambda_{z}^{0})^2}$, $c_{1}$ given in~\eqref{4LApp0}, $\bar{C}_{x}$ given in~\eqref{4LApp2}, and $c_{y1}$ through $c_{y3}$ given in~\eqref{AppLemma2results}.

Substituting~\eqref{4LApp2},~\eqref{4LApp3}, and the relation $\mathbb{E}[\|\boldsymbol{\hat{z}}^{t}-\boldsymbol{\hat{z}}^{t-1}\|^2]\leq \frac{4c_{1}}{(4-|\rho_{2}|)t^{\beta_{0}}}$ into~\eqref{AppLemma3results}, we obtain
\begin{equation}
\textstyle \mathbb{E}[\|\boldsymbol{\hat{z}}^{t+1}-\boldsymbol{\hat{z}}^{t}\|^2]\leq\frac{\bar{C}_{z}}{(t+1)^{\min\{2v_{z}+2\varsigma_{y}-2v_{y},2\varsigma_{z}\}}}=\frac{\bar{C}_{z}}{(t+1)^{\beta_{z}}},\label{4LApp4}
\end{equation}
with $\bar{C}_{z}=b_{3}(\frac{4\beta_{z}}{e\ln(\frac{2}{2-|\rho_{2}|})})^{4\beta_{z}}(\frac{a_{3}(1-|\rho_{2}|)}{b_{3}}+\frac{2}{|\rho_{2}|})$, $a_{3}\!=\!\textstyle \mathbb{E}[\|\boldsymbol{\hat{z}}^{1}-\boldsymbol{\hat{z}}^{0}\|^2]$, $b_{3}\!=\!\bar{C}_{x}c_{z1}(\lambda_{z}^{0})^{2}\!+\!\frac{2\bar{C}_{y}c_{z2}(\lambda_{z}^{0})^{2}}{(\lambda_{y}^{0})^2}\!+\!\frac{4c_{1}c_{z3}(\lambda_{x}^{0})^2}{(4-|\rho_{2}|)}+\frac{c_{z4}(\lambda_{z}^{0})^2\sigma_{y}^2}{(\lambda_{y}^{0})^2}+2mr\sigma_{z}^2$, $c_{1}$ given in~\eqref{4LApp0}, $\bar{C}_{x}$ given in~\eqref{4LApp2}, $\bar{C}_{y}$ given in~\eqref{4LApp3}, and $c_{z1}$ through $c_{z3}$ given in~\eqref{AppLemma3results}.
\end{proof}

For the subsequent analysis, we recall the definition of $u_{i}^{t}$ and introduce a new auxiliary variable $\breve{u}_{i}^{t}$ as follows:
\begin{equation}
\begin{aligned}
&\textstyle u_{i}^{t}=\nabla_{x}f_{i}^{t}(x_{i}^{t},\tilde{y}_{i}^{t})+\nabla g_{i}^{t}(x_{i}^{t})\tilde{z}_{i}^{t};\\
&\textstyle \breve{u}_{i}^{t}=\nabla_{x}f_{i}(x_{i}^{t},g(\boldsymbol{x}^{t}))+\nabla g_{i}(x_{i}^{t})\nabla_{y}\bar{f}(\boldsymbol{x}^{t},g(\boldsymbol{x}^{t})).\label{udefinition}
\end{aligned}
\end{equation}

By defining $\mathbf{\bar{x}}^{t}=\col(\bar{\mathrm{x}}_{(1)}^{t},\cdots,\bar{\mathrm{x}}_{(m)}^{t})$, we have
\begin{flalign}
&\textstyle\mathbb{E}[\langle \mathbf{\bar{x}}^{t}-\boldsymbol{x}^{*},\boldsymbol{\breve{u}}^{t}-\boldsymbol{u}^{t}\rangle]=\mathbb{E}[\langle\mathbf{\bar{x}}^{t}-\boldsymbol{x}^{*},\nonumber\\
&\textstyle\nabla_{x}\boldsymbol{f}(\boldsymbol{x}^{t},\boldsymbol{1}_{m}\otimes g(\boldsymbol{x}^{t}))-\nabla_{x}\boldsymbol{f}^{t}(\boldsymbol{x}^{t},\boldsymbol{\tilde{y}}^{t})\rangle]+\mathbb{E}[\langle\mathbf{\bar{x}}^{t}-\boldsymbol{x}^{*},\nonumber\\
&\textstyle\nabla \boldsymbol{g}(\boldsymbol{x}^{t})(\boldsymbol{1}_{m}\otimes \nabla_{y}\bar{f}(\boldsymbol{x}^{t},\boldsymbol{1}_{m}\otimes g(\boldsymbol{x}^{t})))-\nabla \boldsymbol{g}^{t}(\boldsymbol{x}^{t})\boldsymbol{\tilde{z}}^{t}\rangle].\label{1TApp3}
\end{flalign}

We introduce the following Lemma~\ref{firstlemma} and Lemma~\ref{secondlemma} to characterize the two terms on the right hand side of~\eqref{1TApp3}, respectively. 
\begin{lemma}\label{firstlemma}
Under Assumptions~\ref{A1}-\ref{A4}, if the rates of stepsizes satisfy $1>v_{x}>v_{z}$ and $\frac{1}{2}>v_{z}>v_{y}>0$, and the
rate of DP-noise variance $\varsigma_{x}$ satisfies $\varsigma_{x}>v_{y}-\varsigma_{y}$, then we have
\begin{equation}
\begin{aligned}
&\textstyle\mathbb{E}[\langle\mathbf{\bar{x}}^{t}-\boldsymbol{x}^{*},\nabla_{x}\boldsymbol{f}(\boldsymbol{x}^{t},\boldsymbol{1}_{m}\otimes g(\boldsymbol{x}^{t}))-\nabla_{x}\boldsymbol{f}^{t}(\boldsymbol{x}^{t},\boldsymbol{\tilde{y}}^{t})\rangle]\\
&\textstyle\leq \frac{\kappa}{2}\mathbb{E}[\|\mathbf{\bar{x}}^{t}-\boldsymbol{x}^*\|^2]+\frac{b_{1}}{\kappa(t+1)^{\beta_{1}}}+\frac{b_{2}}{(t+1)^{\beta_{2}}},\label{firstlemmaresult}
\end{aligned}
\end{equation}
for any $\kappa>0$, where $\beta_{1}=\min\{2\varsigma_{x},1\}$ and $\beta_{2}= \varsigma_{x}-v_{y}+\varsigma_{y}$. Here, $b_{1}$ and $b_{2}$ are given by $b_{1}=4d_{0}+d_{5}$ and $b_{2}=d_{3}$, respectively, with $d_{0}$ through $d_{5}$ defined as follows:
\begin{align}
d_{0} &\textstyle = \frac{4\bar{L}_{f}^2\mathbb{E}[\|\boldsymbol{\hat{y}}^{1}-\boldsymbol{\hat{y}}^{0}\|^2]}{(\lambda_{y}^{0})^2(\ln((1-|\rho_{2}|)^{2})e)^2};~d_{1}\!\!=\! n\sigma_{x}^2\!+\!\frac{(\lambda_{x}^{0})^2(c_{u1}(\bar{C}_{z}+r\sigma_{z}^2)+c_{u2})}{m(\lambda_{z}^{0})^2}; \notag \\
d_{2}  &\textstyle= \frac{2^{4\varsigma_{y}+3}mr\bar{L}_{f}^2\sigma_{y}^2}{(\ln(1-|\rho_{2}|)e)^2|\rho_{2}|};~\quad
d_{3}= \frac{2^{\beta_{2}+2+\varsigma_{x}}\sqrt{d_{1}d_{2}}}{\lambda_{y}^{0}(\ln(\sqrt{1-|\rho_{2}|})e)^2(1-\sqrt{1-|\rho_{2}|})};\notag \\
d_{4}  &\textstyle=4L_{g}^2(\lambda_{x}^{0})^{2}(\lambda_{y}^{0})^{2}(\lambda_{z}^{0})^{-2}(\bar{C}_{z}c_{u1}+mrc_{u1}\sigma_{z}^2) \notag \\
& \quad +2L_{g}^2(\lambda_{y}^{0})^{2}(2\bar{C}_{x}+n\sigma_{x}^2+2mc_{u2}(\lambda_{x}^{0})^{2})\notag \\
&\quad+2m(2+v_{y}^{2}) (\lambda_{y}^{0})^{2}\sigma_{g,0}^2
+2md_{g}^2v_{y}^{2}(\lambda_{y}^{0})^2; \notag \\
d_{5} &\textstyle= \frac{2^{2\min\{2v_{x}+2\varsigma_{y},2v_{y}+1,2v_{y}+2\varsigma_{x}\}+2}d_{4}\bar{L}_{f}^2}{(\lambda_{y}^{0})^2(\ln(1-|\rho_{2}|)e)^2|\rho_{2}|}.\label{parameter1}
\end{align}
\end{lemma}
\begin{proof}
We characterize the first term on the right hand side of~\eqref{1TApp3} by using the following decomposition:
\begin{equation}
\begin{aligned}
&\textstyle\nabla_{x}f_{i}(x_{i}^{t},g(\boldsymbol{x}^{t}))-\nabla_{x}f_{i}^{t}(x_{i}^{t},\tilde{y}_{i}^{t})\\
&\textstyle=\nabla_{x}f_{i}(x_{i}^{t},g(\boldsymbol{x}^{t}))-\nabla_{x}f_{i}(x_{i}^{t},g^{t}(\boldsymbol{x}^{t}))\\
&\textstyle\quad+\nabla_{x}f_{i}(x_{i}^{t},g^{t}(\boldsymbol{x}^{t}))-\nabla_{x}f_{i}(x_{i}^{t},\frac{1}{m}\sum_{i=1}^{m}\tilde{y}_{i}^{t})\\
&\textstyle\quad+\nabla_{x}f_{i}(x_{i}^{t},\frac{1}{m}\sum_{i=1}^{m}\tilde{y}_{i}^{t})-\nabla_{x}f_{i}(x_{i}^{t},\tilde{y}_{i}^{t})\\
&\textstyle\quad+\nabla_{x}f_{i}(x_{i}^{t},\tilde{y}_{i}^{t})-\nabla_{x}f_{i}^{t}(x_{i}^{t},\tilde{y}_{i}^{t}).\label{1TApp4}
\end{aligned}
\end{equation}
According to Assumption~\ref{A2}, we have $\mathbb{E}[g^{t}(\boldsymbol{x}^{t})-g(\boldsymbol{x}^{t})]=0$. Moreover, Assumption~\ref{A4} imply $\mathbb{E}[\frac{1}{m}\sum_{i=1}^{m}\tilde{y}_{i}^{t}- g^{t}(\boldsymbol{x}^{t})]=0$. Combining these relations and~\eqref{1TApp4} leads to
\begin{equation}
\begin{aligned}
&\textstyle\mathbb{E}[\langle\mathbf{\bar{x}}^{t}-\boldsymbol{x}^{*},\nabla_{x}\boldsymbol{f}(\boldsymbol{x}^{t},\boldsymbol{1}_{m}\otimes g(\boldsymbol{x}^{t}))-\nabla_{x}\boldsymbol{f}^{t}(\boldsymbol{x}^{t},\boldsymbol{\tilde{y}}^{t})\rangle]\\
&\textstyle=\mathbb{E}[\langle\mathbf{\bar{x}}^{t}-\boldsymbol{x}^{*},\nabla_{x}\boldsymbol{f}(\boldsymbol{x}^{t},\boldsymbol{1}_{m}\otimes\frac{1}{m}\sum_{i=1}^{m}\tilde{y}_{i}^{t})\!-\!\nabla_{x}\boldsymbol{f}(\boldsymbol{x}^{t},\boldsymbol{\tilde{y}}^{t})\rangle],\label{1TApp6}
\end{aligned}
\end{equation}
where in the derivation we have used Assumption~\ref{A2}.

By using the mean-value theorem, we have
\begin{equation}
\vspace{-0.2em}
\begin{aligned}
&\textstyle\nabla_{x}\boldsymbol{f}(\boldsymbol{x}^{t},\boldsymbol{1}_{m}\otimes\frac{1}{m}\sum_{i=1}^{m}\tilde{y}_{i}^{t})-\nabla_{x}\boldsymbol{f}(\boldsymbol{x}^{t},\boldsymbol{\tilde{y}}^{t})\\
&\textstyle=\frac{\partial^{2}\boldsymbol{f}(\boldsymbol{x}^{t},\boldsymbol{\tau}^{t})}{\partial x\partial y}(\boldsymbol{1}_{m}\otimes\frac{1}{m}\sum_{i=1}^{m}\tilde{y}_{i}^{t}-\boldsymbol{\tilde{y}}^{t})\\
&\textstyle\triangleq-M_{\boldsymbol{\tau}^{t}}(\frac{\boldsymbol{\hat{y}}^{t+1}-\boldsymbol{\hat{y}}^{t}}{\lambda_{y}^{t}}),\label{Mtau}
\end{aligned}
\vspace{-0.2em}
\end{equation}
where $\boldsymbol{\tau}^{t}\triangleq c(\boldsymbol{1}_{m}\otimes\frac{1}{m}\sum_{i=1}^{m}\tilde{y}_{i}^{t})+(1-c)\boldsymbol{\tilde{y}}^{t}$ for any $c\in(0,1)$. 

Substituting~\eqref{Mtau} into~\eqref{1TApp6} yields
\begin{equation}
\begin{aligned}
&\textstyle\mathbb{E}[\langle\mathbf{\bar{x}}^{t}-\boldsymbol{x}^{*},\nabla_{x}\boldsymbol{f}(\boldsymbol{x}^{t},\boldsymbol{1}_{m}\otimes g(\boldsymbol{x}^{t}))-\nabla_{x}\boldsymbol{f}^{t}(\boldsymbol{x}^{t},\boldsymbol{\tilde{y}}^{t})\rangle]\\
&\textstyle=\mathbb{E}[\langle \mathbf{\bar{x}}^{t}-\boldsymbol{x}^*,-\frac{1}{\lambda_{y}^{t}}M_{\boldsymbol{\tau}^{t}}(\boldsymbol{\hat{y}}^{t+1}-\boldsymbol{\hat{y}}^{t})\rangle].\label{1TApp8}
\end{aligned}
\end{equation}
By using~\eqref{2LApp5} and the definition $\boldsymbol{\hat{y}}^{t}\!=\!\boldsymbol{y}^{t}\!-\!\boldsymbol{1}_{m}\otimes \bar{y}^{t}$, we have
\begin{equation}
\vspace{-0.2em}
\begin{aligned}
&\textstyle \boldsymbol{\hat{y}}^{t+1}-\boldsymbol{\hat{y}}^{t}=(\tilde{W}\otimes I_{r})(\boldsymbol{\hat{y}}^{t}-\boldsymbol{\hat{y}}^{t-1})+K_{1}(\boldsymbol{\chi}^{t}-\boldsymbol{\chi}^{t-1})\\
&\textstyle\quad+K_{2}(\lambda_{y}^{t}\boldsymbol{g}^{t}(\boldsymbol{x}^{t})-\lambda_{y}^{t-1}\boldsymbol{g}^{t-1}(\boldsymbol{x}^{t-1})),\label{1TApp9}
\end{aligned}
\vspace{-0.2em}
\end{equation}
where $K_{1}\triangleq(\Gamma\otimes I_{r})(W^{0}\otimes I_{r})$ and $K_{2}\triangleq\Gamma\otimes I_{r}$.

For the sake of notational simplicity, we define $\Xi_{y}^{t}\triangleq\boldsymbol{\hat{y}}^{t}-\boldsymbol{\hat{y}}^{t-1}$, $\Xi_{\chi}^{t}\triangleq\boldsymbol{\chi}^{t}-\boldsymbol{\chi}^{t-1}$, and $\Xi_{g}^{t}\triangleq\lambda_{y}^{t}\boldsymbol{g}^{t}(\boldsymbol{x}^{t})-\lambda_{y}^{t-1}\boldsymbol{g}^{t-1}(\boldsymbol{x}^{t-1})$. 

By iterating~\eqref{1TApp9} from $1$ to $t$, we have
\begin{equation}
\textstyle\Xi_{y}^{t+1}=(\tilde{W}\otimes I_{r})^{t}\Xi_{y}^{1}+\sum_{p=1}^{t}(\tilde{W}\otimes I_{r})^{t-p}(K_{1}\Xi_{\chi}^{p}+K_{2}\Xi_{g}^{p}).\label{1TApp10}
\end{equation}
Substituting~\eqref{1TApp10} into~\eqref{1TApp8}, we obtain 
\begin{equation}
\vspace{-0.2em}
\begin{aligned}
&\textstyle\mathbb{E}[\langle\mathbf{\bar{x}}^{t}-\boldsymbol{x}^{*},\nabla_{x}\boldsymbol{f}(\boldsymbol{x}^{t},\boldsymbol{1}_{m}\otimes g(\boldsymbol{x}^{t}))-\nabla_{x}\boldsymbol{f}^{t}(\boldsymbol{x}^{t},\boldsymbol{\tilde{y}}^{t})\rangle]\\
&\textstyle=\mathbb{E}[\langle \mathbf{\bar{x}}^{t}-\boldsymbol{x}^*,-\frac{1}{\lambda_{y}^{t}}M_{\boldsymbol{\tau}^{t}}(\tilde{W}\otimes I_{r})^{t}\Xi_{y}^{1}\rangle]\\
&\textstyle\quad+\mathbb{E}[\langle \mathbf{\bar{x}}^{t}-\boldsymbol{x}^*,-\frac{1}{\lambda_{y}^{t}}M_{\boldsymbol{\tau}^{t}}\sum_{p=1}^{t}(\tilde{W}\otimes I_{r})^{t-p}K_{1}\Xi_{\chi}^{p}\rangle]\\
&\textstyle\quad+\mathbb{E}[\langle \mathbf{\bar{x}}^{t}-\boldsymbol{x}^*,-\frac{1}{\lambda_{y}^{t}}M_{\boldsymbol{\tau}^{t}}\sum_{p=1}^{t}(\tilde{W}\otimes I_{r})^{t-p}K_{2}\Xi_{g}^{p}\rangle].\label{1TApp11}
\end{aligned}
\vspace{-0.2em}
\end{equation}

Next, we characterize each item on the right hand side of~\eqref{1TApp11}.

By using Lemma~\ref{lemmae} and the relations  $\|\tilde{W}\|<1-|\rho_{2}|$ and $\max_{i\in[m],t>0}\mathbb{E}[\|M_{\boldsymbol{\tau}^{t}}\|]\leq \bar{L}_{f}$, for any $\kappa>0$, the first term on the right hand side of~\eqref{1TApp11} satisfies
\begin{equation}
\vspace{-0.2em}
\begin{aligned}
&\textstyle\mathbb{E}[\langle \mathbf{\bar{x}}^{t}-\boldsymbol{x}^*,-\frac{1}{\lambda_{y}^{t}}M_{\boldsymbol{\tau}^{t}}(\tilde{W}\otimes I_{r})^{t}\Xi_{y}^{1}\rangle]\\
&\textstyle\quad\leq \frac{\kappa}{4}\mathbb{E}[\|\mathbf{\bar{x}}^{t}-\boldsymbol{x}^*\|^2]+\frac{4d_{0}}{\kappa(t+1)^{2-2v_{y}}},\label{1TApp12}
\end{aligned}
\vspace{-0.2em}
\end{equation}
where $d_{0}=4\bar{L}_{f}^2\mathbb{E}[\|\boldsymbol{\hat{y}}^{1}-\boldsymbol{\hat{y}}^{0}\|^2](\lambda_{y}^{0})^{-2}(\ln((1-|\rho_{2}|)^{2})e)^{-2}$.

The second term on the right hand side of~\eqref{1TApp11} satisfies
\begin{equation}
\begin{aligned}
&\textstyle\mathbb{E}[\langle \mathbf{\bar{x}}^{t}-\boldsymbol{x}^*,-\frac{1}{\lambda_{y}^{t}}M_{\boldsymbol{\tau}^{t}}\sum_{p=1}^{t}(\tilde{W}\otimes I_{r})^{t-p}K_{1}\Xi_{\chi}^{p}\rangle]\\
&\textstyle=\mathbb{E}[\langle \mathbf{\bar{x}}^{t}-\!\mathbf{\bar{x}}^{t-1},-\frac{1}{\lambda_{y}^{t}}M_{\boldsymbol{\tau}^{t}}\sum_{p=1}^{t}(\tilde{W}\otimes I_{r})^{t-p}K_{1}\Xi_{\chi}^{p}\rangle]\\
&\textstyle+\mathbb{E}[\langle \mathbf{\bar{x}}^{t-1}\!-\!\boldsymbol{x}^*,-\frac{1}{\lambda_{y}^{t}}M_{\boldsymbol{\tau}^{t}}\sum_{p=1}^{t}(\tilde{W}\!\otimes\! I_{r})^{t-p}K_{1}\Xi_{\chi}^{p}\rangle].\label{1TApp13}
\end{aligned}
\end{equation}
Since $\mathbf{\bar{x}}^{t-1}-\boldsymbol{x}^*$ and $\Xi_{\chi}^{t}$ are independent, we have 
$\mathbb{E}[\langle \mathbf{\bar{x}}^{t-1}-\boldsymbol{x}^*,-\frac{1}{\lambda_{y}^{t}}M_{\boldsymbol{\tau}^{t}}(\tilde{W}\!\otimes\! I_{r})^{0}K_{1}\Xi_{\chi}^{t}\rangle]=0$. Hence, the second term on the right hand side of~\eqref{1TApp13} can be rewritten as
\begin{equation}
\begin{aligned}
&\textstyle\mathbb{E}[\langle \mathbf{\bar{x}}^{t-1}\!-\!\boldsymbol{x}^*,-\frac{1}{\lambda_{y}^{t}}M_{\boldsymbol{\tau}^{t}}\sum_{p=1}^{t}(\tilde{W}\!\otimes\! I_{r})^{t-p}K_{1}\Xi_{\chi}^{p}\rangle]\\
&\textstyle=\mathbb{E}[\langle \mathbf{\bar{x}}^{t-1}\!-\!\boldsymbol{x}^*,\\
&\textstyle\quad-\frac{1}{\lambda_{y}^{t}}M_{\boldsymbol{\tau}^{t}}(\tilde{W}\!\otimes\! I_{r})\sum_{p=1}^{t-1}(\tilde{W}\!\otimes\! I_{r})^{t-1-p}K_{1}\Xi_{\chi}^{p}\rangle].\label{1TApp14}
\end{aligned}
\end{equation}
Applying an argument similar to the derivation of~\eqref{1TApp14}, we iterate~\eqref{1TApp13} from $0$ to $t$ to obtain
\begin{equation}
\begin{aligned}
&\textstyle\mathbb{E}[\langle \mathbf{\bar{x}}^{t}-\boldsymbol{x}^*,-\frac{1}{\lambda_{y}^{t}}M_{\boldsymbol{\tau}^{t}}\sum_{p=1}^{t}(\tilde{W}\otimes I_{r})^{t-p}K_{1}\Xi_{\chi}^{p}\rangle]\\
&\textstyle=\mathbb{E}[\langle \mathbf{\bar{x}}^{t}-\mathbf{\bar{x}}^{t-1},-\frac{1}{\lambda_{y}^{t}}M_{\boldsymbol{\tau}^{t}}\sum_{p=1}^{t}(\tilde{W}\otimes I_{r})^{t-p}K_{1}\Xi_{\chi}^{p}\rangle]\\
&\textstyle\quad+\mathbb{E}[\langle \mathbf{\bar{x}}^{t-1}-\mathbf{\bar{x}}^{t-2},\\
&\textstyle\quad-\frac{1}{\lambda_{y}^{t}}M_{\boldsymbol{\tau}^{t}}(\tilde{W}\otimes I_{r})\sum_{p=1}^{t-1}(\tilde{W}\otimes I_{r})^{t-1-p}K_{1}\Xi_{\chi}^{p}\rangle]\\
&\textstyle\quad~~\vdots\\
&\textstyle\quad+\mathbb{E}[\langle \mathbf{\bar{x}}^{1}-\mathbf{\bar{x}}^{0},-\frac{1}{\lambda_{y}^{t}}M_{\boldsymbol{\tau}^{t}}(\tilde{W}\otimes I_{r})^{t-1}K_{1}\Xi_{\chi}^{1}\rangle],\label{1TApp15}
\end{aligned}
\end{equation}
where we used $\mathbb{E}[\langle \mathbf{\bar{x}}^{0}\!-\!\boldsymbol{x}^{*},-\frac{1}{\lambda_{y}^{t}}M_{\boldsymbol{\tau}^{t}}(\tilde{W}\!\otimes\! I_{r})^{t-1}K_{1}\Xi_{\chi}^{1}\rangle]=0$ based on  $\mathbf{\bar{x}}^{0}-\boldsymbol{x}^{*}$ and $\Xi_{\chi}^{1}$ being independent.

The Cauchy-Schwartz inequality implies that $\mathbb{E}[\langle a,b\rangle]\leq \sqrt{\mathbb{E}[\|a\|^2]\mathbb{E}[\|b\|^2]}$ holds for any random variables $a,b\in\mathbb{R}^{n}$. Hence, the first term on the right hand side of~\eqref{1TApp15} satisfies
\begin{equation}
\begin{aligned}
&\textstyle \mathbb{E}[\langle \mathbf{\bar{x}}^{t}-\mathbf{\bar{x}}^{t-1},-\frac{1}{\lambda_{y}^{t}}M_{\boldsymbol{\tau}^{t}}\sum_{p=1}^{t}(\tilde{W}\otimes I_{r})^{t-p}K_{1}\Xi_{\chi}^{p}\rangle]\\
&\textstyle\leq \sqrt{\mathbb{E}[\|\mathbf{\bar{x}}^{t}-\mathbf{\bar{x}}^{t-1}\|^2]}\\
&\textstyle\quad\times\frac{1}{\lambda_{y}^{t}}\sqrt{\mathbb{E}[\|M_{\boldsymbol{\tau}^{t}}\sum_{p=1}^{t}(\tilde{W}\otimes I_{r})^{t-p}K_{1}\Xi_{\chi}^{p}\|^2]}.\label{1TApp16}
\end{aligned}
\end{equation}

By defining $\mathbf{\bar{x}}^{t}=\col(\bar{x}_{(1)}^{t},\cdots,\bar{x}_{(m)}^{t})\in\mathbb{R}^{n}$ and using the dynamics of $x_{i}^{t}$ in Algorithm~\ref{algorithm1}, we have
\begin{equation}
\textstyle\mathbb{E}[\|\mathbf{\bar{x}}^{t}-\mathbf{\bar{x}}^{t-1}\|^2]\leq n(\sigma_{x}^{t-1})^2+\frac{(\lambda_{x}^{t-1})^2}{m^2}\mathbb{E}[\|\boldsymbol{u}^{t-1}\|^2].\label{1TApp17}
\end{equation}
Substituting~\eqref{1AppL9} and~\eqref{ratez} into~\eqref{1TApp17}, we have
\begin{equation}
\begin{aligned}
&\textstyle \mathbb{E}[\|\mathbf{\bar{x}}^{t}-\mathbf{\bar{x}}^{t-1}\|^2]\leq\frac{n\sigma_{x}^2}{t^{2\varsigma_{x}}}+\frac{\bar{C}_{z}c_{u1}(\lambda_{x}^{0})^2}{m^2(\lambda_{z}^{0})^2t^{\beta_{z}+2v_{x}-2v_{z}}}\\
&\textstyle\quad+\frac{c_{u1}r(\lambda_{x}^{0})^2\sigma_{z}^2}{m(\lambda_{z}^{0})^2t^{2v_{x}+2\varsigma_{z}-2v_{z}}}+\frac{c_{u2}(\lambda_{x}^{0})^2}{mt^{2v_{x}}}\leq \frac{d_{1}}{t^{2\varsigma_{x}}}\leq \frac{2^{2\varsigma_{x}}d_{1}}{(t+1)^{2\varsigma_{x}}},\label{1TApp18}
\end{aligned}
\end{equation}
with $d_{1}=n\sigma_{x}^2+m^{-1}(\lambda_{x}^{0})^2(\lambda_{z}^{0})^{-2}(c_{u1}(\bar{C}_{z}+r\sigma_{z}^2)+c_{u2})$. Here, in the derivation, we have used the definition $\beta_{z}=\min\{2v_{z}+2\varsigma_{y}-2v_{y},2\varsigma_{z}\}$ and the relation $v_{x}-v_{y}>v_{x}-v_{z}>\varsigma_{x}$.

Using $\|\tilde{W}\|\leq 1-|\rho_{2}|$, $\max_{i\in[m],t>0}\mathbb{E}[\|M_{\boldsymbol{\tau}^{t}}\|]\leq \bar{L}_{f}$, and $\|K_{1}\|\leq 1$ yields that the second term on the right hand side of~\eqref{1TApp16} satisfies
\begin{equation}
\begin{aligned}
&\textstyle\mathbb{E}[\|M_{\boldsymbol{\tau}^{t}}\sum_{p=1}^{t}(\tilde{W}\otimes I_{r})^{t-p}K_{1}\Xi_{\chi}^{p}\|^2]\\
&\textstyle\leq 2mr\bar{L}_{f}^2\sum_{p=1}^{t}((1-|\rho_{2}|)^{2})^{t-p}(\sigma_{y}^{p-1})^2.\label{1TApp19}
\end{aligned}
\end{equation}
Combining Lemma \ref{lemmaexp} and~\eqref{1TApp19} and using $(t+1)^{2\varsigma_{y}}\leq 2^{2\varsigma_{y}}t^{2\varsigma_{y}}$ for any $t>0$, we have
\begin{equation}
\begin{aligned}
&\textstyle\mathbb{E}[\|M_{\boldsymbol{\tau}^{t}}\sum_{p=1}^{t}(\tilde{W}\otimes I_{r})^{t-p}K_{1}\Xi_{\chi}^{p}\|^2]\\
&\textstyle\leq 2^{2\varsigma_{y}+1}mr\bar{L}_{f}^2\sigma_{y}^2\sum_{p=1}^{t}\frac{((1-|\rho_{2}|)^{2})^{t-p}}{(p+1)^{2\varsigma_{y}}}\leq \frac{d_{2}}{(t+1)^{2\varsigma_{y}}},\label{1TApp20}
\end{aligned}
\end{equation}
with $d_{2}=2^{4\varsigma_{y}+3}(\ln(1-|\rho_{2}|)e)^{-2}|\rho_{2}|^{-1}mr\bar{L}_{f}^2\sigma_{y}^2$.

Incorporating~\eqref{1TApp18} and~\eqref{1TApp20} into~\eqref{1TApp16}, we obtain
\begin{equation}
\begin{aligned}
\textstyle\mathbb{E}[\langle \mathbf{\bar{x}}^{t}\!-\!\mathbf{\bar{x}}^{t-1},-\frac{M_{\boldsymbol{\tau}^{t}}}{\lambda_{y}^{t}}\sum_{p=1}^{t}(\tilde{W}\otimes I_{r})^{t-p}K_{1}\Xi_{\chi}^{p}\rangle]\leq \frac{\sqrt{2^{2\varsigma_{x}}d_{1}d_{2}}}{\lambda_{y}^{0}(t+1)^{\beta_{2}}},\nonumber
\end{aligned}
\end{equation}
with $\beta_{2}=\varsigma_{x}-v_{y}+\varsigma_{y}$.

By applying an argument similar to the derivation of the preceding inequality to the remaining terms on the right hand side of~\eqref{1TApp15}, we have that the second term on the right hand side of~\eqref{1TApp11} satisfies 
\begin{equation}
\begin{aligned}
&\textstyle \mathbb{E}[\langle \mathbf{\bar{x}}^{t}-\boldsymbol{x}^*,-\frac{1}{\lambda_{y}^{t}}M_{\boldsymbol{\tau}^{t}}\sum_{p=1}^{t}(\tilde{W}\otimes I_{r})^{t-p}K_{1}\Xi_{\chi}^{p}\rangle]\\
&\textstyle\leq \sum_{p=1}^t\frac{\sqrt{2^{2\varsigma_{x}}d_{1}d_{2}}(1-|\rho_{2}|)^{t-p}}{\lambda_{y}^{0}(p+1)^{\beta_{2}}}\leq\frac{d_{3}}{(t+1)^{\beta_{2}}},\label{1TApp22}
\end{aligned}
\end{equation}
where $\beta_{2}$ is given by $\beta_{2}=\varsigma_{x}-v_{y}+\varsigma_{y}$ and $d_{3}$ is given by $d_{3}=2^{\beta_{2}+2+\varsigma_{x}}(\lambda_{y}^{0})^{-1}\sqrt{d_{1}d_{2}}(\ln(\sqrt{1-|\rho_{2}|})e)^{-2}(1\!-\!\sqrt{1-|\rho_{2}|})^{-1}$.

By using the Young's inequality,  for any $\kappa>0$, the third term on the right hand side of~\eqref{1TApp11} satisfies 
\begin{equation}
\begin{aligned}
&\textstyle\mathbb{E}[\langle \mathbf{\bar{x}}^{t}-\boldsymbol{x}^*,-\frac{1}{\lambda_{y}^{t}}M_{\boldsymbol{\tau}^{t}}\sum_{p=1}^{t}(\tilde{W}\otimes I_{r})^{t-p}K_{2}\Xi_{g}^{p}\rangle]\\
&\textstyle\leq \frac{\kappa}{4}\mathbb{E}[\|\mathbf{\bar{x}}^{t}-\boldsymbol{x}^*\|^2]+\frac{\bar{L}_{f}^2\sum_{p=1}^{t}((1-|\rho_{2}|)^{2})^{t-p}\mathbb{E}[\|\Xi_{g}^{p}\|^2]}{\kappa(\lambda_{y}^{t})^2}.\label{1TApp23}
\end{aligned}
\end{equation}

According to the definition of $\Xi_{g}^{t}$, we combine~\eqref{ratex} and~\eqref{ratez} with~\eqref{2LApp9} to obtain
\begin{flalign}
&\textstyle\mathbb{E}[\|\Xi_{g}^{p}\|^2]\leq \frac{4\bar{C}_{x}L_{g}^2(\lambda_{y}^{0})^{2}}{p^{2v_{y}+2\varsigma_{x}}}+\frac{4\bar{C}_{z}L_{g}^2c_{u1}(\lambda_{y}^{0})^{2}(\lambda_{x}^{0})^2}{(\lambda_{z}^{0})^2p^{\beta_{z}+2v_{y}+2v_{x}-2v_{z}}}\nonumber\\
&\textstyle+\frac{4m\sigma_{g,0}^2(\lambda_{y}^{0})^{2}}{p^{1+2v_{y}}}+\frac{2nL_{g}^2(\lambda_{y}^{0})^{2}\sigma_{x}^2}{p^{2v_{y}+2\varsigma_{x}}}+\frac{4mL_{g}^2c_{u1}r(\lambda_{y}^{0})^{2}(\lambda_{x}^{0})^2\sigma_{z}^2}{(\lambda_{z}^{0})^2p^{2v_{y}+2v_{x}+2\varsigma_{z}-2v_{z}}}\nonumber\\
&\textstyle\!+\!\frac{4mL_{g}^2c_{u2}(\lambda_{y}^{0})^{2}(\lambda_{x}^{0})^2}{p^{2v_{y}+2v_{x}}}\!+\!\frac{2(\sigma_{g,0}^2+d_{g}^2)mv_{y}^{2}(\lambda_{y}^{0})^2}{p^{2+2v_{y}}}\!\leq\! \frac{2^{\alpha_{1}}d_{4}}{(p+1)^{\alpha_{1}}},\label{gammag}
\end{flalign}
where $\alpha_{1}$ is given by $\alpha_{1}=\min\{2v_{x}+2\varsigma_{y},2v_{y}+1,2v_{y}+2\varsigma_{x}\}$ and $d_{4}=4L_{g}^2(\lambda_{x}^{0})^{2}(\lambda_{y}^{0})^{2}(\lambda_{z}^{0})^{-2}(\bar{C}_{z}c_{u1}+mrc_{u1}\sigma_{z}^2)+2L_{g}^2(\lambda_{y}^{0})^{2}(2\bar{C}_{x}+n\sigma_{x}^2+2mc_{u2}(\lambda_{x}^{0})^{2})+2m(2+v_{y}^{2}) (\lambda_{y}^{0})^{2}\sigma_{g,0}^2
+2md_{g}^2v_{y}^{2}(\lambda_{y}^{0})^2$.

By substituting~\eqref{gammag} into~\eqref{1TApp23} and using Lemma~\ref{lemmaexp}, we have that the third term on the right hand side of~\eqref{1TApp11} satisfies 
\begin{equation}
\begin{aligned}
&\textstyle\mathbb{E}[\langle \mathbf{\bar{x}}^{t}-\boldsymbol{x}^*,-\frac{1}{\lambda_{y}^{t}}M_{\boldsymbol{\tau}^{t}}\sum_{p=1}^{t}(\tilde{W}\otimes I_{r})^{t-p}K_{2}\Xi_{g}^{p}\rangle]\\
&\textstyle\leq \frac{\kappa}{4}\mathbb{E}[\|\mathbf{\bar{x}}^{t}-\boldsymbol{x}^*\|^2]+\frac{d_{5}}{\kappa(t+1)^{\beta_{1}}},\label{1TApp24}
\end{aligned}
\end{equation}
where $\beta_{1}$ is given by $\beta_{1}\!=\!\min\{2\varsigma_{x},1\}$ and $d_{5}$ is given by $d_{5}\!=\!2^{2\alpha_{1}+2}(\lambda_{y}^{0})^{-1}d_{4}(\ln(\sqrt{1-|\rho_{2}|})e)^{-2}(1\!-\!\sqrt{1-|\rho_{2}|})^{-1}\bar{L}_{f}^2$.

By substituting~\eqref{1TApp12},~\eqref{1TApp22},~\eqref{1TApp24} into~\eqref{1TApp11} and using the relation $2-2v_{y}>1$, we arrive at~\eqref{firstlemmaresult}.
\end{proof}
\begin{lemma}\label{secondlemma}
Under Assumptions~\ref{A1}-\ref{A4}, if the rates of stepsizes satisfy $1>v_{x}>v_{z}$ and $\frac{1}{2}>v_{z}>v_{y}>0$, and the rates of DP-noise variances satisfy $\varsigma_{x}>\max\{v_{y}-\varsigma_{y},v_{z}-\varsigma_{z}\}$, $\varsigma_{y}>v_{y}-\frac{1}{2}$, and $\varsigma_{z}>v_{z}-\frac{1}{2}$, then we have
\begin{flalign}
&\textstyle \mathbb{E}[\langle\mathbf{\bar{x}}^{t}\!-\!\boldsymbol{x}^{*},\nabla \boldsymbol{g}(\boldsymbol{x}^{t})(\boldsymbol{1}\!\otimes\! \nabla_{y}\bar{f}(\boldsymbol{x}^{t},\boldsymbol{1}\!\otimes\! g(\boldsymbol{x}^{t})))\!-\!\nabla \boldsymbol{g}^{t}(\boldsymbol{x}^{t})\boldsymbol{\tilde{z}}^{t}\rangle]\nonumber\\
&\textstyle\leq \frac{(7+2(\sigma_{g,1}^2\!+L_{g}^2))\kappa}{4}\mathbb{E}[\|\mathbf{\bar{x}}^{t}\!-\!\boldsymbol{x}^*\|^2]\!+\!\frac{b_{3}}{\kappa(t+1)^{\beta_{1}}}\!+\!\frac{b_{4}}{(t+1)^{\beta_{3}}},\label{secondlemmaresult}
\end{flalign}
for any $\kappa>0$, where $\beta_{1}=\min\{2\varsigma_{x},1\}$ and $\beta_{3}=\min\{\varsigma_{x}-v_{y}+\varsigma_{y},\varsigma_{x}-v_{z}+\varsigma_{z},\frac{1}{2}-v_{y}+\varsigma_{y},\frac{1}{2}-v_{z}+\varsigma_{z}\}$. Here, the constants $b_{3}$ and $b_{4}$ are given by  $b_{3}=m\sigma_{g,1}^2L_{f}^2+mL_{g}^2\sigma_{f}^2+d_{8}+d_{10}+d_{14}+d_{16}+d_{21}+d_{28}$ and $b_{4}=d_{9}+d_{12}+d_{18}+d_{25}$, respectively, with $d_{6}$ through $d_{28}$ defined as follows:
\begin{align}
&\textstyle \! d_{6}= 2\bar{C}_{x}+\frac{2c_{u1}(\lambda_{x}^{0})^2(\bar{C}_{z}+mr\sigma_{z}^2)}{(\lambda_{z}^{0})^{2}}+2mc_{u2}(\lambda_{x}^{0})^2+n\sigma_{x}^2; \notag \\
&\textstyle\! d_{7}= 2md_{\Omega}^2\sigma_{g,1}^2+2^{2\varsigma_{x}+1}d_{6}d_{\Omega}^2\bar{L}_{g}^2+2^{2\varsigma_{x}+1}d_{1}L_{g}^2; \notag \\
&\textstyle\!  d_{8}= 4d_{0}+d_{5};~d_{9}= \frac{2^{\min\{\frac{1}{2}-v_{y}+\varsigma_{y},\varsigma_{x}-v_{y}+\varsigma_{y}\}+2}\sqrt{d_{7}d_{2}}}{\lambda_{y}^{0}(\ln(\sqrt{1-|\rho_{2}|})e)^2(1-\sqrt{1-|\rho_{2}|})}; \notag \\
&\textstyle\!  d_{10}=\frac{16(\sigma_{g,1}^2+L_{g}^2)\mathbb{E}[\|\boldsymbol{\hat{z}}^{1}-\boldsymbol{\hat{z}}^{0}\|^2]}{(\lambda_{z}^{0})^2(\ln((1-|\rho_{2}|)^{2})e)^2}; \quad
d_{11}= \frac{2^{4\varsigma_{z}+3}mr\sigma_{z}^2}{(\ln(1-|\rho_{2}|)e)^2|\rho_{2}|};\notag \\
&\textstyle\!  d_{12}=\frac{2^{\min\{\frac{1}{2}-v_{z}+\varsigma_{z},\varsigma_{x}-v_{z}+\varsigma_{z}\}+2}\sqrt{d_{7}d_{11}}}{\lambda_{z}^{0}(\ln(\sqrt{1-|\rho_{2}|})e)^2(1-\sqrt{1-|\rho_{2}|})};\notag \\
&\textstyle\!  d_{13}=2(\lambda_{z}^{0})^2(m\sigma_{f}^2+mv_{z}^2L_{f}^2+2^{2\varsigma_{x}}d_{6}\bar{L}_{f}^2);\notag \\
&\textstyle\!  d_{14}=\frac{2^{2\min\{2\varsigma_{x}+2v_{z},1+2v_{z}\}+2}d_{13}(\sigma_{g,1}^2+L_{g}^2)}{(\lambda_{z}^{0})^2(\ln(1-|\rho_{2}|)e)^2|\rho_{2}|}; \notag \\
&\textstyle\!  d_{15}=2(\lambda_{z}^{0})^2(m\sigma_{g,0}^2+2^{2\varsigma_{x}}d_{6}L_{g}^2);\notag \\
&\textstyle\!  d_{16}=\frac{d_{14}d_{15}\bar{L}_{f}^2}{d_{13}};\quad d_{17}=\frac{2^{2\varsigma_{y}+2}mr\sigma_{y}^2(1+2^{2\varsigma_{y}})\bar{L}_{f}^2(\lambda_{z}^{0})^2}{(\lambda_{y}^{0})^2(\ln(1-|\rho_{2}|)e)^2|\rho_{2}|};\notag \\
&\textstyle\!  d_{18}=\frac{2^{\min\{\frac{1}{2}-v_{y}+\varsigma_{y},\varsigma_{x}-v_{y}+\varsigma_{y}\}+2}\sqrt{d_{7}d_{17}}}{\lambda_{z}^{0}(\ln(\sqrt{1-|\rho_{2}|})e)^2(1-\sqrt{1-|\rho_{2}|})};d_{19}\!=\!\frac{16}{(\ln((1-|\rho_{2}|)^2)e)^2};\notag \\
&\textstyle\!  d_{20}=\frac{2^{2v_{z}+4}d_{19}\lambda_{z}^{0} }{(1-|\rho_{2}|)^2(\ln(1-|\rho_{2}|)e)^2|\rho_{2}|}; \notag \\
&\textstyle\!  d_{21}=\frac{d_{20}(\sigma_{g,1}^2+L_{g}^2)\bar{L}_{f}^2\mathbb{E}[\|\frac{1}{\lambda_{y}^{2}}(\boldsymbol{\hat{y}}^{2}-\boldsymbol{\hat{y}}^{1})\!-\frac{1}{\lambda_{y}^{1}}(\boldsymbol{\hat{y}}^{1}-\boldsymbol{\hat{y}}^{0})\|^2]}{(\lambda_{z}^{0})^2};\notag \\
&\textstyle\! d_{22}=\frac{4^{3}m\sigma_{y}^2}{(\lambda_{y}^{0})^2};
d_{23}=\frac{2^{3}(1-|\rho_2|)^{\frac{1}{1-\sqrt{1-|\rho_2|}}-\frac{1}{2}}\bar{L}_{f}\lambda_{z}^{0}\sqrt{d_{22}}}{(1-\sqrt{1-|\rho_2|})(\ln((1-|\rho_{2}|)^{\frac{1}{4}})e)^2(1-(1-|\rho_{2}|)^{\frac{1}{4}})};\notag \\
&\textstyle\! d_{24}=\frac{d_{23}\sqrt{d_{7}}}{\lambda_{z}^{0}\sqrt{1-|\rho_{2}|}};d_{25}\!=\!\frac{\sqrt{d_{7}}d_{23}}{\lambda_{z}^{0}}\!+\!\frac{ 2^{\varsigma_{y}-v_{y}+v_{z}\!+\!\frac{\min\{2\varsigma_{x},1\}}{2}+2}d_{24}}{(\ln((1-|\rho_{2}|)^{\frac{1}{4}})e)^2(1-(1-|\rho_{2}|)^{\frac{1}{4}})};\notag \\
&\textstyle\!  d_{26}=\frac{(1+2^{3})2^{3}d_{3}}{(\lambda_{y}^{0})^2};d_{27}=\frac{d_{26}d_{23}^2(\ln((1-|\rho_{2}|)^{\frac{1}{4}})e)^4)}{2d_{22}(\ln(\sqrt{1-|\rho_{2}|})e)^2(1-\sqrt{1-|\rho_{2}|})};\notag\\
&\textstyle\! d_{28}=d_{27}(\sigma_{g,1}^2+L_{g}^2)\mu^{-1} (\lambda_{z}^{0})^{-2}.
\end{align}
\end{lemma}
\begin{proof}
To characterize the second term on the right hand side of~\eqref{1TApp3}, we use the following decomposition:
\begin{flalign}
&\textstyle\nabla g_{i}(x_{i}^{t})\nabla_{y}\bar{f}(\boldsymbol{x}^{t},\boldsymbol{1}_{m}\otimes g(\boldsymbol{x}^{t}))-\nabla g_{i}^{t}(x_{i}^{t})\tilde{z}_{i}^{t}\nonumber\\
&\textstyle=\left(\nabla g_{i}(x_{i}^{t})-\nabla g_{i}^{t}(x_{i}^{t})\right)\nabla_{y}\bar{f}(\boldsymbol{x}^{t},\boldsymbol{1}_{m}\otimes g(\boldsymbol{x}^{t}))\nonumber\\
&\textstyle+\nabla g_{i}^{t}(x^t_{i(i)})\left(\nabla_{y}\bar{f}^{t}(\boldsymbol{x}^{t},\boldsymbol{1}_{m}\otimes g(\boldsymbol{x}^{t}))-\frac{1}{m}\sum_{i=1}^{m}\tilde{z}_{i}^{t}\right)\nonumber\\
&\textstyle+\nabla g_{i}^{t}(x_{i}^{t})\left(\nabla_{y}\bar{f}(\boldsymbol{x}^{t},\boldsymbol{1}_{m}\!\otimes\! g(\boldsymbol{x}^{t}))\!-\!\nabla_{y}\bar{f}^{t}(\boldsymbol{x}^{t},\boldsymbol{1}_{m}\!\otimes\! g(\boldsymbol{x}^{t}))\right)\nonumber\\
&\textstyle+\nabla g_{i}^{t}(x_{i}^{t})\left(\frac{1}{m}\sum_{i=1}^{m}\tilde{z}_{i}^{t}-\tilde{z}_{i}^{t}\right).\label{2TApp0}
\end{flalign}

By using the Young's inequality, for any $\kappa>0$, the first term on the right hand side of~\eqref{2TApp0} satisfies
\begin{equation}
\begin{aligned}
&\textstyle\mathbb{E}[\langle \mathbf{\bar{x}}^{t}\!-\!\boldsymbol{x}^*,(\nabla \boldsymbol{g}(\boldsymbol{x}^{t})\!-\!\nabla \boldsymbol{g}^{t}(\boldsymbol{x}^{t}))\\
&\textstyle\quad\times(\boldsymbol{1}_{m}\!\otimes\!\nabla_{y}\bar{f}(\boldsymbol{x}^{t},\boldsymbol{1}_{m}\otimes g(\boldsymbol{x}^{t})))\rangle]\\
&\textstyle\leq\frac{\kappa}{4}\mathbb{E}[\| \mathbf{\bar{x}}^{t}-\boldsymbol{x}^*\|^2]+\frac{m\sigma_{g,1}^2L_{f}^2}{\kappa(t+1)}.\label{2TApp1}
\end{aligned}
\end{equation}

Using $\frac{1}{m}\sum_{i=1}^{m}\tilde{z}_{i}^{t}=\frac{1}{\lambda_{z}^{t}}(\bar{\zeta}_{w}^{t}+\lambda_{z}^{t}\nabla_{y}\bar{f}^{t}(\boldsymbol{x}^{t},\boldsymbol{\tilde{y}}^{t}))$ and  $\mathbb{E}[ \bar{\zeta}_{w}^{t}]=0$, the second term on the right hand side of~\eqref{2TApp0} satisfies
\begin{equation}
\begin{aligned}
&\textstyle\mathbb{E}[\langle \mathbf{\bar{x}}^{t}\!-\!\boldsymbol{x}^*,\nabla \boldsymbol{g}^{t}(\boldsymbol{x}^t) (\boldsymbol{1}_{m}\otimes\nabla_{y}\bar{f}^{t}(\boldsymbol{x}^{t},\boldsymbol{1}_{m}\otimes g(\boldsymbol{x}^{t})))\\
&\textstyle\quad-\nabla \boldsymbol{g}^{t}(\boldsymbol{x}^t)(\boldsymbol{1}_{m}\otimes\frac{1}{m}\sum_{i=1}^{m}\tilde{z}_{i}^{t})\rangle]\\
&\textstyle=\mathbb{E}[\langle \nabla \boldsymbol{g}^{t}(\boldsymbol{x}^t)^{\top}(\mathbf{\bar{x}}^{t}\!-\!\boldsymbol{x}^*),\\
&\textstyle \quad\boldsymbol{1}_{m}\otimes \nabla_{y}\bar{f}^{t}(\boldsymbol{x}^{t},\boldsymbol{1}_{m}\!\otimes\! g(\boldsymbol{x}^{t}))-\boldsymbol{1}_{m}\!\otimes\!\nabla_{y}\bar{f}^{t}(\boldsymbol{x}^{t},\boldsymbol{\tilde{y}}^{t})\rangle].
\label{2TApp2}
\end{aligned}
\end{equation}

To characterize the right hand side of~\eqref{2TApp2}, we first estimate an upper bound on $\mathbb{E}[\|\Xi_{x}^{t}-\Xi_{x}^{t-1}\|^2]$ with denoting $\Xi_{x}^{t}\triangleq\nabla \boldsymbol{g}^{t}(\boldsymbol{x}^{t})^{\top}(\mathbf{\bar{x}}^{t}-\boldsymbol{x}^*)$. We have the following decomposition:
\begin{flalign}
&\textstyle \mathbb{E}[\|\Xi_{x}^{t}-\Xi_{x}^{t-1}\|^2]=\mathbb{E}[\|\nabla \boldsymbol{g}^{t}(\boldsymbol{x}^{t})^{\top}(\mathbf{\bar{x}}^{t}-\boldsymbol{x}^*)\nonumber\\
&\textstyle\quad-\nabla \boldsymbol{g}^{t-1}(\boldsymbol{x}^{t-1})^{\top}(\mathbf{\bar{x}}^{t-1}-\boldsymbol{x}^*)\|^2]\nonumber\\
&\textstyle\leq\frac{2m\sigma_{g,1}^2d_{\Omega}^2}{t+1}\!+\!2d_{\Omega}^2\bar{L}_{g}^2\mathbb{E}[\|\boldsymbol{x}^{t}-\boldsymbol{x}^{t-1}\|^2]\nonumber\\
&\textstyle\quad+2L_{g}^2\mathbb{E}[\|\mathbf{\bar{x}}^{t}-\mathbf{\bar{x}}^{t-1}\|^2],\label{2TApp10}
\end{flalign}
where in the derivation we have used that the diameter of the compact set $\Omega$ is denoted as $d_{\Omega}$.

By using the relation $\|\boldsymbol{x}^{t}-\boldsymbol{x}^{t-1}\|^2=\sum_{i=1}^{m}\mathbb{E}[\|x_{i}^{t}-x_{i}^{t-1}\|^2]$ and incorporating~\eqref{ratex} and~\eqref{ratez} into~\eqref{xsubx}, the second term on the right hand side of~\eqref{2TApp10} satisfies
\begin{equation}
\begin{aligned}
&\textstyle \mathbb{E}[\|\boldsymbol{x}^{t}-\boldsymbol{x}^{t-1}\|^2]\leq \frac{2\bar{C}_{x}}{t^{2\varsigma_{x}}}+\frac{2c_{u1}\bar{C}_{z}(\lambda_{x}^{0})^2}{(\lambda_{z}^{0})^2t^{\beta_{z}+2v_{x}-2v_{z}}}\\
&\textstyle\quad+\frac{2mc_{u1}r(\lambda_{x}^{0})^2\sigma_{z}^2}{(\lambda_{z}^{0})^2t^{2v_{x}+2\varsigma_{z}-2v_{z}}}\!+\!\frac{2mc_{u2}(\lambda_{x}^{0})^2}{t^{2v_{x}}}\!+\!\frac{n\sigma_{x}^2}{t^{2\varsigma_{x}}}\leq\frac{2^{2\varsigma_{x}}d_{6}}{(t+1)^{2\varsigma_{x}}},\label{2TApp12}
\end{aligned}
\end{equation}
where the positive constant $d_{6}$ is given by $d_{6}=2\bar{C}_{x}+2c_{u1}(\bar{C}_{z}+mr\sigma_{z}^2)(\lambda_{x}^{0})^2(\lambda_{z}^{0})^{-2}+2mc_{u2}(\lambda_{x}^{0})^2+n\sigma_{x}^2$.

Substituting~\eqref{1TApp18} and~\eqref{2TApp12} into~\eqref{2TApp10} yields
\begin{equation}
\begin{aligned}
\textstyle \mathbb{E}[\|\Xi_{x}^{t}-\Xi_{x}^{t-1}\|^2]&\textstyle\leq\frac{2m\sigma_{g,1}^2d_{\Omega}^2}{t+1}+\frac{d_{\Omega}^2\bar{L}_{g}^22^{2\varsigma_{x}+1}d_{6}}{(t+1)^{2\varsigma_{x}}}\\
&\textstyle\quad+\frac{2^{2\varsigma_{x}+1}d_{1}L_{g}^2}{(t+1)^{2\varsigma_{x}}}\leq\frac{d_{7}}{(t+1)^{\beta_{1}}},\label{2TApp13}
\end{aligned}
\end{equation}
where $\beta_{1}$ is given by $\beta_{1}=\min\{2\varsigma_{x},1\}$ and $d_{7}$ is given by $d_{7}=2md_{\Omega}^2\sigma_{g,1}^2+2^{2\varsigma_{x}+1}d_{6}d_{\Omega}^2\bar{L}_{g}^2+2^{2\varsigma_{x}+1}d_{1}L_{g}^2$.

By using an argument similar to the derivation of~\eqref{firstlemmaresult}, except that $\mathbb{E}[|\mathbf{\bar{x}}^{t}-\mathbf{\bar{x}}^{t-1}|^2]$ in~\eqref{1TApp16} is replaced with $\mathbb{E}[\|\Xi_{x}^{t}-\Xi_{x}^{t-1}\|^2]$, the right hand side of~\eqref{2TApp2} satisfies
\begin{equation}
\begin{aligned}
&\textstyle \mathbb{E}[\langle \nabla \boldsymbol{g}^{t}(\boldsymbol{x}^t)^{\top}(\mathbf{\bar{x}}^{t}\!-\!\boldsymbol{x}^*),\\
&\textstyle \quad\boldsymbol{1}_{m}\otimes \nabla_{y}\bar{f}^{t}(\boldsymbol{x}^{t},\boldsymbol{1}_{m}\!\otimes\! g(\boldsymbol{x}^{t}))-\boldsymbol{1}_{m}\!\otimes\!\nabla_{y}\bar{f}^{t}(\boldsymbol{x}^{t},\boldsymbol{\tilde{y}}^{t})\rangle]\\
&\textstyle\leq \frac{\kappa(\sigma_{g,1}^2+L_{g}^2)}{2}\mathbb{E}[\|\mathbf{\bar{x}}^{t}-\boldsymbol{x}^*\|^2]+\frac{d_{8}}{\kappa(t+1)^{\beta_{1}}}+\frac{d_{9}}{(t+1)^{\alpha_{2}}},\label{2TApp30}
\end{aligned}
\end{equation}
where $\beta_{1}=\min\{1,2\varsigma_{x}\}$, $\alpha_{2}=\min\{\frac{1}{2}-v_{y}+\varsigma_{y},\varsigma_{x}-v_{y}+\varsigma_{y}\}$, $d_{8}=4d_{0}+d_{5}$, and $d_{9}=2^{\alpha_{2}+2}\sqrt{d_{2}d_{7}}(\ln(\sqrt{1-|\rho_{2}|})e)^{-2}(1-\sqrt{1-|\rho_{2}|})^{-1}(\lambda_{y}^{0})^{-1}$.

Following an argument similar to the derivation of~\eqref{2TApp1}, for any $\kappa\!>\!0$, the third term on the right hand side of~\eqref{2TApp0} satisfies
\begin{flalign}
&\textstyle\mathbb{E}[\langle \mathbf{\bar{x}}^{t}-\boldsymbol{x}^*, \nabla \boldsymbol{g}^{t}(\boldsymbol{x}^{t})(\boldsymbol{1}_{m}\otimes\nabla_{y}\bar{f}(\boldsymbol{x}^{t},\boldsymbol{1}_{m}\otimes g(\boldsymbol{x}^{t})))\nonumber\\
&\textstyle\quad-\nabla \boldsymbol{g}^{t}(\boldsymbol{x}^{t})(\boldsymbol{1}_{m}\otimes(\nabla_{y}\bar{f}^{t}(\boldsymbol{x}^{t},\boldsymbol{1}_{m}\otimes g(\boldsymbol{x}^{t}))))\rangle]\nonumber\\
&\textstyle\leq \frac{\kappa}{4}\mathbb{E}[\| \mathbf{\bar{x}}^{t}-\boldsymbol{x}^*\|^2]+\frac{1}{\kappa}\mathbb{E}[\|\nabla \boldsymbol{g}^{t}(\boldsymbol{x}^{t})\|^2]\nonumber\\
&\textstyle\quad\times m\mathbb{E}[\|\nabla_{y}\bar{f}(\boldsymbol{x}^{t},\boldsymbol{1}_{m}\otimes g(\boldsymbol{x}^{t}))-\nabla_{y}\bar{f}^{t}(\boldsymbol{x}^{t},\boldsymbol{1}_{m}\otimes g(\boldsymbol{x}^{t}))\|^2]\nonumber\\
&\textstyle\leq\frac{\kappa}{4}\mathbb{E}[\| \mathbf{\bar{x}}^{t}-\boldsymbol{x}^*\|^2]+\frac{m(\sigma_{g,1}^2+L_{g}^2)\sigma^2_{f,1}}{\kappa(t+1)},\label{2TApp3}
\end{flalign}
where we have used Assumptions~\ref{A1} and~\ref{A2} in the last inequality.

For the sake of notational simplicity, we define $\Xi_{z}^{t}\triangleq\boldsymbol{\hat{z}}^{t}-\boldsymbol{\hat{z}}^{t-1}$,  $\Xi_{\zeta}^{t}\triangleq\boldsymbol{\zeta}^{t}-\boldsymbol{\zeta}^{t-1}$, and $\Xi_{f}^{t}\triangleq\lambda_{z}^{t}\nabla_{y}\boldsymbol{f}^{t}(\boldsymbol{x}^{t},\boldsymbol{\tilde{y}}^{t})-\lambda_{z}^{t-1}\nabla_{y}\boldsymbol{f}^{t-1}(\boldsymbol{x}^{t-1},\boldsymbol{\tilde{y}}^{t-1})$. By using an argument similar to the derivation of~\eqref{1TApp11}, we obtain
\begin{flalign}
&\textstyle \mathbb{E}[\langle \mathbf{\bar{x}}^{t}-\boldsymbol{x}^{*},\nabla \boldsymbol{g}^{t}(\boldsymbol{x}^{t})\left(\boldsymbol{1}_{m}\otimes\frac{1}{m}\sum_{i=1}^{m}\tilde{z}_{i}^{t}-\boldsymbol{\tilde{z}}^{t}\right)\rangle]\nonumber\\
&\textstyle=\mathbb{E}[\langle \mathbf{\bar{x}}^{t}\!-\!\boldsymbol{x}^*,-\frac{1}{\lambda_{z}^{t}}\nabla \boldsymbol{g}^{t}(\boldsymbol{x}^{t})(\tilde{W}\!\otimes\! I_{r})^{t}\Xi_{z}^{1}\rangle]\nonumber\\
&\textstyle+\mathbb{E}[\langle \mathbf{\bar{x}}^{t}\!-\!\boldsymbol{x}^*,-\frac{1}{\lambda_{z}^{t}}\nabla \boldsymbol{g}^{t}(\boldsymbol{x}^{t})\sum_{p=1}^{t}(\tilde{W}\!\otimes\! I_{r})^{t-p}K_{1}\Xi_{\zeta}^{p}\rangle]\nonumber\\
&\textstyle+\mathbb{E}[\langle \mathbf{\bar{x}}^{t}\!-\!\boldsymbol{x}^*,-\frac{1}{\lambda_{z}^{t}}\nabla \boldsymbol{g}^{t}(\boldsymbol{x}^{t})\!\sum_{p=1}^{t}(\tilde{W}\!\otimes\! I_{r})^{t-p}K_{2}\Xi_{f}^{p}\rangle].\label{2TApp6}
\end{flalign}

Using Lemma~\ref{lemmae} and $\|\tilde{W}\|\!<\!1\!-\!|\rho_{2}|$, for any $\kappa>\!0$, we have
\begin{equation}
\vspace{-0.2em}
\begin{aligned}
&\textstyle\mathbb{E}[\langle \mathbf{\bar{x}}^{t}-\boldsymbol{x}^*,-\frac{1}{\lambda_{z}^{t}}\nabla \boldsymbol{g}^{t}(\boldsymbol{x}^{t})(\tilde{W}\otimes I_{r})^{t}\Xi_{z}^{1}\rangle]\\
&\textstyle\leq \frac{\kappa}{4}\mathbb{E}[\|\mathbf{\bar{x}}^{t}-\boldsymbol{x}^*\|^2]+\frac{4d_{10}}{\kappa(t+1)^{2-2v_{z}}},\label{2TApp7}
\end{aligned}
\vspace{-0.2em}
\end{equation}
where $d_{10}$ is given by $d_{10}\!=\!4(\sigma_{g,1}^2+L_{g}^2)\mathbb{E}[\|\boldsymbol{\hat{z}}^{1}-\boldsymbol{\hat{z}}^{0}\|^2](\ln((1-|\rho_{2}|)^{2})e)^{-2}(\lambda_{z}^{0})^{-2}.$

Using an argument similar to the derivation of~\eqref{1TApp22}, except that $\mathbb{E}[|\mathbf{\bar{x}}^{t}-\mathbf{\bar{x}}^{t-1}|^2]$ in~\eqref{1TApp16} is replaced with $\mathbb{E}[|\Xi_{x}^{t}-\Xi_{x}^{t-1}|^2]$, the second term on the right hand side of~\eqref{2TApp6} satisfies
\begin{equation}
\vspace{-0.2em}
\begin{aligned}
&\textstyle\mathbb{E}[\langle \mathbf{\bar{x}}^{t}-\boldsymbol{x}^*,-\frac{1}{\lambda_{z}^{t}}\nabla \boldsymbol{g}^{t}(x^{t})\sum_{p=1}^{t}(\tilde{W}\otimes I_{r})^{t-p}K_{1}\Xi_{\zeta}^{p}\rangle]\\
&\textstyle\leq \sum_{p=1}^{t}\frac{\sqrt{d_{7}d_{11}}(1-|\rho_{2}|)^{t-p}}{\lambda_{z}^{0}(t+1)^{\alpha_{3}}}\leq \frac{d_{12}}{(t+1)^{\alpha_{3}}},\label{2TApp16}
\end{aligned}
\vspace{-0.2em}
\end{equation}
where $\alpha_{3}\!=\!\min\{\frac{1}{2}-v_{z}+\varsigma_{z},\varsigma_{x}-v_{z}+\varsigma_{z}\}$ and $d_{12}\!=\!2^{\alpha_{3}+2}\sqrt{d_{7}d_{11}}(\ln(\sqrt{1-|\rho_{2}|})e)^{-2}(1\!-\!\sqrt{1-|\rho_{2}|})^{-1}(\lambda_{z}^{0})^{-1}$ with $d_{11}\!=\!2^{4\varsigma_{z}+3}mr\sigma_{z}^2(\ln(1-|\rho_{2}|)e)^{-2}|\rho_{2}|^{-1}.$

Next we characterize the third term on the right hand side of~\eqref{2TApp6} by using the following decomposition:
\begin{flalign}
&\textstyle \mathbb{E}[\langle \mathbf{\bar{x}}^{t}-\boldsymbol{x}^*,-\frac{1}{\lambda_{z}^{t}}\nabla \boldsymbol{g}^{t}(\boldsymbol{x}^{t})\sum_{p=1}^{t}(\tilde{W}\otimes I_{r})^{t-p}K_{2}\Xi_{f}^{p}\rangle]\nonumber\\
&\textstyle=\mathbb{E}[\langle \mathbf{\bar{x}}^{t}-\boldsymbol{x}^*,-\frac{1}{\lambda_{z}^{t}}\nabla \boldsymbol{g}^{t}(\boldsymbol{x}^{t})\sum_{p=1}^{t}(\tilde{W}\otimes I_{r})^{t-p}K_{2}\Upsilon_{f}^{p}\rangle]\nonumber\\
&\textstyle\quad+\mathbb{E}\big[\big\langle \mathbf{\bar{x}}^{t}-\boldsymbol{x}^*,-\frac{1}{\lambda_{z}^{t}}\nabla \boldsymbol{g}^{t}(\boldsymbol{x}^{t})\sum_{p=1}^{t}(\tilde{W}\otimes I_{r})^{t-p}K_{2}\nonumber\\
&\textstyle\quad\times\lambda_{z}^{p}\left(\nabla_{y}\boldsymbol{f}(\boldsymbol{x}^{p-1},\boldsymbol{\tilde{y}}^{p})-\nabla_{y}\boldsymbol{f}(\boldsymbol{x}^{p-1},\boldsymbol{\tilde{y}}^{p-1})\right)\big\rangle\big],\label{2TApp17}
\end{flalign}
where $\Upsilon_{f}^{p}\triangleq \Xi_{f}^{p}-\lambda_{z}^{p}\left(\nabla_{y}\boldsymbol{f}(\boldsymbol{x}^{p-1},\boldsymbol{\tilde{y}}^{p})-\nabla_{y}\boldsymbol{f}(\boldsymbol{x}^{p-1},\boldsymbol{\tilde{y}}^{p-1})\right).$

By using $\|\tilde{W}\|<1-|\rho_{2}|$ and $\|K_{2}\|=1$, for any $\kappa>0$, the first term on the right hand side of~\eqref{2TApp17} satisfies
\begin{equation}
\vspace{-0.2em}
\begin{aligned}
&\textstyle \mathbb{E}[\langle \mathbf{\bar{x}}^{t}-\boldsymbol{x}^*,-\frac{1}{\lambda_{z}^{t}}\nabla \boldsymbol{g}^{t}(\boldsymbol{x}^{t})\sum_{p=1}^{t}(\tilde{W}\otimes I_{r})^{t-p}K_{2}\Upsilon_{f}^{p}\rangle]\\
&\textstyle\leq \frac{\kappa}{4}\mathbb{E}[\|\mathbf{\bar{x}}^{t}-\boldsymbol{x}^*\|^2]\\
&\textstyle\quad+\frac{(\sigma_{g,1}^2+L_{g}^2)}{\kappa(\lambda_{z}^{t})^2}\sum_{p=1}^{t}((1-|\rho_{2}|)^{2})^{t-p}\mathbb{E}[\|\Upsilon_{f}^{p}\|^2].\label{2TApp18}
\end{aligned}
\vspace{-0.2em}
\end{equation}
According to the definition of $\Upsilon_{f}^{p}$ given in~\eqref{2TApp17} and $\Xi_{f}^{p}=\lambda_{z}^{p}\nabla_{y}\boldsymbol{f}^{p}(\boldsymbol{x}^{p},\boldsymbol{\tilde{y}}^{p})-\lambda_{z}^{p-1}\nabla_{y}\boldsymbol{f}^{p-1}(\boldsymbol{x}^{p-1},\boldsymbol{\tilde{y}}^{p-1})$, we have
\begin{flalign}
&\textstyle \mathbb{E}[\|\Upsilon_{f}^{p}\|^2]=\mathbb{E}[\|\lambda_{z}^{p}\left(\nabla_{y}\boldsymbol{f}^{p}(\boldsymbol{x}^{p},\boldsymbol{\tilde{y}}^{p})-\nabla_{y}\boldsymbol{f}(\boldsymbol{x}^{p},\boldsymbol{\tilde{y}}^{p})\right)\nonumber\\
&\textstyle\quad-\lambda_{z}^{p-1}\left(\nabla_{y}\boldsymbol{f}^{p-1}(\boldsymbol{x}^{p-1},\boldsymbol{\tilde{y}}^{p-1})-\nabla_{y}\boldsymbol{f}(\boldsymbol{x}^{p-1},\boldsymbol{\tilde{y}}^{p-1})\right)\nonumber\\
&\textstyle\quad-(\lambda_{z}^{p-1}-\lambda_{z}^{p})\nabla_{y}\boldsymbol{f}(\boldsymbol{x}^{p-1},\boldsymbol{\tilde{y}}^{p-1})\nonumber\\
&\textstyle\quad+\lambda_{z}^{p}\left(\nabla_{y}\boldsymbol{f}(\boldsymbol{x}^{p},\boldsymbol{\tilde{y}}^{p})-\nabla_{y}\boldsymbol{f}(\boldsymbol{x}^{p-1},\boldsymbol{\tilde{y}}^{p})\right)\|^2].\label{2TApp19}
\end{flalign}
By substituting~\eqref{2TApp12} into~\eqref{2TApp19}, we have
\begin{equation}
\begin{aligned}
&\textstyle\mathbb{E}[\|\Upsilon_{f}^{p}\|^2]\leq \frac{m\sigma_{f}^2(\lambda_{z}^{0})^2}{(p+1)^{1+2v_{z}}}+\frac{m\sigma_{f}^2(\lambda_{z}^{0})^2}{p^{1+2v_{z}}}+\frac{2mL_{f}^2v_{z}^2(\lambda_{z}^{0})^2}{p^{2+2v_{z}}}\\
&\textstyle\quad+\frac{\bar{L}_{f}^2(\lambda_{z}^{0})^22^{2\varsigma_{x}+1}d_{6}}{(p+1)^{2\varsigma_{x}+2v_{z}}}\leq \frac{2^{\alpha_{4}}d_{13}}{(p+1)^{\alpha_{4}}},\label{2TApp21}
\end{aligned}
\end{equation}
where  $\alpha_{4}=\min\{2\varsigma_{x}+2v_{z},1+2v_{z}\}$ and $d_{13}=2(\lambda_{z}^{0})^2(m\sigma_{f}^2+mv_{z}^2L_{f}^2+2^{2\varsigma_{x}}d_{6}\bar{L}_{f}^2)$.

By substituting~\eqref{2TApp21} into~\eqref{2TApp18} and using Lemma~\ref{lemmaexp}, the first term on the right hand side of~\eqref{2TApp17} satisfies
\begin{equation}
\begin{aligned}
&\textstyle \mathbb{E}[\langle \mathbf{\bar{x}}^{t}-\boldsymbol{x}^*,-\frac{1}{\lambda_{z}^{t}}\nabla \boldsymbol{g}^{t}(\boldsymbol{x}^{t})\sum_{p=1}^{t}(\tilde{W}\otimes I_{r})^{t-p}K_{2}\Upsilon_{f}^{p}\rangle]\\
&\textstyle\leq \frac{\kappa}{4}\mathbb{E}[\|\mathbf{\bar{x}}^{t}-\boldsymbol{x}^*\|^2]+ \frac{d_{14}}{\kappa(t+1)^{\beta_{1}}},\label{2TApp22}
\end{aligned}
\end{equation}
where $\beta_{1}=\min\{2\varsigma_{x},1\}$ and $d_{14}=\frac{2^{2\alpha_{4}+2}d_{13}(\sigma_{g,1}^2+L_{g}^2)}{(\lambda_{z}^{0})^2(\ln(1-|\rho_{2}|)e)^2|\rho_{2}|}$.

We proceed to characterize the second term on the right hand side of~\eqref{2TApp17}. By using the mean-value theorem, we have
\begin{equation}
\begin{aligned}
&\textstyle\nabla_{y}\boldsymbol{f}(\boldsymbol{x}^{p-1},\boldsymbol{\tilde{y}}^{p})-\nabla_{y}\boldsymbol{f}(\boldsymbol{x}^{p-1},\boldsymbol{\tilde{y}}^{p-1})\\
&\textstyle=\frac{\partial^2\boldsymbol{f}(\boldsymbol{x}^{p-1},\tau^{p})}{\partial y^2}(\boldsymbol{\tilde{y}}^{p}-\boldsymbol{\tilde{y}}^{p-1})\triangleq M_{\tau^{p}}(\boldsymbol{\tilde{y}}^{p}-\boldsymbol{\tilde{y}}^{p-1}),\label{2TApp23}
\end{aligned}
\end{equation}
where $\tau^{p}=c\boldsymbol{\tilde{y}}^{p}+(1-c)\boldsymbol{\tilde{y}}^{p-1}$ for any $c\in(0,1)$.

By using the relation $\boldsymbol{\tilde{y}}^{t}=\frac{1}{\lambda_{y}^{t}}(\boldsymbol{\hat{y}}^{t+1}-\boldsymbol{\hat{y}}^{t})+\frac{1}{\lambda_{y}^{t}}(\boldsymbol{1}_{m}\otimes(\bar{y}^{t+1}-\bar{y}^{t}))$, we have
\begin{flalign}
&\textstyle \boldsymbol{\tilde{y}}^{p}-\boldsymbol{\tilde{y}}^{p-1}=\boldsymbol{1}_{m}\otimes\left(\frac{1}{\lambda_{y}^{p}}(\bar{y}^{p+1}-\bar{y}^{p})-\frac{1}{\lambda_{y}^{p-1}}(\bar{y}^{p}\!-\!\bar{y}^{p-1})\right)\nonumber\\
&\textstyle \quad+\frac{1}{\lambda_{y}^{p}}(\boldsymbol{\hat{y}}^{p+1}-\boldsymbol{\hat{y}}^{p})-\frac{1}{\lambda_{y}^{p-1}}(\boldsymbol{\hat{y}}^{p}-\boldsymbol{\hat{y}}^{p-1}),\label{2TApp230}
\end{flalign}
where $\bar{y}^{p+1}-\bar{y}^{p}$ satisfies $\bar{y}^{p+1}-\bar{y}^{p}\textstyle=\bar{\chi}_{w}^{p}+\lambda_{y}^{p}g^{p}(\boldsymbol{x}^{p})$ based on the dynamics of $y_{i}^{t}$ in Algorithm~\ref{algorithm1}.

The first term on the right hand side of~\eqref{2TApp230} satisfies
\begin{flalign}
&\textstyle \mathbb{E}\left[\left\langle \mathbf{\bar{x}}^{t}-\boldsymbol{x}^*,-\frac{1}{\lambda_{z}^{t}}\nabla \boldsymbol{g}^{t}(\boldsymbol{x}^{t})\sum_{p=1}^{t}(\tilde{W}\otimes I_{r})^{t-p}K_{2}\lambda_{z}^{p}M_{\tau^{p}}\right.\right.\nonumber\\
&\left.\left.\textstyle\quad\times\boldsymbol{1}_{m}\otimes\left(\frac{1}{\lambda_{y}^{p}}(\bar{y}^{p+1}-\bar{y}^{p})-\frac{1}{\lambda_{y}^{p-1}}(\bar{y}^{p}\!-\!\bar{y}^{p-1})\right)\right\rangle\right]\nonumber\\
&=\textstyle \mathbb{E}\left[\left\langle \mathbf{\bar{x}}^{t}-\boldsymbol{x}^*,-\frac{1}{\lambda_{z}^{t}}\nabla \boldsymbol{g}^{t}(\boldsymbol{x}^{t})\sum_{p=1}^{t}(\tilde{W}\otimes I_{r})^{t-p}K_{2}\lambda_{z}^{p}M_{\tau^{p}}\right.\right.\nonumber\\
&\left.\left.\textstyle\quad\times \boldsymbol{1}_{m}\otimes(g^{p}(\boldsymbol{x}^{p})-g^{p-1}(\boldsymbol{x}^{p-1}))\right\rangle\right]\nonumber\\
&\textstyle\quad+\mathbb{E}\left[\left\langle \mathbf{\bar{x}}^{t}-\boldsymbol{x}^*,-\frac{1}{\lambda_{z}^{t}}\nabla \boldsymbol{g}^{t}(\boldsymbol{x}^{t})\sum_{p=1}^{t}(\tilde{W}\otimes I_{r})^{t-p}K_{2}\lambda_{z}^{p}M_{\tau^{p}}\right.\right.\nonumber\\
&\left.\left.\textstyle\quad\times \left(\boldsymbol{1}_{m}\otimes\left(\frac{\bar{\chi}_{w}^{p}}{\lambda_{y}^{p}}-\frac{\bar{\chi}_{w}^{p-1}}{\lambda_{y}^{p-1}}\right)\right)\right\rangle\right].\label{2TApp26}
\end{flalign}

By using an argument similar to the derivation of~\eqref{2TApp22}, the first term on the right hand side of~\eqref{2TApp26} satisfies
\begin{flalign}
&\textstyle \mathbb{E}[\langle \mathbf{\bar{x}}^{t}-\boldsymbol{x}^*,-\frac{1}{\lambda_{z}^{t}}\nabla \boldsymbol{g}^{t}(\boldsymbol{x}^{t})\sum_{p=1}^{t}(\tilde{W}\otimes I_{r})^{t-p}K_{2}\lambda_{z}^{p}M_{\tau^{p}}\nonumber\\
&\textstyle\quad\times \boldsymbol{1}_{m}\otimes(g^{p}(\boldsymbol{x}^{p})-g^{p-1}(\boldsymbol{x}^{p-1}))\rangle]\nonumber\\
&\textstyle\leq \frac{\kappa}{4}\mathbb{E}[\|\mathbf{\bar{x}}^{t}-\boldsymbol{x}^*\|^2]+\frac{d_{16}}{\kappa(t+1)^{\beta_{1}}},\label{2TApp29}
\end{flalign}
where $\beta_{1}=\min\{2\varsigma_{x},1\}$ and  $d_{16}=\frac{2^{2\alpha_{4}+2}d_{15}\bar{L}_{f}^2(\sigma_{g,1}^2+L_{g}^2)}{(\lambda_{z}^{0})^2(\ln(1-|\rho_{2}|)e)^2|\rho_{2}|}$ with $d_{15}=2(\lambda_{z}^{0})^2(m\sigma_{g,0}^2+2^{2\varsigma_{x}}d_{6}L_{g}^2)$.

By defining $\Upsilon_{\chi}^{p}\triangleq\boldsymbol{1}_{m}\otimes(\frac{\bar{\chi}_{w}^{p}}{\lambda_{y}^{p}}-\frac{\bar{\chi}_{w}^{p-1}}{\lambda_{y}^{p-1}})$, we have $\mathbb{E}[\|\Upsilon_{\chi}^{p}\|^2]\leq \frac{(1+2^{2\varsigma_{y}})mr\sigma_{y}^2}{(\lambda_{y}^{0})^2(p+1)^{2\varsigma_{y}-2v_{y}}}$, which implies that the second term on the right hand side of~\eqref{2TApp26} satisfies 
\begin{flalign}
&\textstyle \mathbb{E}[\langle \mathbf{\bar{x}}^{t}-\boldsymbol{x}^*,-\frac{1}{\lambda_{z}^{t}}\nabla \boldsymbol{g}^{t}(\boldsymbol{x}^{t})\sum_{p=1}^{t}(\tilde{W}\otimes I_{r})^{t-p}K_{2}M_{\tau^{p}}\lambda_{z}^{p}\Upsilon_{\chi}^{p}\rangle]\nonumber\\
&\textstyle \leq \sum_{p=1}^{t}\frac{\sqrt{d_{7}d_{17}}(1-|\rho_{2}|)^{t-p}}{\lambda_{z}^{0}(t+1)^{\alpha_{2}}}\leq \frac{d_{18}}{(t+1)^{\alpha_{2}}},\label{2TApp32}
\end{flalign}
where $\alpha_{2}$ is given by $\alpha_{2}=\min\{\frac{1}{2}-v_{y}+\varsigma_{y},\varsigma_{x}-v_{y}+\varsigma_{y}\}$ and $d_{18}$ is given by $d_{18}=\frac{2^{\alpha_{2}+2}\sqrt{d_{7}d_{17}}}{\lambda_{z}^{0}(\ln(\sqrt{1-|\rho_{2}|})e)^2(1-\sqrt{1-|\rho_{2}|})}$ with $d_{17}=\frac{2^{2\varsigma_{y}+2}mr\sigma_{y}^2(1+2^{2\varsigma_{y}})\bar{L}_{f}^2(\lambda_{z}^{0})^2}{(\lambda_{y}^{0})^2(\ln(1-|\rho_{2}|)e)^2|\rho_{2}|}$.

Substituting~\eqref{2TApp29} and~\eqref{2TApp32} into~\eqref{2TApp26} yields
\begin{flalign}
&\textstyle \mathbb{E}\left[\left\langle \mathbf{\bar{x}}^{t}-\boldsymbol{x}^*,-\frac{1}{\lambda_{z}^{t}}\nabla \boldsymbol{g}^{t}(\boldsymbol{x}^{t})\sum_{p=1}^{t}(\tilde{W}\otimes I_{r})^{t-p}K_{2}\lambda_{z}^{p}M_{\tau^{p}}\right.\right.\nonumber\\
&\left.\left.\textstyle\quad\times\boldsymbol{1}_{m}\otimes\left(\frac{1}{\lambda_{y}^{p}}(\bar{y}^{p+1}-\bar{y}^{p})-\frac{1}{\lambda_{y}^{p-1}}(\bar{y}^{p}\!-\!\bar{y}^{p-1})\right)\right\rangle\right]\nonumber\\
&\textstyle \leq \frac{\kappa}{4}\mathbb{E}[\|\mathbf{\bar{x}}^{t}-\boldsymbol{x}^*\|^2]+\frac{d_{16}}{\kappa(t+1)^{\beta_{1}}}+\frac{d_{18}}{(t+1)^{\alpha_{2}}}.\label{2TApp33}
\end{flalign}

Next we characterize the second and third terms on the right hand side of~\eqref{2TApp230}. According to the dynamics in~\eqref{2LApp5}, we have
\begin{flalign}
&\textstyle \frac{1}{\lambda_{y}^{p}}(\boldsymbol{\hat{y}}^{p+1}\!-\!\boldsymbol{\hat{y}}^{p})\!-\!\frac{1}{\lambda_{y}^{p-1}}(\boldsymbol{\hat{y}}^{p}\!-\!\boldsymbol{\hat{y}}^{p-1})\!=\!\frac{1}{\lambda_{y}^{p}}(\tilde{W}\!\otimes\! I_{r})(\boldsymbol{\hat{y}}^{p}-\boldsymbol{\hat{y}}^{p-1})\nonumber\\
&\textstyle\quad-\frac{1}{\lambda_{y}^{p-1}}(\tilde{W}\otimes I_{r})(\boldsymbol{\hat{y}}^{p-1}-\boldsymbol{\hat{y}}^{p-2})+\Theta_{\chi}^{p}+\Theta_{g}^{p},\label{2TApp35}
\end{flalign}
where $\Theta_{\chi}^{p}$ and $\Theta_{g}^{p}$ are defined as
\begin{equation}
\begin{aligned}
&\textstyle \Theta_{\chi}^{p}\triangleq\frac{1}{\lambda_{y}^{p}}K_{1}(\boldsymbol{\chi}^{p}-\boldsymbol{\chi}^{p-1})-\frac{1}{\lambda_{y}^{p-1}}K_{1}(\boldsymbol{\chi}^{p-1}-\boldsymbol{\chi}^{p-2});\\
&\textstyle \Theta_{g}^{p}\triangleq\frac{1}{\lambda_{y}^{p}}K_{2}(\lambda_{y}^{p}\boldsymbol{g}^{p}(\boldsymbol{x}^{p})-\lambda_{y}^{p-1}\boldsymbol{g}^{p-1}(\boldsymbol{x}^{p-1}))\\
&\textstyle\quad-\frac{1}{\lambda_{y}^{p-1}}K_{2}(\lambda_{y}^{p-1}\boldsymbol{g}^{p-1}(\boldsymbol{x}^{p-1})-\lambda_{y}^{p-2}\boldsymbol{g}^{p-2}(\boldsymbol{x}^{p-2})).\label{2TApp36}
\end{aligned} 
\end{equation}

Iterating~\eqref{2TApp35} from $2$ to $p$ yields
\begin{equation}
\begin{aligned}
&\textstyle \frac{1}{\lambda_{y}^{p}}(\boldsymbol{\hat{y}}^{p+1}-\boldsymbol{\hat{y}}^{p})-\frac{1}{\lambda_{y}^{p-1}}(\boldsymbol{\hat{y}}^{p}-\boldsymbol{\hat{y}}^{p-1})\\
&\textstyle=(\tilde{W}\otimes I_{r})^{p-1}\big(\frac{1}{\lambda_{y}^{2}}(\boldsymbol{\hat{y}}^{2}-\boldsymbol{\hat{y}}^{1})-\frac{1}{\lambda_{y}^{1}}(\boldsymbol{\hat{y}}^{1}-\boldsymbol{\hat{y}}^{0})\big)\\
&\textstyle\quad+\sum_{k=2}^{p}(\tilde{W}\otimes I_{r})^{p-k}\big(\Theta_{\chi}^{k}+\Theta_{g}^{k}\big).\label{2TApp37}
\end{aligned}
\end{equation}
According to~\eqref{2TApp37}, we have that the second and third terms on the right hand side of~\eqref{2TApp230} satisfy
\begin{flalign}
&\textstyle \mathbb{E}[\langle \mathbf{\bar{x}}^{t}-\boldsymbol{x}^*,-\frac{1}{\lambda_{z}^{t}}\nabla \boldsymbol{g}^{t}(\boldsymbol{x}^{t})\sum_{p=1}^{t}(\tilde{W}\otimes I_{r})^{t-p}K_{2}\lambda_{z}^{p}M_{\tau^{p}}\nonumber\\
&\textstyle\quad\times(\frac{1}{\lambda_{y}^{p}}(\boldsymbol{\hat{y}}^{p+1}-\boldsymbol{\hat{y}}^{p})-\frac{1}{\lambda_{y}^{p-1}}(\boldsymbol{\hat{y}}^{p}-\boldsymbol{\hat{y}}^{p-1}))\rangle]\nonumber\\
&\textstyle=\mathbb{E}[\langle \mathbf{\bar{x}}^{t}-\boldsymbol{x}^*,-\frac{1}{\lambda_{z}^{t}}\nabla \boldsymbol{g}^{t}(\boldsymbol{x}^{t})\sum_{p=1}^{t}(\tilde{W}\otimes I_{r})^{t-p}K_{2}M_{\tau^{p}}\nonumber\\
&\textstyle\quad\times\lambda_{z}^{p}\left((\tilde{W}\otimes I_{r})^{p-1}\big(\frac{1}{\lambda_{y}^{2}}(\boldsymbol{\hat{y}}^{2}-\boldsymbol{\hat{y}}^{1})-\frac{1}{\lambda_{y}^{1}}(\boldsymbol{\hat{y}}^{1}-\boldsymbol{\hat{y}}^{0})\big)\right)\nonumber\\
&\textstyle\quad+\mathbb{E}[\langle \mathbf{\bar{x}}^{t}-\boldsymbol{x}^*,-\frac{1}{\lambda_{z}^{t}}\nabla \boldsymbol{g}^{t}(\boldsymbol{x}^{t})\sum_{p=1}^{t}(\tilde{W}\otimes I_{r})^{t-p}K_{2}M_{\tau^{p}}\nonumber\\
&\textstyle\quad\times\lambda_{z}^{p}\sum_{k=2}^{p}(\tilde{W}\otimes I_{r})^{p-k}(\Theta_{\chi}^{k}+\Theta_{g}^{k})\rangle].\label{2TApp38}
\end{flalign}

Lemma~\ref{lemmae} implies $((1-|\rho_{2}|)^2)^{p}\leq \frac{16}{(\ln((1-|\rho_{2}|)^2)e)^2(p+1)^2}\triangleq\frac{d_{19}}{(p+1)^2}$. Then, we use Lemma~\ref{lemmaexp} to obtain
\begin{equation}
\begin{aligned}
&\textstyle\sum_{p=1}^{t}\lambda_{z}^{p}((1-|\rho_{2}|)^2)^{t-1}=\frac{\lambda_{z}^{0}\sum_{p=1}^{t}((1-|\rho_{2}|)^2)^{t-p}((1-|\rho_{2}|)^2)^{p}}{(1-|\rho_{2}|)^2(p+1)^{2v_{z}}} \\
&\textstyle \leq \sum_{p=1}^{t}((1-|\rho_{2}|)^2)^{t-p}\frac{d_{19}\lambda_{z}^{0}}{(1-|\rho_{2}|)^2(p+1)^{2v_{z}+2}}\leq \frac{d_{20}}{(t+1)^{2v_{z}+2}},\nonumber
\end{aligned}
\end{equation}
with $d_{20}=\frac{2^{2v_{z}+4}d_{19}\lambda_{z}^{0} }{(1-|\rho_{2}|)^2(\ln(1-|\rho_{2}|)e)^2|\rho_{2}|}$.

Using the preceding inequality and the Young's inequality, for any $\kappa\!>\!0$, the first term on the right hand side of~\eqref{2TApp38} satisfies
\begin{flalign}
&\textstyle \mathbb{E}[\langle \mathbf{\bar{x}}^{t}-\boldsymbol{x}^*,-\frac{1}{\lambda_{z}^{t}}\nabla \boldsymbol{g}^{t}(\boldsymbol{x}^{t})\sum_{p=1}^{t}(\tilde{W}\otimes I_{r})^{t-p}K_{2}M_{\tau^{p}}\nonumber\\
&\textstyle\quad\times\lambda_{z}^{p}\left((\tilde{W}\otimes I_{r})^{p-1}\big(\frac{1}{\lambda_{y}^{2}}(\boldsymbol{\hat{y}}^{2}-\boldsymbol{\hat{y}}^{1})-\frac{1}{\lambda_{y}^{1}}(\boldsymbol{\hat{y}}^{1}-\boldsymbol{\hat{y}}^{0})\big)\right)\nonumber\\
&\textstyle \leq \frac{\kappa}{4}\mathbb{E}[\| \mathbf{\bar{x}}^{t}-\boldsymbol{x}^*\|^2]+\frac{d_{21}}{\kappa(t+1)^{2}},\label{2TApp391}
\end{flalign}
where $d_{21}=\frac{d_{20}(\sigma_{g,1}^2+L_{g}^2)\bar{L}_{f}^2\mathbb{E}[\|\frac{1}{\lambda_{y}^{2}}(\boldsymbol{\hat{y}}^{2}-\boldsymbol{\hat{y}}^{1})-\frac{1}{\lambda_{y}^{1}}(\boldsymbol{\hat{y}}^{1}-\boldsymbol{\hat{y}}^{0})\|^2]}{(\lambda_{z}^{0})^2}$.

To characterize the second term on the right hand side of~\eqref{2TApp38}, we define an auxiliary variable $M^{k}\triangleq\sum_{p=k}^{t}\lambda_{z}^{p}(\tilde{W}\otimes I_{r})^{t-p}K_{2}M_{\tau^{p}}(\tilde{W}\otimes I_{r})^{p-k}$. Then, we have $\sum_{p=1}^{t}\lambda_{z}^{p}(\tilde{W}\otimes I_{r})^{t-p}K_{2}M_{\tau^{p}}\sum_{k=2}^{p}(\tilde{W}\otimes I_{r})^{p-k}\Theta_{\chi}^{k}=\sum_{k=2}^{t}M^{k}\Theta_{\chi}^{k}$ and the following relationship:
\begin{equation}
\begin{aligned}
&\textstyle \mathbb{E}[\langle \mathbf{\bar{x}}^{t}-\boldsymbol{x}^*,-\frac{1}{\lambda_{z}^{t}}\nabla \boldsymbol{g}^{t}(\boldsymbol{x}^{t})\sum_{p=1}^{t}(\tilde{W}\otimes I_{r})^{t-p}K_{2}M_{\tau^{p}}\\
&\textstyle\quad\times\lambda_{z}^{p}\sum_{k=2}^{p}(\tilde{W}\otimes I_{r})^{p-k}\Theta_{\chi}^{k}\rangle]\\
&\textstyle=\mathbb{E}[\langle \Xi_{x}^{t}-\Xi_{x}^{t-1},-\frac{1}{\lambda_{z}^{t}}\sum_{k=2}^{t}M^{k}\Theta_{\chi}^{k}\rangle]\\
&\textstyle\quad+\mathbb{E}[\langle \Xi_{x}^{t-1}-\Xi_{x}^{t-2},-\frac{1}{\lambda_{z}^{t}}\sum_{k=2}^{t}M^{k}\Theta_{\chi}^{k}\rangle]\\
&\textstyle\quad+\mathbb{E}[\langle \Xi_{x}^{t-2}-\Xi_{x}^{t-3},-\frac{1}{\lambda_{z}^{t}}\sum_{k=2}^{t-1}M^{k}\Theta_{\chi}^{k}\rangle]\\
&\textstyle\quad~~\vdots\\
&\textstyle\quad+\mathbb{E}[\langle \Xi_{x}^{2}-\Xi_{x}^{1},-\frac{1}{\lambda_{z}^{t}}\sum_{k=2}^{3}M^{k}\Theta_{\chi}^{k}\rangle]\\
&\textstyle\quad+\mathbb{E}[\langle \Xi_{x}^{1}-\Xi_{x}^{0},-\frac{1}{\lambda_{z}^{t}}M^{2}\Theta_{\chi}^{2}\rangle],\label{2TApp40}
\end{aligned}
\end{equation}
where in the derivation we have used the fact that $\Xi_{x}^{t-p}-\Xi_{x}^{t-p-1}$ and $M^{k}\Theta_{\chi}^{k}$ are independent for any $k\!\geq\! t-p+1$, and hence, $\mathbb{E}[\langle \Xi_{x}^{t-p}-\Xi_{x}^{t-p-1},-\frac{1}{\lambda_{z}^{t}}\sum_{k=2}^{t}M^{k}\Theta_{\chi}^{k}\rangle]=\mathbb{E}[\langle \Xi_{x}^{t-p}-\Xi_{x}^{t-p-1},-\frac{1}{\lambda_{z}^{t}}\sum_{k=2}^{t-p+1}M^{k}\Theta_{\chi}^{k}\rangle]$ holds for any $p\geq 1$.

To characterize the right hand side of~\eqref{2TApp40}, we first estimate an upper bound on $\|M^{k}\|$. According to the definition $M^{k}=\sum_{p=k}^{t}\lambda_{z}^{p}(\tilde{W}\otimes I_{r})^{t-p}K_{2}M_{\tau^{p}}(\tilde{W}\otimes I_{r})^{p-k}$, we have
\begin{equation}
\begin{aligned}
\|M^{k}\|&\textstyle\leq \bar{L}_{f}\lambda_{z}^{k}(1-|\rho_{2}|)^{t-k}(t-k+1)\\
&\textstyle \leq M_{0}\bar{L}_{f}(1-|\rho_{2}|)^{\frac{t-k}{2}}\lambda_{z}^{k},\label{2TApp41}
\end{aligned}
\end{equation}
where in the derivation we have used the relation $\sum_{p=k}^{t}\lambda_{z}^{p}\leq\lambda_{z}^{k}(t-k+1)$ and the inequality $(1-|\rho_{2}|)^{\frac{t-k}{2}}(t-k+1)\leq (1-|\rho_2|)^{\frac{c_{2}}{1-c_{2}}-\frac{1}{2}}\frac{c_{2}}{1-c_{2}}\triangleq M_{0}$ with $c_{2}=\sqrt{1-|\rho_2|}$.

Next, we estimate an upper bound on $\mathbb{E}[\|\Theta_{\chi}^{k}\|^2]$. According to the definition of $\Theta_{\chi}^{k}$ given in~\eqref{2TApp36}, we have $\mathbb{E}[\|\Theta_{\chi}^{k}\|^2]\leq\frac{d_{22}}{(k+1)^{2\varsigma_{y}-2v_{y}}}$ with $d_{22}\!=\!\frac{4^{2\varsigma_{y}+1}m\sigma_{y}^2}{(\lambda_{y}^{0})^2}$, which leads to
\begin{equation}
\textstyle \mathbb{E}[\|\sum_{k=2}^{t}M^{k}\Theta_{\chi}^{k}\|]\leq \frac{d_{23}}{(k+1)^{\varsigma_{y}-v_{y}+v_{z}}},\label{2TApp430}
\end{equation}
with $d_{23}=\frac{2^{\varsigma_{y}-v_{y}+v_{z}+2}M_{0}\bar{L}_{f}\lambda_{z}^{0}\sqrt{d_{22}}}{(\ln((1-|\rho_{2}|)^{\frac{1}{4}})e)^2(1-(1-|\rho_{2}|)^{\frac{1}{4}})}$.

By using~\eqref{2TApp430} and~\eqref{2TApp13}, the first term on the right hand side of~\eqref{2TApp40} satisfies
\begin{equation}
\textstyle \mathbb{E}[\langle \Xi_{x}^{t}-\Xi_{x}^{t-1},-\frac{1}{\lambda_{z}^{t}}\sum_{k=2}^{t}M^{k}\Theta_{\chi}^{k}\rangle]\leq \frac{\sqrt{d_{7}}d_{23}}{\lambda_{z}^{0}(t+1)^{\alpha_{2}}},\label{2TApp43}
\end{equation}
with $\alpha_{2}=\min\{\frac{1}{2}-v_{y}+\varsigma_{y},\varsigma_{x}-v_{y}+\varsigma_{y}\}$.

By using an argument similar to the derivation of~\eqref{2TApp43}, the remaining terms on the right hand side of~\eqref{2TApp40} satisfy
\begin{equation}
\begin{aligned}
&\textstyle \mathbb{E}[\langle \Xi_{x}^{t-p}-\Xi_{x}^{t-p-1},-\frac{1}{\lambda_{z}^{t}}\sum_{k=2}^{t-p+1}M^{k}\Theta_{\chi}^{k}\rangle]\\
&\textstyle \leq \frac{d_{24}(t+1)^{v_{z}}(1-|\rho_{2}|)^{\frac{p}{2}}}{(t-p+1)^{\varsigma_{y}-v_{y}+v_{z}+\frac{\beta_{1}}{2}}},~\forall p\in[1,t-2], \label{2TApp45}
\end{aligned}
\end{equation}
with $d_{24}=\frac{ 2^{\varsigma_{y}-v_{y}+v_{z}+2}\sqrt{d_{7}}M_{0}\bar{L}_{f}\lambda_{z}^{0}\sqrt{d_{22}}}{\lambda_{z}^{0}\sqrt{1-|\rho_{2}|}(\ln((1-|\rho_{2}|^{\frac{1}{4}}))e)^2(1-1-|\rho_{2}|^{\frac{1}{4}})}.$

Combining~\eqref{2TApp43}, and~\eqref{2TApp45} with~\eqref{2TApp40}, we arrive at
\begin{flalign}
&\textstyle \mathbb{E}[\langle \mathbf{\bar{x}}^{t}-\boldsymbol{x}^*,-\frac{1}{\lambda_{z}^{t}}\nabla \boldsymbol{g}^{t}(\boldsymbol{x}^{t})\sum_{p=1}^{t}(\tilde{W}\otimes I_{r})^{t-p}K_{2}M_{\tau^{p}}\nonumber\\
&\textstyle\times\lambda_{z}^{p}\sum_{k=2}^{p}(\tilde{W}\otimes I_{r})^{p-k}\Theta_{\chi}^{k}\rangle]\leq \frac{\sqrt{d_{7}}d_{23}}{\lambda_{z}^{0}(t+1)^{\alpha_{2}}}\nonumber\\
&\textstyle+d_{24}(t+1)^{v_{z}}\sum_{p=1}^{t-1}\frac{(\sqrt{1-|\rho_{2}|})^{p}}{(t-p+1)^{\varsigma_{y}-v_{y}+v_{z}+\frac{\beta_{1}}{2}}}\!\leq\! \frac{d_{25}}{(t+1)^{\alpha_{2}}},\label{2TApp460}
\end{flalign}
where in the derivation we used $\mathbb{E}[\langle \Xi_{x}^{1}-\Xi_{x}^{0},-\frac{1}{\lambda_{z}^{t}}M^{2}\Theta_{\chi}^{2}\rangle]\leq \frac{d_{24}(t+1)^{v_{z}}(1-|\rho_{2}|)^{\frac{t-1}{2}}}{2^{\varsigma_{y}-v_{y}+v_{z}+\frac{\beta_{1}}{2}}}$ and $\sum_{p=1}^{t-1}\frac{a^{p}}{(t-p+1)^{b}}=\sum_{p=1}^{t-1}\frac{a^{t-p}}{(p+1)^{b}}$ for any $a,b>0$. Here, $\alpha_{2}$ and $d_{25}$ are given by $\alpha_{2}=\min\{\frac{1}{2}-v_{y}+\varsigma_{y},\varsigma_{x}-v_{y}+\varsigma_{y}\}$ and $d_{25}=\frac{\sqrt{d_{7}}d_{23}}{\lambda_{z}^{0}}+\frac{d_{24} 2^{\varsigma_{y}-v_{y}+v_{z}+\frac{\beta_{1}}{2}+2}}{(\ln((1-|\rho_{2}|)^{\frac{1}{4}})e)^2(1-(1-|\rho_{2}|)^{\frac{1}{4}})}$, respectively.

By using the Young's inequality, for any $\kappa>0$, the last term on the right hand side of~\eqref{2TApp38} satisfies
\begin{flalign}
&\textstyle\mathbb{E}[\langle \mathbf{\bar{x}}^{t}-\boldsymbol{x}^*,-\frac{1}{\lambda_{z}^{t}}\nabla \boldsymbol{g}^{t}(\boldsymbol{x}^{t})\sum_{p=1}^{t}(\tilde{W}\otimes I_{r})^{t-p}K_{2}M_{\tau^{p}}\nonumber\\
&\textstyle\quad\times\lambda_{z}^{p}\sum_{k=2}^{p}(\tilde{W}\otimes I_{r})^{p-k}\Theta_{g}^{k}\rangle]\leq \frac{\kappa}{4}\mathbb{E}[\|\mathbf{\bar{x}}^{t}-\boldsymbol{x}^*\|^2]\nonumber\\
&\textstyle\quad+\frac{(\sigma_{g,1}^2+L_{g}^2)}{\kappa (\lambda_{z}^{t})^2}\sum_{k=2}^{t}\mathbb{E}[\|M^{k}\Theta_{g}^{k}\|^2].\label{2TApp46}
\end{flalign}

Incorporating~\eqref{gammag} and~\eqref{2TApp41} into the last term on the right hand side of~\eqref{2TApp46} and using the definition of $\Theta_{g}^{k}$ given in~\eqref{2TApp36} yield
\begin{equation}
\begin{aligned}
&\textstyle\sum_{k=2}^{t}\mathbb{E}[\|M^{k}\Theta_{g}^{k}\|^2]\leq M_{0}^2\bar{L}_{f}^2(\lambda_{z}^{0})^2\\
&\textstyle\times\sum_{k=2}^{t}(1-|\rho_{2}|)^{t-k}\frac{d_{26}}{(k+1)^{\beta_{1}+2v_{z}}}\leq \frac{d_{27}}{(t+1)^{\alpha_{4}}},\label{2TApp52}
\end{aligned}
\end{equation}
where $\alpha_{4}=\min\{2\varsigma_{x}+2v_{z},1+2v_{z}\}$ and $d_{27}=\frac{ 2^{\alpha_{4}+2}d_{26}M_{0}^2\bar{L}_{f}^2(\lambda_{z}^{0})^2}{(\ln(\sqrt{1-|\rho_{2}|})e)^2(1-\sqrt{1-|\rho_{2}|})}$ with $d_{26}\!=\!\frac{(1+2^{\alpha_{1}})2^{\alpha_{1}}d_{3}}{(\lambda_{y}^{0})^2}$ and $\alpha_{1}=\min\{2v_{x}+2\varsigma_{y},2v_{y}+1,2v_{y}+2\varsigma_{x}\}$.

Further substituting~\eqref{2TApp52} into~\eqref{2TApp46} yields the last term on the right hand side of~\eqref{2TApp38} satisfying
\begin{equation}
\begin{aligned}
&\textstyle\mathbb{E}[\langle \mathbf{\bar{x}}^{t}-\boldsymbol{x}^*,-\frac{1}{\lambda_{z}^{t}}\nabla \boldsymbol{g}^{t}(\boldsymbol{x}^{t})\sum_{p=1}^{t}(\tilde{W}\otimes I_{r})^{t-p}K_{2}M_{\tau^{p}}\\
&\textstyle\quad\!\times\!\lambda_{z}^{p}\sum_{k=2}^{p}(\tilde{W}\!\otimes\! I_{r})^{p-k}\Theta_{g}^{k}\rangle]\!\leq\! \frac{\kappa}{4}\mathbb{E}[\|\mathbf{\bar{x}}^{t}\!-\!\boldsymbol{x}^*\|^2]\!+\!\frac{d_{28}}{\kappa(t+1)^{\beta_{1}}},\label{2TApp53}
\end{aligned}
\end{equation}
where $\beta_{1}=\min\{1,2\varsigma_{x}\}$ and $d_{28}=\frac{d_{27}(\sigma_{g,1}^2+L_{g}^2)}{\mu (\lambda_{z}^{0})^2}.$

Substituting~\eqref{2TApp391},~\eqref{2TApp460}, and~\eqref{2TApp53} into~\eqref{2TApp38}, we arrive at
\begin{flalign}
&\textstyle \mathbb{E}[\langle \mathbf{\bar{x}}^{t}-\boldsymbol{x}^*,-\frac{1}{\lambda_{z}^{t}}\nabla \boldsymbol{g}^{t}(\boldsymbol{x}^{t})\sum_{p=1}^{t}(\tilde{W}\otimes I_{r})^{t-p}K_{2}\lambda_{z}^{p}M_{\tau^{p}}\nonumber\\
&\textstyle\quad\times(\frac{1}{\lambda_{y}^{p}}(\boldsymbol{\hat{y}}^{p+1}-\boldsymbol{\hat{y}}^{p})-\frac{1}{\lambda_{y}^{p-1}}(\boldsymbol{\hat{y}}^{p}-\boldsymbol{\hat{y}}^{p-1}))\rangle]\nonumber\\
&\textstyle\leq \frac{\kappa}{2}\mathbb{E}[\| \mathbf{\bar{x}}^{t}-\boldsymbol{x}^*\|^2]+\frac{d_{21}+d_{28}}{\kappa(t+1)^{\beta_{1}}}+\frac{d_{25}}{(t+1)^{\alpha_{2}}},\label{2TApp54}
\end{flalign}
where $\beta_{1}\!=\!\min\{1,2\varsigma_{x}\}$ and $\alpha_{2}\!=\!\min\{\frac{1}{2}-v_{y}+\varsigma_{y},\varsigma_{x}-v_{y}+\varsigma_{y}\}$.

By using~\eqref{2TApp23},~\eqref{2TApp230},~\eqref{2TApp33}, and~\eqref{2TApp54}, we have that the second term on the right hand side of~\eqref{2TApp17} satisfies
\begin{equation}
\begin{aligned}
&\textstyle \mathbb{E}\big[\big\langle \mathbf{\bar{x}}^{t}-\boldsymbol{x}^*,-\frac{1}{\lambda_{z}^{t}}\nabla \boldsymbol{g}^{t}(\boldsymbol{x}^{t})\sum_{p=1}^{t}(\tilde{W}\otimes I_{r})^{t-p}K_{2}\\
&\textstyle\quad\times\lambda_{z}^{p}\left(\nabla_{y}\boldsymbol{f}(\boldsymbol{x}^{p-1},\boldsymbol{\tilde{y}}^{p})-\nabla_{y}\boldsymbol{f}(\boldsymbol{x}^{p-1},\boldsymbol{\tilde{y}}^{p-1})\right)\big\rangle\big]\\
&\textstyle \leq \frac{3\kappa}{4}\mathbb{E}[\|\mathbf{\bar{x}}^{t}-\boldsymbol{x}^*\|^2]+\frac{d_{16}+d_{21}+d_{28}}{\kappa(t+1)^{\beta_{1}}}+\frac{d_{18}+d_{25}}{(t+1)^{\alpha_{2}}}.\label{2TApp55}
\end{aligned}
\end{equation}

Further substituting~\eqref{2TApp22} and~\eqref{2TApp55} into~\eqref{2TApp17}, we have that the third term on the right hand side of~\eqref{2TApp6} satisfies
\begin{equation}
\begin{aligned}
&\textstyle \mathbb{E}[\langle \mathbf{\bar{x}}^{t}-\boldsymbol{x}^*,-\frac{1}{\lambda_{z}^{t}}\nabla \boldsymbol{g}^{t}(\boldsymbol{x}^{t})\sum_{p=1}^{t}(\tilde{W}\otimes I_{r})^{t-p}K_{2}\Xi_{f}^{p}\rangle]\\
&\textstyle \leq \kappa\mathbb{E}[\|\mathbf{\bar{x}}^{t}-\boldsymbol{x}^*\|^2]+\frac{d_{14}+d_{16}+d_{21}+d_{28}}{\kappa(t+1)^{\beta_{1}}}+\frac{d_{18}+d_{25}}{(t+1)^{\alpha_{2}}}.\label{2TApp56}
\end{aligned}
\end{equation}

Substituting~\eqref{2TApp7},~\eqref{2TApp16}, and~\eqref{2TApp56} into~\eqref{2TApp6}, we have that the fourth term on the right hand side of~\eqref{2TApp0} satisfies
\begin{flalign}
&\textstyle \sum_{i=1}^{m}\mathbb{E}[\langle \bar{\mathrm{x}}_{(i)}^{t}-x_{i}^*,\nabla g_{i}^{t}(x_{i}^{t})\left(\frac{1}{m}\sum_{i=1}^{m}\tilde{z}_{i}^{t}-\tilde{z}_{i}^{t}\right)\rangle]\nonumber\\
&\textstyle\leq \frac{5\kappa}{4}\mathbb{E}[\|\mathbf{\bar{x}}^{t}-\boldsymbol{x}^*\|^2]+\frac{d_{29}}{\kappa(t+1)^{\beta_{1}}}+\frac{d_{30}}{(t+1)^{\beta_{3}}},\label{2TApp57}
\end{flalign}
where $\beta_{1}$, $\beta_{3}$, $d_{29}$, and $d_{30}$ are given by $\beta_{1}=\min\{1,2\varsigma_{x}\}$, $\beta_{3}=\min\{\varsigma_{x}-v_{y}+\varsigma_{y},\varsigma_{x}-v_{z}+\varsigma_{z},\frac{1}{2}-v_{y}+\varsigma_{y},\frac{1}{2}-v_{z}+\varsigma_{z}\}$, $d_{29}=d_{10}+d_{14}+d_{16}+d_{21}+d_{28}$, and $d_{30}=d_{12}+d_{18}+d_{25}$, respectively.

Combining~\eqref{2TApp1},~\eqref{2TApp30},~\eqref{2TApp3}, and~\eqref{2TApp57} with~\eqref{2TApp0}, the second term on the right hand side of~\eqref{1TApp3} satisfies
\begin{flalign}
&\textstyle \mathbb{E}[\langle\mathbf{\bar{x}}^{t}-\boldsymbol{x}^{*},\nabla \boldsymbol{g}(\boldsymbol{x}^{t})\nonumber\\
&\textstyle\quad\times(\boldsymbol{1}_{m}\otimes \nabla_{y}\bar{f}(\boldsymbol{x}^{t},\boldsymbol{1}_{m}\otimes g(\boldsymbol{x}^{t}))-\nabla \boldsymbol{g}^{t}(\boldsymbol{x}^{t})\boldsymbol{\tilde{z}}^{t})\rangle]\nonumber\\
&\textstyle \leq \frac{(7+2(\sigma_{g,1}^2+L_{g}^2))\kappa}{4}\mathbb{E}[\|\mathbf{\bar{x}}^{t}-\boldsymbol{x}^*\|^2]+\frac{m\sigma_{g,1}^2L_{f}^2+mL_{g}^2\sigma_{f}^2}{\kappa(t+1)}\nonumber\\
&\textstyle\quad +\frac{d_{8}+d_{29}}{\kappa(t+1)^{\beta_{1}}}+\frac{d_{9}}{(t+1)^{\alpha_{2}}}+\frac{d_{30}}{(t+1)^{\beta_{3}}}\nonumber\\
&\textstyle\leq \frac{(7\!+\!2(\sigma_{g,1}^2\!+\!L_{g}^2))\kappa}{4}\mathbb{E}[\|\mathbf{\bar{x}}^{t}\!-\!\boldsymbol{x}^*\|^2]\!+\!\frac{b_{3}}{\kappa(t+1)^{\beta_{1}}}\!+\!\frac{b_{4}}{(t+1)^{\beta_{3}}},\label{2TApp58}
\end{flalign}
where $b_{3}$, $b_{4}$, $\beta_{1}$, and $\beta_{3}$ are given in the lemma statement. 
\end{proof}

By substituting~\eqref{firstlemmaresult} and~\eqref{secondlemmaresult} into~\eqref{1TApp3}, we obtain
\begin{flalign}
\textstyle \mathbb{E}[\langle \mathbf{\bar{x}}^{t}-\boldsymbol{x}^{*},\boldsymbol{\breve{u}}^{t}-\boldsymbol{u}^{t}\rangle]&\textstyle\leq\frac{(9+2(\sigma_{g,1}^2\!+L_{g}^2))\kappa}{4}\mathbb{E}[\|\mathbf{\bar{x}}^{t}-\boldsymbol{x}^{*}\|^2]\nonumber\\
&\textstyle\quad+\frac{b_{5}}{\kappa(t+1)^{\beta_{1}}}+\frac{b_{6}}{(t+1)^{\beta_{3}}},\label{uhatbound}
\end{flalign}
where $\beta_{1}=\min\{2\varsigma_{x},1\}$, $\beta_{3}=\min\{\varsigma_{x}-v_{y}+\varsigma_{y},\varsigma_{x}-v_{z}+\varsigma_{z},\frac{1}{2}-v_{y}+\varsigma_{y},\frac{1}{2}-v_{z}+\varsigma_{z}\}$, $b_{5}=b_{1}+b_{3}$, and $b_{6}=b_{2}+b_{4}$, with $b_{1}$ and $b_{2}$ given in~\eqref{firstlemmaresult} and $b_{3}$ and $b_{4}$ given in~\eqref{secondlemmaresult}.

\subsection{Proof of Theorem~\ref{T1}}
(i) Convergence rate when $F(\boldsymbol{x})$ is strongly convex.

By defining $\vartheta_{w(i)}^{t}\triangleq\sum_{j\in\mathcal{N}_{i}}w_{ij}\vartheta_{j(i)}^{t}$ and using the projection inequality, we have 
\begin{equation}
\textstyle\|\bar{\mathrm{x}}_{(i)}^{t+1}- x_{i}^{*}\|^2\leq\|\bar{\mathrm{x}}_{(i)}^{t}+\bar{\vartheta}_{w(i)}^{t}-\frac{\lambda_{x}^{t}}{m}u_{i}^{t}-x_{i}^{*}\|.\label{T11}
\end{equation}
Based on the definitions of $\boldsymbol{u}^{t}$ and $\boldsymbol{x}^{t}$, we have
\begin{flalign}
&\textstyle\mathbb{E}[\|\mathbf{\bar{x}}^{t+1}- \boldsymbol{x}^{*}\|^2]\leq \mathbb{E}[\|\mathbf{\bar{x}}^{t}- \boldsymbol{x}^{*}\|^2]+\frac{(\lambda_{x}^{t})^2}{m^2}\mathbb{E}[\|\boldsymbol{u}^{t}\|^2]\nonumber\\
&\textstyle\quad+n(\sigma_{x}^{t})^2-\frac{2\lambda_{x}^{t}}{m}\mathbb{E}[\langle\mathbf{\bar{x}}^{t}- \boldsymbol{x}^{*},\boldsymbol{u}^{t}\rangle].\label{1tApp1}
\end{flalign}

Using the relations $\mathbb{E}[\|\boldsymbol{u}_{(i)}^{t}\|^2]\!=\! \mathbb{E}[\|u_{i}^{t}\|^2]$ and $\mathbb{E}[\|\boldsymbol{u}^{t}\|^2]\!=\! \sum_{i=1}^{m}\mathbb{E}[\|u_{i}^{t}\|^2]$, we substitute~\eqref{ratez} into~\eqref{1AppL9} to obtain
\begin{equation}
\textstyle \frac{(\lambda_{x}^{t})^2}{m^2}\mathbb{E}[\|\boldsymbol{u}^{t}\|^2]\leq \frac{c_{2}}{(t+1)^{\varrho_{1}}},\label{1tApp2}
\end{equation}
where $\varrho_{1}=\min\{2v_{x}+2\varsigma_{y}-2v_{y},2v_{x}+2\varsigma_{z}-2v_{z}\}$ and $c_{2}=\frac{\bar{C}_{z}c_{u1}(\lambda_{x}^{0})^2}{m^2(\lambda_{z}^{0})^2}+\frac{c_{u1}r(\lambda_{x}^{0})^2\sigma_{z}^2}{m(\lambda_{z}^{0})^2}+\frac{c_{u2}(\lambda_{x}^{0})^2}{m}$.

The $\mu$-strong convexity of $F(x)$ implies that the third term on the right hand side of~\eqref{1tApp1} satisfies
\begin{equation}
\begin{aligned}
&\textstyle-\frac{2\lambda_{x}^{t}}{m}\mathbb{E}[\langle\mathbf{\bar{x}}^{t}- \boldsymbol{x}^{*},\boldsymbol{u}^{t}\rangle]\leq-\frac{\lambda_{x}^{t}\mu}{m}\mathbb{E}[\|\mathbf{\bar{x}}^{t}-\boldsymbol{x}^*\|^2]\\
&\textstyle\quad+\frac{2\lambda_{x}^{t}}{m}\mathbb{E}[\langle\mathbf{\bar{x}}^{t}\!-\!\boldsymbol{x}^*,\boldsymbol{\breve{u}}^{t}-\boldsymbol{u}^{t}\rangle].\label{1tApp4}
\end{aligned}
\end{equation}

By substituting~\eqref{uhatbound} into~\eqref{1tApp4} and further substituting~\eqref{1tApp2} and~\eqref{1tApp4} into~\eqref{1tApp1}, we obtain
\begin{equation}
\begin{aligned}
&\textstyle\mathbb{E}[\|\mathbf{\bar{x}}^{t+1}- \boldsymbol{x}^{*}\|^2]\\
&\textstyle\leq \left(1-\frac{\lambda_{x}^{t}\mu}{m}+\frac{(9+2(\sigma_{g,1}^2\!+L_{g}^2))\lambda_{x}^{t}\kappa}{2m}\right)\mathbb{E}[\|\mathbf{\bar{x}}^{t}- \boldsymbol{x}^{*}\|^2]\\
&\textstyle+\frac{2b_{5}\lambda_{x}^{t}}{\kappa m(t+1)^{\beta_{1}}}+\frac{2b_{6}\lambda_{x}^{t}}{ m(t+1)^{\beta_{3}}}+\frac{c_{2}}{(t+1)^{\varrho_{1}}}+\frac{n\sigma_{x}^2}{(t+1)^{2\varsigma_{x}}}.\label{1tApp5}
\end{aligned}
\end{equation}
By letting $\kappa=\frac{\mu}{9+2(\sigma_{g,1}^2\!+L_{g}^2)}$, inequality~\eqref{1tApp5} reduces to
\begin{equation}
\vspace{-0.2em}
\textstyle\mathbb{E}[\|\mathbf{\bar{x}}^{t+1}\!-\! \boldsymbol{x}^{*}\|^2]\!\leq\! (1-\frac{\lambda_{x}^{t}\mu}{2m})\mathbb{E}[\|\mathbf{\bar{x}}^{t}\!-\! \boldsymbol{x}^{*}\|^2]\!+\!\frac{c_{3}}{(t+1)^{\varrho_{2}}},\label{1tApp6}
\end{equation}
where $\varrho_{2}=\min\{2\varsigma_{x},v_{x}+\frac{1}{2}-v_{y}+\varsigma_{y},v_{x}+\frac{1}{2}-v_{z}+\varsigma_{z}\}$ and $c_{3}=c_{2}+n\sigma_{x}^2+\frac{2b_{5}\lambda_{x}^{0}}{\mu m(9+2(\sigma_{g,1}^2\!+L_{g}^2))}+\frac{2b_{6}\lambda_{x}^{0}}{m}.$

Using the relations $\varsigma_{y}>v_{y}-\frac{1}{2}$,~$\varsigma_{z}>v_{z}-\frac{1}{2}$, and $\varsigma_{x}>\frac{v_{x}}{2}$ from the theorem statement, we have $\varrho_{2}>v_{x}$. Therefore, according to~Lemma 4 in~\cite{Tailoring}, there exists a constant $c_{4}\triangleq\frac{\lambda_{x}^{0}\mu}{2mc_{3}}\max\{\mathbb{E}[\|\mathbf{\bar{x}}^{0}-\boldsymbol{x}^*\|^2],\frac{2mc_{3}}{\lambda_{x}^{0}\mu-2m(\varrho_{2}-v_{x})}\}$ such that the following inequality always holds:
\begin{equation}
\textstyle\mathbb{E}[\|\mathbf{\bar{x}}^{t}\!-\! \boldsymbol{x}^{*}\|^2]\leq \frac{2mc_{3}c_{4}}{\lambda_{x}^{0}\mu(t+1)^{\varrho_{2}-v_{x}}}=\frac{2mc_{3}c_{4}}{\lambda_{x}^{0}\mu(t+1)^{\beta}},\label{1tApp7}
\end{equation}
where $\beta=\min\{2\varsigma_{x}-v_{x},\frac{1}{2}-v_{y}+\varsigma_{y},\frac{1}{2}-v_{z}+\varsigma_{z}\}$.

\begin{equation}
\begin{aligned}
&\textstyle \mathbb{E}[\|\mathbf{x}_{i}^{t}-\boldsymbol{x}^*\|^2]\leq 2\mathbb{E}[\|\mathbf{x}_{i}^{t}\!-\!\mathbf{\bar{x}}^{t}\|^2]+2\mathbb{E}[\|\mathbf{\bar{x}}^{t}\!-\! \boldsymbol{x}^{*}\|^2]\\
&\textstyle \leq 2\sum_{i=1}^{m}\mathbb{E}[\|\mathbf{x}_{(i)}^{t}-\boldsymbol{1}_{m}\otimes \bar{\mathrm{x}}_{(i)}^{t}\|^2]+2\mathbb{E}[\|\mathbf{\bar{x}}^{t}\!-\! \boldsymbol{x}^{*}\|^2]\\
&\textstyle \leq \frac{2\bar{C}_{x}+4mc_{3}c_{4}(\lambda_{x}^{0}\mu^{-1})}{(t+1)^{\beta}},\label{strongconvexresult}
\end{aligned}
\end{equation}
which implies~\eqref{Theorem1results1} with $\bar{C}_{x}$ given in~\eqref{ratex}, $c_{3}$ given in~\eqref{1tApp6}, and $c_{4}$ given in the paragraph above~\eqref{1tApp7}.
\vspace{0.3em}

(ii) Convergence rate when $F(\boldsymbol{x})$ is general convex.

Recalling the definition of $\breve{u}_{i}^{t}$ given in~\eqref{udefinition}, we have $\sum_{i=1}^{m}\breve{u}_{i}^{t}=\nabla F(\boldsymbol{x}^{t})$. Using this and the convexity of $F$, the last term on the right hand side of~\eqref{1tApp1} satisfies
\begin{equation}
\begin{aligned}
&\textstyle -\frac{2\lambda_{x}^{t}}{m}\mathbb{E}[\langle\mathbf{\bar{x}}^{t}- \boldsymbol{x}^{*},\boldsymbol{u}^{t}\rangle]\leq  -\frac{2\lambda_{x}^{t}}{m}\mathbb{E}[F(\boldsymbol{x}^{t})-F(\boldsymbol{x}^{*})]\\
&\textstyle\quad-\frac{2\lambda_{x}^{t}}{m}\mathbb{E}[\langle\mathbf{\bar{x}}^{t}\!-\! \boldsymbol{x}^{t},\boldsymbol{\breve{u}}^{t}\rangle]\!+\!\frac{2\lambda_{x}^{t}}{m}\mathbb{E}[\langle\mathbf{\bar{x}}^{t}- \boldsymbol{x}^{*},\boldsymbol{\breve{u}}^{t}-\boldsymbol{u}^{t}\rangle].\label{2tApp2}
\end{aligned}
\end{equation}

Substituting~\eqref{ratex} and~\eqref{uhatbound} into~\eqref{2tApp2} and letting $\kappa=\frac{2a_{t}}{(9+2(\sigma_{g,1}^2\!+L_{g}^2))\lambda_{x}^{t}}$ for any $a_{t}>0$, we obtain
\begin{equation}
\begin{aligned}
&\textstyle -\frac{2\lambda_{x}^{t}}{m}\mathbb{E}[\langle\mathbf{\bar{x}}^{t}- \boldsymbol{x}^{*},\boldsymbol{u}^{t}\rangle]\leq  -\frac{2\lambda_{x}^{t}}{m}\mathbb{E}[F(\boldsymbol{x}^{t})-F(\boldsymbol{x}^{*})]\\
&\textstyle\quad  +\frac{a_{t}}{m}\mathbb{E}[\|\mathbf{\bar{x}}^{t}-\boldsymbol{x}^*\|^2]+\frac{(9+2(\sigma_{g,1}^2\!+L_{g}^2))b_{5}(\lambda_{x}^{t})^2}{a_{t}m(t+1)^{\beta_{1}}}\\
&\textstyle\quad+\frac{2b_{6}\lambda_{x}^{t}}{m(t+1)^{\beta_{3}}}+\frac{c_{5}}{(t+1)^{\varsigma_{x}+v_{x}}}.\label{2tApp3}
\end{aligned}
\end{equation}

Further substituting~\eqref{2tApp3} into~\eqref{1tApp1} leads to
\begin{equation}
\vspace{-0.2em}
\begin{aligned}
&\textstyle\mathbb{E}[\|\mathbf{\bar{x}}^{t+1}-\boldsymbol{x}^*\|^2]\leq -\frac{2\lambda_{x}^{t}}{m}\mathbb{E}[F(\boldsymbol{x}^{t})-F(\boldsymbol{x}^{*})]\\
&\textstyle\quad+(1+\frac{a_{t}}{m})\mathbb{E}[\|\mathbf{\bar{x}}^{t}-\boldsymbol{x}^*\|^2]+\Phi_{t},\label{2tApp4}
\end{aligned}
\vspace{-0.2em}
\end{equation}
where the term $\Phi_{t}$ is given by
\begin{equation}
\vspace{-0.2em}
\begin{aligned}
&\textstyle \Phi_{t}=\frac{(\lambda_{x}^{t})^2}{m^2}\mathbb{E}[\|\boldsymbol{u}^{t}\|^2]+n(\sigma_{x}^{t})^2+\frac{(9+2(\sigma_{g,1}^2\!+L_{g}^2))b_{5}(\lambda_{x}^{t})^2}{a_{t}m(t+1)^{\beta_{1}}}\\
&\textstyle\quad+\frac{2b_{6}\lambda_{x}^{t}}{m(t+1)^{\beta_{3}}}+\frac{c_{5}}{(t+1)^{\varsigma_{x}+v_{x}}}.\label{2tApp5}
\end{aligned}
\vspace{-0.2em}
\end{equation}

Since $F(\boldsymbol{x}^{t})\geq F(\boldsymbol{x}^{*})$ always holds, we drop the negative term 
$-\frac{2\lambda_{x}^{t}}{m}\mathbb{E}[F(\boldsymbol{x}^{t})-F(\boldsymbol{x}^{*})]$ in~\eqref{2tApp4} to obtain
\begin{equation}
\vspace{-0.2em}
\begin{aligned}
&\textstyle \mathbb{E}[\|\mathbf{\bar{x}}^{t+1}-\boldsymbol{x}^*\|^2]\leq (1+\frac{a_{t}}{m})\mathbb{E}[\|\mathbf{\bar{x}}^{t}-\boldsymbol{x}^*\|^2]+\Phi_{t}\\
&\textstyle \leq \left(\prod_{t=0}^{T}(1+a_{t})\right)\left(\mathbb{E}[\|\mathbf{\bar{x}}^{0}-\boldsymbol{x}^*\|^2]+\sum_{t=0}^{T}\Phi_{t}\right).\label{2tApp6}
\end{aligned}
\vspace{-0.2em}
\end{equation}
Using $\ln(1+u)\leq u$ holding for any $u>0$ and letting $a_{t}=\frac{1}{(t+1)^{s}}$ with $1<s<2v_{x}$ and $a_{0}=1$, we have
\begin{equation}
\vspace{-0.2em}
\begin{aligned}
&\textstyle \ln\left(\prod_{t=0}^{T}(1+a_{t})\right)=\sum_{t=0}^{T}\ln(1+a_{t})\leq \sum_{t=0}^{T}a_{t}\nonumber\\
&\textstyle\leq \sum_{t=0}^{T}\frac{1}{(t+1)^{s}}\leq 1+\int_{1}^{\infty}\frac{1}{x^{s}}dx\leq 1+\frac{1}{s-1},\label{2tApp7}
\end{aligned}
\vspace{-0.2em}
\end{equation}
which implies $\prod_{t=0}^{T}(1+a_{t})\leq e^{\frac{s}{s-1}}$. Then, inequality~\eqref{2tApp4} can be rewritten as follows:
\begin{equation}
\vspace{-0.2em}
\textstyle \mathbb{E}[\|\mathbf{\bar{x}}^{t+1}\!-\!\boldsymbol{x}^*\|^2]\leq  e^{\frac{s}{s-1}}\left(\mathbb{E}[\|\mathbf{\bar{x}}^{0}\!-\!\boldsymbol{x}^*\|^2]\!+\!\sum_{t=0}^{T}\Phi_{t}\right).\label{2tApp8}
\end{equation}
Next, we estimate an upper bound on $\sum_{t=0}^{T}\Phi_{t}$. Substituting~\eqref{1tApp2} into~\eqref{2tApp5} and using the following inequality:
\begin{equation}
\vspace{-0.2em}
\textstyle \sum_{t=0}^{T}\frac{1}{(t+1)^{s}}\leq 1+\int_{1}^{\infty}\frac{1}{x^{s}}dx\leq\frac{s}{s-1},\label{2tApp10}
\end{equation}
for any $1<s<2v_{x}$, we obtain
\begin{flalign}
&\textstyle \sum_{t=0}^{T}\Phi_{t}\leq c_{2}\max\{\frac{2v_{x}+2\varsigma_{y}-2v_{y}}{2v_{x}+2\varsigma_{y}-2v_{y}-1},\frac{2v_{x}+2\varsigma_{z}-2v_{z}}{2v_{x}+2\varsigma_{z}-2v_{z}-1}\}\nonumber\\
&\textstyle+\frac{n\sigma_{x}^2(2\varsigma_{x})}{2\varsigma_{x}-1}+\frac{(9+2(\sigma_{g,1}^2\!+L_{g}^2))b_{5}(\lambda_{x}^{0})^2(1+2v_{x}-s)}{m(2v_{x}-s)}\!+\!\frac{c_{5}(v_{x}+\varsigma_{x})}{v_{x}+\varsigma_{x}-1}\nonumber\\
&\textstyle+\frac{2b_{6}\lambda_{x}^{0}}{m}\max\{\frac{v_{x}+\frac{1}{2}-v_{y}+\varsigma_{y}}{v_{x}+\frac{1}{2}-v_{y}+\varsigma_{y}-1},\frac{v_{x}+\frac{1}{2}-v_{z}+\varsigma_{z}}{v_{x}+\frac{1}{2}-v_{z}+\varsigma_{z}-1}\}\triangleq\frac{c_{6}}{m}.\label{2tApp11}
\end{flalign}

Substituting~\eqref{2tApp11} into~\eqref{2tApp8}, we arrive at
\begin{equation}
\textstyle \mathbb{E}[\|\mathbf{\bar{x}}^{t+1}-\boldsymbol{x}^*\|^2]\leq  e^{\frac{s}{s-1}}\left(\mathbb{E}[\|\mathbf{\bar{x}}^{0}-\boldsymbol{x}^*\|^2]+\frac{c_{6}}{m}\right).\label{2tApp12}
\end{equation}

We proceed to sum both sides of~\eqref{2tApp4} from $0$ to $T$ ($T$ can be any positive integer):
\begin{flalign}
&\textstyle\sum_{t=0}^{T}\frac{2\lambda_{x}^{t}}{m}\mathbb{E}[F(\boldsymbol{x}^{t})-F(\boldsymbol{x}^{*})]\leq -\sum_{t=0}^{T}\mathbb{E}[\|\mathbf{\bar{x}}^{t+1}- \boldsymbol{x}^{*}\|^2]\nonumber\\
&\textstyle\quad+\sum_{t=0}^{T}(1+\frac{a_{t}}{m})\mathbb{E}[\|\mathbf{\bar{x}}^{t}- \boldsymbol{x}^{*}\|^2]+\sum_{t=0}^{T}\Phi_{t}.\label{2tApp13}
\end{flalign}
The first and second terms on the right hand side of~\eqref{2tApp13} can be simplified as follows:
\begin{flalign}
&\textstyle \sum_{t=0}^{T}(1+\frac{a_{t}}{m})\mathbb{E}[\|\mathbf{\bar{x}}^{t}- \boldsymbol{x}^{*}\|^2]-\sum_{t=0}^{T}\mathbb{E}[\|\mathbf{\bar{x}}^{t+1}- \boldsymbol{x}^{*}\|^2]\nonumber\\
&\textstyle \leq \frac{a_{0}}{m}\mathbb{E}[\|\mathbf{\bar{x}}^{0}- \boldsymbol{x}^{*}\|^2]+\sum_{t=1}^{T}\frac{a_{t}}{m}\mathbb{E}[\|\mathbf{\bar{x}}^{t}- \boldsymbol{x}^{*}\|^2]\nonumber\\
&\textstyle\quad+\mathbb{E}[\|\mathbf{\bar{x}}^{0}- \boldsymbol{x}^{*}\|^2]-\mathbb{E}[\|\mathbf{\bar{x}}^{T+1}- \boldsymbol{x}^{*}\|^2]\leq\frac{c_{7}}{m},\label{2tApp14}
\end{flalign}
with $\frac{c_{7}}{m}=(\frac{se^{\frac{s}{s-1}}}{m(s-1)}+1+\frac{1}{m})\mathbb{E}[\|\mathbf{\bar{x}}^{0}-\boldsymbol{x}^{*}\|^2]+\frac{c_{6}s}{m(s-1)}$. Here, we have used~\eqref{2tApp10} and~\eqref{2tApp12} in the last inequality.

Substituting~\eqref{2tApp11} and~\eqref{2tApp14} into~\eqref{2tApp13} yields the inequality $\sum_{t=0}^{T}2\lambda_{x}^{t}\mathbb{E}[F(\boldsymbol{x}^{t})-F(\boldsymbol{x}^{*})]\leq c_{6}+c_{7}$.

By using~\eqref{ratex} and Assumption~\ref{A1}, we have
\begin{flalign}
&\textstyle\sum_{t=0}^{T}2\lambda_{x}^{t}\mathbb{E}[F(\mathbf{x}_{i}^{t})-F(\boldsymbol{x}^{*})]=\textstyle\sum_{t=0}^{T}2\lambda_{x}^{t}\mathbb{E}\left[F(\mathbf{x}_{i}^{t})-F(\mathbf{\bar{x}}^{t}),\right.\nonumber\\
&\left.\textstyle\quad+F(\mathbf{\bar{x}}^{t})-F(\boldsymbol{x}^{t})\right]+\textstyle\sum_{t=0}^{T}2\lambda_{x}^{t}\mathbb{E}[F(\boldsymbol{x}^{t})-F(\boldsymbol{x}^{*})]\nonumber\\
&\textstyle\leq \sum_{t=0}^{T}4\lambda_{x}^{t}L_{f}\sum_{i=1}^{m}\mathbb{E}[\|\mathbf{x}_{(i)}^{t}-\boldsymbol{1}_{m}\otimes \bar{\mathrm{x}}_{(i)}^{t}\|^2]+c_{6}+c_{7}\nonumber\\
&\textstyle\leq \frac{4\lambda_{0}^{t}L_{f}\sqrt{\bar{C}_{x}}(v_{x}+\varsigma_{x})}{v_{x}+\varsigma_{x}-1}+c_{6}+c_{7}\triangleq c_{8}.\nonumber
\end{flalign}
Further using $\lambda_{x}^{T}\leq \lambda_{x}^{t}$ holding for any $t\in[0,T]$, we have $\frac{1}{T+1}\sum_{t=0}^{T}\mathbb{E}[F(\mathbf{x}_{i}^{t})-F(\boldsymbol{x}^{*})]\leq \frac{c_{8}}{2\lambda_{x}^{0}(T+1)^{1-v_{x}}}$.
\vspace{0.2em}

(iii) Convergence rate when $F(\boldsymbol{x})$ is nonconvex.

To prove the convergence of Algorithm~\ref{algorithm1} for nonconvex objective functions, we first need to estimate an upper bound on $\mathbb{E}[\langle \nabla F(\mathbf{\bar{x}}^{t})-\nabla F(\boldsymbol{x}^{*}),\boldsymbol{\breve{u}}^{t}-\boldsymbol{u}^{t}\rangle]$. 

By using an argument similar to the derivation of~\eqref{1TApp3}, we use the following decomposition:
\begin{flalign}
&\textstyle\mathbb{E}[\langle \nabla F(\mathbf{\bar{x}}^{t})-\nabla F(\boldsymbol{x}^{*}),\boldsymbol{\breve{u}}^{t}-\boldsymbol{u}^{t}\rangle]=\mathbb{E}[\langle\nabla F(\mathbf{\bar{x}}^{t})\!-\!\nabla F(\boldsymbol{x}^{*}),\nonumber\\
&\textstyle\quad\nabla_{x}\boldsymbol{f}(\boldsymbol{x}^{t},\boldsymbol{1}_{m}\otimes g(\boldsymbol{x}^{t}))-\nabla_{x}\boldsymbol{f}^{t}(\boldsymbol{x}^{t},\boldsymbol{\tilde{y}}^{t})\rangle]\nonumber\\
&\textstyle\quad+\mathbb{E}[\langle\nabla F(\mathbf{\bar{x}}^{t})-\nabla F(\boldsymbol{x}^{*}),\nabla \boldsymbol{g}(\boldsymbol{x}^{t})\nonumber\\
&\textstyle\quad\times(\boldsymbol{1}_{m}\otimes \nabla_{y}\bar{f}(\boldsymbol{x}^{t},\boldsymbol{1}_{m}\otimes g(\boldsymbol{x}^{t}))-\nabla \boldsymbol{g}^{t}(\boldsymbol{x}^{t})\boldsymbol{\tilde{z}}^{t})\rangle].\label{3TApp3}
\end{flalign}

By using an argument similar to the derivation of~\eqref{firstlemmaresult}, the first term on the right hand side of~\eqref{3TApp3} satisfies
\begin{flalign}
&\textstyle\mathbb{E}[\langle\nabla F(\mathbf{\bar{x}}^{t})\!-\!\nabla F(\boldsymbol{x}^{*}),\!\nabla_{x}\boldsymbol{f}(\boldsymbol{x}^{t},\!\boldsymbol{1}_{m}\!\otimes\! g(\boldsymbol{x}^{t}))\!-\!\nabla_{x}\boldsymbol{f}^{t}(\boldsymbol{x}^{t},\boldsymbol{\tilde{y}}^{t})\rangle]\nonumber\\
&\textstyle\leq \frac{\kappa}{2}\mathbb{E}[\|\nabla F(\mathbf{\bar{x}}^{t})\!-\!\nabla F(\boldsymbol{x}^{*})\|^2]+\frac{b_{1}}{\kappa(t+1)^{\beta_{1}}}+\frac{L_{F}b_{2}}{(t+1)^{\beta_{2}}},\label{firstlemmaresult1}
\end{flalign}
where $b_{1}$ and $b_{2}$ are the same as those given in~\eqref{firstlemmaresult}.

We characterize the second term on the right hand side of~\eqref{3TApp3} by estimating an upper bound on the term $\mathbb{E}[\|\Xi_{x}^{\prime t}-\Xi_{x}^{\prime t-1}\|]$ with $\Xi_{x}^{\prime t}\triangleq \nabla \boldsymbol{g}^{t}(\boldsymbol{x}^{t})^{\top}\nabla F(\mathbf{\bar{x}}^{t})$. By using an argument similar to the derivation of~\eqref{2TApp10}, we have
\begin{flalign}
&\textstyle \mathbb{E}[\|\Xi_{x}^{\prime t}-\Xi_{x}^{\prime t-1}\|^2]\nonumber\\
&\textstyle\leq \mathbb{E}[\|(\nabla \boldsymbol{g}^{t}(\boldsymbol{x}^{t})-\nabla \boldsymbol{g}(\boldsymbol{x}^{t}))^{\top}\nabla F(\mathbf{\bar{x}}^{t})\|^2]\nonumber\\
&\textstyle\quad+\mathbb{E}[\|(\nabla \boldsymbol{g}^{t-1}(\boldsymbol{x}^{t-1})-\nabla \boldsymbol{g}(\boldsymbol{x}^{t-1}))^{\top}\nabla F(\mathbf{\bar{x}}^{t-1})\|^2]\nonumber\\
&\textstyle\quad+2\mathbb{E}[\|(\nabla \boldsymbol{g}(\boldsymbol{x}^{t-1})-\nabla \boldsymbol{g}(\boldsymbol{x}^{t}))^{\top}\nabla F(\mathbf{\bar{x}}^{t})\|]\nonumber\\
&\textstyle\quad+2\mathbb{E}[\|\nabla F(\mathbf{\bar{x}}^{t})-\nabla F(\mathbf{\bar{x}}^{t-1})\|^2\|\nabla \boldsymbol{g}(\boldsymbol{x}^{t-1})\|^2]\nonumber\\
&\textstyle\leq  \frac{2m\sigma_{g,1}^2L_{f}^2}{t}+2\bar{L}_{g}^2L_{f}^2\mathbb{E}[\|\boldsymbol{x}^{t}-\boldsymbol{x}^{t-1}\|^2]\nonumber\\
&\textstyle\quad+2L_{g}^2L_{F}^2\mathbb{E}[\|\mathbf{\bar{x}}^{t}-\mathbf{\bar{x}}^{t-1}\|^2].\label{3tApp1}
\end{flalign}
Substituting~\eqref{1TApp18} and~\eqref{2TApp12} into~\eqref{2TApp10} yields
\begin{equation}
\begin{aligned}
\textstyle \mathbb{E}[\|\Xi_{x}^{\prime t}-\Xi_{x}^{\prime t-1}\|^2]&\textstyle\leq\frac{2m\sigma_{g,1}^2L_{f}^2}{t+1}+\frac{L_{f}^2\bar{L}_{g}^22^{2\varsigma_{x}+1}d_{6}}{(t+1)^{2\varsigma_{x}}}\\
&\textstyle\quad+\frac{2^{2\varsigma_{x}+1}d_{1}L_{g}^2L_{F}^2}{(t+1)^{2\varsigma_{x}}}\leq\frac{d_{7}^{\prime}}{(t+1)^{\beta_{1}}},\label{3TApp2}
\end{aligned}
\end{equation}
with  $d_{7}^{\prime}=2m\sigma_{g,1}^2L_{f}^2+2^{2\varsigma_{x}+1}d_{6}L_{f}^2\bar{L}_{g}^2+2^{2\varsigma_{x}+1}d_{1}L_{g}^2L_{F}$ and $\beta_{1}=\min\{2\varsigma_{x},1\}$.

By using an argument similar to the derivation of~\eqref{secondlemmaresult} with replacing an upper bound on $\mathbb{E}[|\Xi_{x}^{t}-\Xi_{x}^{ t-1}|^2]$ with the one provided in~\eqref{3TApp2}, we obtain that the second term on the right hand side of~\eqref{3TApp3} satisfies
\begin{flalign}
&\textstyle \mathbb{E}[\langle\nabla F(\mathbf{\bar{x}}^{t})-\nabla F(\boldsymbol{x}^{*}),\nabla \boldsymbol{g}(\boldsymbol{x}^{t})\nonumber\\
&\textstyle\quad\times(\boldsymbol{1}_{m}\otimes \nabla_{y}\bar{f}(\boldsymbol{x}^{t},\boldsymbol{1}_{m}\otimes g(\boldsymbol{x}^{t}))-\nabla \boldsymbol{g}^{t}(\boldsymbol{x}^{t})\boldsymbol{\tilde{z}}^{t})\rangle]\nonumber\\
&\textstyle\leq \frac{(7+2(\sigma_{g,1}^2\!+L_{g}^2))\kappa}{4}\mathbb{E}[\|\nabla F(\mathbf{\bar{x}}^{t})-\nabla F(\boldsymbol{x}^{*})\|^2]\nonumber\\
&\textstyle\quad+\frac{b_{3}}{\kappa(t+1)^{\beta_{1}}}\!+\!\frac{b_{4}^{\prime}}{(t+1)^{\beta_{3}}},\label{secondlemmaresult2}
\end{flalign}
where $b_{4}^{\prime}$ is the same as that in~\eqref{secondlemmaresult} except that $d_{7}$ is replaced with $d_{7}^{\prime}=2m\sigma_{g,1}^2L_{f}^2+2^{2\varsigma_{x}+1}d_{6}L_{f}^2\bar{L}_{g}^2+2^{2\varsigma_{x}+1}d_{1}L_{g}^2L_{F}$. 

By substituting~\eqref{firstlemmaresult1} and~\eqref{secondlemmaresult2} into~\eqref{3TApp3}, we obtain
\begin{equation}
\begin{aligned}
&\textstyle \mathbb{E}[\langle \nabla F(\mathbf{\bar{x}}^{t})-\nabla F(\boldsymbol{x}^{*}),\boldsymbol{\breve{u}}^{t}-\boldsymbol{u}^{t}\rangle]\\
&\textstyle\leq\frac{(9+2(\sigma_{g,1}^2\!+L_{g}^2))\kappa}{4}\mathbb{E}[\|\nabla F(\mathbf{\bar{x}}^{t})-\nabla F(\boldsymbol{x}^{*})\|^2]\\
&\textstyle\quad+\frac{b_{5}}{\kappa(t+1)^{\beta_{1}}}+\frac{b_{6}^{\prime}}{(t+1)^{\beta_{3}}},\label{uhatbound2}
\end{aligned}
\end{equation}
with $\beta_{1}=\min\{2\varsigma_{x},1\}$, $\beta_{3}=\min\{\varsigma_{x}-v_{y}+\varsigma_{y},\varsigma_{x}-v_{z}+\varsigma_{z},\frac{1}{2}-v_{y}+\varsigma_{y},\frac{1}{2}-v_{z}+\varsigma_{z}\}$, $b_{5}$ given in~\eqref{uhatbound}, and $b_{6}^{\prime}=L_{F}b_{2}+b_{4}^{\prime}$.  

The Lipschitz continuity of $F(\boldsymbol{x})$ implies
\begin{equation}
\begin{aligned}
&\textstyle \mathbb{E}[F(\mathbf{\bar{x}}^{t+1})-F(\mathbf{\bar{x}}^{t})]\leq \mathbb{E}[\langle\nabla F(\mathbf{\bar{x}}^{t}), \mathbf{\bar{x}}^{t+1}-\mathbf{\bar{x}}^{t}\rangle]\\
&\textstyle\quad+\frac{L_{F}^2}{2}\mathbb{E}[\|\mathbf{\bar{x}}^{t+1}-\mathbf{\bar{x}}^{t}\|^2].\label{3tApp3}
\end{aligned}
\end{equation}
The first term on the right hand side of~\eqref{3tApp3} satisfies
\begin{flalign}
&\textstyle \mathbb{E}[\langle\nabla F(\mathbf{\bar{x}}^{t}), \mathbf{\bar{x}}^{t+1}-\mathbf{\bar{x}}^{t}\rangle]=-\mathbb{E}[\langle\nabla F(\mathbf{\bar{x}}^{t}), \frac{\lambda_{x}^{t}}{m}\boldsymbol{u}^{t}\rangle]\nonumber\\
&\textstyle\leq \frac{\lambda_{x}^{t}}{m}\mathbb{E}[\langle\nabla F(\mathbf{\bar{x}}^{t}),\boldsymbol{\breve{u}}^{t}- \boldsymbol{u}^{t}\rangle]-\frac{\lambda_{x}^{t}}{m}\mathbb{E}[\langle\nabla F(\mathbf{\bar{x}}^{t}),\nabla F(\mathbf{\bar{x}}^{t})\rangle]\nonumber\\
&\textstyle \quad-\frac{\lambda_{x}^{t}}{m}\mathbb{E}[\langle\nabla F(\mathbf{\bar{x}}^{t}),\boldsymbol{\breve{u}}^{t}-\nabla F(\mathbf{\bar{x}}^{t})\rangle].\label{3tApp4}
\end{flalign}
By using the Young's inequality, for any $\kappa>0$, the last term on the right hand side of~\eqref{3tApp4} satisfies
\begin{equation}
\begin{aligned}
&\textstyle \mathbb{E}[\langle\nabla F(\mathbf{\bar{x}}^{t}),\boldsymbol{\breve{u}}^{t}-\nabla F(\mathbf{\bar{x}}^{t})\rangle]\\
&\textstyle\leq \frac{\kappa}{4}\mathbb{E}[\|\nabla F(\mathbf{\bar{x}}^{t})\|^2]+\frac{1}{\kappa}\mathbb{E}[\|\boldsymbol{\breve{u}}^{t}-\nabla F(\mathbf{\bar{x}}^{t})\|^2].\label{3tApp5}
\end{aligned}
\end{equation}
By using the definition of $\breve{u}_{i}^{t}$ in~\eqref{udefinition}, we have
\begin{flalign}
&\textstyle \mathbb{E}[\|\boldsymbol{\breve{u}}^{t}-\nabla F(\mathbf{\bar{x}}^{t})\|^2]\nonumber\\
&\textstyle\leq (2\bar{L}_{f}^2(1+L_{g}^2)(1+2L_{g}^2)+4\bar{L}_{g}^2L_{f}^2)\mathbb{E}[\|\boldsymbol{x}^{t}-\mathbf{\bar{x}}^{t}\|^2]\nonumber\\
&\textstyle\leq\frac{c_{9}}{(t+1)^{2\varsigma_{x}}},\label{3tApp6}
\end{flalign}
where $c_{9}$ is given by $c_{9}=\bar{C}_{x}(2\bar{L}_{f}^2(1+L_{g}^2)(1+2L_{g}^2)+4\bar{L}_{g}^2L_{f}^2)$ with $\bar{C}_{x}$ given in~\eqref{ratex}. Here, we have used $\mathbb{E}[\|\boldsymbol{x}^{t}-\mathbf{\bar{x}}^{t}\|^2]\leq \sum_{i=1}^{m}\mathbb{E}[\|\boldsymbol{x}_{(i)}^{t}-\boldsymbol{1}_{m}\otimes \bar{x}_{(i)}^{t}\|^2]\leq\frac{\bar{C}_{x}}{(t+1)^{2\varsigma_{x}}}$ in the last inequality.

By substituting~\eqref{3tApp6} into~\eqref{3tApp5} and further substituting~\eqref{uhatbound2} and~\eqref{3tApp5} into~\eqref{3tApp4}, we arrive at
\begin{equation}
\begin{aligned}
&\textstyle \mathbb{E}[\langle\nabla F(\mathbf{\bar{x}}^{t}), \mathbf{\bar{x}}^{t+1}-\mathbf{\bar{x}}^{t}\rangle]\leq -\frac{\lambda_{x}^{t}}{2m}\mathbb{E}[\|\nabla F(\mathbf{\bar{x}}^{t})\|^2]\\
&\textstyle\quad
+\frac{b_{5}\lambda_{x}^{t}}{m\kappa(t+1)^{\beta_{1}}}+\frac{b_{6}^{\prime}\lambda_{x}^{t}}{m(t+1)^{\beta_{3}}}+\frac{c_{9}}{\kappa(t+1)^{2\varsigma_{x}}}.\label{3tApp7}
\end{aligned}
\end{equation}

Substituting~\eqref{3tApp7} into~\eqref{3tApp3} and letting $\kappa\!=\!\frac{1}{5+\sigma_{g,1}^2+L_{g}^2}$ yield
\begin{equation}
\textstyle \mathbb{E}[F(\mathbf{\bar{x}}^{t+1})-F(\mathbf{\bar{x}}^{t})]\leq -\frac{\lambda_{x}^{t}}{2m}\mathbb{E}[\|\nabla F(\mathbf{\bar{x}}^{t})\|^2]+\Phi_{t},\label{3tApp8}
\end{equation}
where the term $\Phi_{t}$ is given by
\begin{equation}
\Phi_{t}\textstyle=\frac{b_{5}(5+\sigma_{g,1}^2+L_{g}^2)\lambda_{x}^{0}}{m(t+1)^{\beta_{1}+v_{x}}}\!+\!\frac{b_{6}^{\prime}\lambda_{x}^{0}}{m(t+1)^{\beta_{3}+v_{x}}}\!+\!\frac{c_{9}(5+\sigma_{g,1}^2+L_{g}^2)+2^{2\varsigma_{x}}d_{1}L_{F}^2}{(t+1)^{2\varsigma_{x}}}.\label{3tApp9}
\end{equation}

Summing both sides of~\eqref{3tApp8} from $0$ to $T$ and using the relationship $F(\boldsymbol{x}^*)\leq F(\mathbf{\bar{x}}^{T+1})$, we obtain
\begin{equation}
\textstyle\sum_{t=0}^{T}\!\lambda_{x}^{t}\mathbb{E}[\|\nabla F(\mathbf{\bar{x}}^{t})\|^2] \!\leq\! 2m\mathbb{E}[F(\mathbf{\bar{x}}^{0})-F(\boldsymbol{x}^{*})]+2m\sum_{t=0}^{T}\Phi_{t}.\nonumber
\end{equation}

Combining the preceding inequality and the relationship $\|\nabla F(\mathbf{x}_{i}^{t})\|^2\leq 2\|\nabla F(\mathbf{x}_{i}^{t})-\nabla F(\mathbf{\bar{x}}^{t})\|^2+2\|\nabla F(\mathbf{\bar{x}}^{t})\|^2$ yields
\begin{equation}
\begin{aligned}
&\textstyle \textstyle\sum_{t=0}^{T}\lambda_{x}^{t}\mathbb{E}[\|\nabla F(\mathbf{x}_{i}^{t})\|^2] \leq 4m\mathbb{E}[F(\mathbf{\bar{x}}^{0})-F(\boldsymbol{x}^{*})]\\
&\textstyle\!+\!4m\sum_{t=0}^{T}\Phi_{t}\!+\!2\sum_{t=0}^{T}\lambda_{x}^{t}\mathbb{E}[\|\nabla F(\mathbf{x}_{i}^{t})\!-\!\nabla F(\mathbf{\bar{x}}^{t})\|^2].\label{3tApp10}
\end{aligned}
\end{equation}

By using~\eqref{ratex} and~\eqref{3tApp9} the second and third terms on the right hand side of~\eqref{3tApp10} satisfy
\begin{flalign}
&\textstyle 4m\sum_{t=0}^{T}\Phi_{t}\!+\!2\sum_{t=0}^{T}\lambda_{x}^{t}\mathbb{E}[\|\nabla F(\mathbf{x}_{i}^{t})\!-\!\nabla F(\mathbf{\bar{x}}^{t})\|^2]\nonumber\\
&\textstyle\leq \frac{2\lambda_{x}^{0}L_{F}\bar{C}_{x}(v_{x}+2\varsigma_{x})}{v_{x}+2\varsigma_{x}-1}+\frac{4b_{5}(5+\sigma_{g,1}^2+L_{g}^2)\lambda_{x}^{0}(1+v_{x})}{v_{x}}\nonumber\\
&\textstyle\quad +4b_{6}^{\prime}\lambda_{x}^{0}\max\{\frac{v_{x}+\frac{1}{2}-v_{y}+\varsigma_{y}}{v_{x}+\frac{1}{2}-v_{y}+\varsigma_{y}-1},\frac{v_{x}+\frac{1}{2}-v_{z}+\varsigma_{z}}{v_{x}+\frac{1}{2}-v_{z}+\varsigma_{z}-1}\}\nonumber\\
&\textstyle\quad+\frac{4mc_{9}(5+\sigma_{g,1}^2+L_{g}^2)2\varsigma_{x}}{2\varsigma_{x}-1}+\frac{m2^{2\varsigma_{x}}d_{1}L_{F}^2\varsigma_{x}+2}{2\varsigma_{x}-1}\triangleq c_{10}.\label{3tApp11}
\end{flalign}

We define $c_{11}=4m\mathbb{E}[F(\mathbf{\bar{x}}^{0})-F(\boldsymbol{x}^{*})]$ and substitute~\eqref{3tApp11} into~\eqref{3tApp10} to obtain $\frac{1}{T+1}\sum_{t=0}^{T}\mathbb{E}[\|\nabla F(\mathbf{x}_{i}^{t})\|^2]\leq\frac{c_{10}+c_{11}}{\lambda_{x}^{0}(T+1)^{1-v_{x}}}$.

{\color{blue}\subsection{Proof of Corollary~\ref{complexity}}
	For any $i\in[m]$, the per-iteration computational cost of Algorithm~\ref{algorithm1} at iteration $t$ is $(2r+n_{i}r+n_{i})(t+1)$. Hence, the computational complexity of Algorithm~\ref{algorithm1} over $T$ iterations is $\frac{(2r+n_{i}r+n_{i})(T-1)T}{2}\approx\mathcal{O}(n_{i}rT^2)$. Furthermore, for any $i\in[m]$, the per-iteration communication cost of Algorithm~\ref{algorithm1} is $(2r+n)d_{i}$. Hence, the communication complexity of Algorithm~\ref{algorithm1} over $T$ iterations is $\mathcal{O}((2r+n)d_{i}T)$.}
\vspace{-0.8em}
\subsection{Proof of Corollary~\ref{tradeoff}}
By applying the inequality $\textstyle\sum_{t=1}^{T}\frac{1}{(t+1)^{r}}\leq \int_{0}^{T}\frac{1}{(t+1)^{r}}dx=\frac{1}{1-r}((T+1)^{1-r}-1)$ to~\eqref{4T19}, we can obtain
{\color{blue}\begin{equation}
		\textstyle\epsilon_{i}\leq \frac{\sqrt{2}C_{i,x}}{\sigma_{i,x}(v_{x}-v_{z}-\varsigma_{i,x})}+\frac{\sqrt{2}C_{i,y}}{\sigma_{i,x}(v_{y}-\varsigma_{i,y})}+\frac{\sqrt{2}C_{i,z}}{\sigma_{i,x}(v_{z}-\varsigma_{i,z})},\label{TC2}
	\end{equation}
	where $C_{i,x}$, $C_{i,y}$, and $C_{i,z}$ are given in~\eqref{4T16},~\eqref{4T14}, and~\eqref{4T17}, respectively.}

We denote the given cumulative privacy budget as $\hat{\epsilon}_{i}>0$. According to~\eqref{TC2}, we can choose DP-noise parameters as $\sigma_{i,x}\!=\!\frac{3\sqrt{2}C_{i,x}}{(v_{x}-v_{z}-\varsigma_{i,x})\hat{\epsilon}_{i}}$, $\sigma_{i,y}\!=\!\frac{3\sqrt{2}C_{i,y}}{(v_{y}-\varsigma_{i,y})\hat{\epsilon}_{i}}$, and $\sigma_{i,z}\!=\!\frac{3\sqrt{2}C_{i,z}}{(v_{z}-\varsigma_{i,z})\hat{\epsilon}_{i}}$. It is clear that a smaller $\hat{\epsilon}_{i}$ leads to larger $\sigma_{i,x}$, $\sigma_{i,y}$, and $\sigma_{i,z}$.

Next, we analyze the convergence rate of Algorithm~\ref{algorithm1} under a strongly convex $F(x)$. According to~\eqref{strongconvexresult}, we have $\mathbb{E}[\|\mathbf{x}_{i}^{t}-\boldsymbol{x}^*\|^2]\leq \frac{2\bar{C}_{x}+4mc_{3}c_{4}(\lambda_{x}^{0}\mu^{-1})}{(t+1)^{\beta}}$. The definitions of $\bar{C}_{x}$, $c_{3}$, and $c_{4}$ (which are given in~\eqref{ratex},~\eqref{1tApp6}, and the paragraph above~\eqref{1tApp7}, respectively) imply that they are positively correlated with the DP-noise parameters $\sigma_{i,x}^2$, $\sigma_{i,y}^2$, and $\sigma_{i,z}^2$, which are all inversely proportional to $\hat{\epsilon}_{i}^2$. Therefore, we obtain $\mathbb{E}[\|\mathbf{x}_{i}^{t}-\boldsymbol{x}^*\|^2]\leq \mathcal{O}(\frac{T^{-\beta}}{\hat{\epsilon}_{i}^2})$.

Similarly, the convergence rates of Algorithm~\ref{algorithm1} under a convex and nonconvex $F(\boldsymbol{x})$ satisfies $ \frac{c_{8}}{2\lambda_{x}^{0}(T+1)^{1-v_{x}}}$ and $\frac{c_{10}+c_{11}}{\lambda_{x}^{0}(T+1)^{1-v_{x}}}$, respectively. Given that $c_{8}$ and $c_{10}+c_{11}$ are both positively correlated with DP-noise parameters $\sigma_{x}^2$, $\sigma_{y}^2$, and $\sigma_{z}^2$ which are all inversely proportional to $\hat{\epsilon}_{i}^2$, we arrive at the result in Corollary~\ref{tradeoff}. 
\bibliographystyle{ieeetr}  
\bibliography{arixvbib}

@article{lixiuxian1,
	title={Distributed aggregative optimization over multi-agent networks},
	author={Li, Xiuxian and Xie, Lihua and Hong, Yiguang},
	journal={IEEE Trans. Autom. Control},
	volume={67},
	number={6},
	pages={3165--3171},
	year={2021},
	publisher={IEEE}
}

@article{EVcharging,
	title={Decentralized charging control of large populations of plug-in electric vehicles},
	author={Ma, Zhongjing and Callaway, Duncan S and Hiskens, Ian A},
	journal={IEEE Trans. Control Syst. Technol.
	},
	volume={21},
	number={1},
	pages={67--78},
	year={2011},
	publisher={IEEE}
}

@article{truthfulness,
	title={Ensuring Truthfulness in Distributed Aggregative Optimization},
	author={Chen, Ziqin and Egerstedt, Magnus and Wang, Yongqiang},
	journal={IEEE Trans. Autom. Control},
	volume={71},
	number={2},
	pages={1175--1190},
	year={2025}
}

@article{Carnevale1,
	title={Distributed online aggregative optimization for dynamic multirobot coordination},
	author={Carnevale, Guido and Camisa, Andrea and Notarstefano, Giuseppe},
	journal={IEEE Trans. Autom. Control},
	volume={68},
	number={6},
	pages={3736--3743},
	year={2022},
	publisher={IEEE}
}

@inproceedings{personalized1,
	title={A learning-based distributed algorithm for personalized aggregative optimization},
	author={Carnevale, Guido and Notarstefano, Giuseppe},
	booktitle={Proc. IEEE 61st Conf. Decis. Control (CDC)},
	pages={1576--1581},
	year={2022},
	organization={IEEE}
}

@article{yipeng,
title={Distributed projection-free algorithm for constrained aggregative optimization},
author={Wang, Tongyu and Yi, Peng},
journal={Int. J. Robust Nonlinear Control},
volume={33},
number={10},
pages={5273--5288},
year={2023},
publisher={Wiley Online Library}
}

@article{Carnevale2,
title={Nonconvex distributed feedback optimization for aggregative cooperative robotics},
author={Carnevale, Guido and Mimmo, Nicola and Notarstefano, Giuseppe},
journal={Automatica},
volume={167},
pages={111767},
year={2024},
publisher={Elsevier}
}

@article{wenguanghui2,
title={Compressed gradient tracking algorithm for distributed aggregative optimization},
author={Chen, Liyuan and Wen, Guanghui and Liu, Hongzhe and Yu, Wenwu and Cao, Jinde},
journal={IEEE Trans. Autom. Control},
volume={69},
number={10},
pages={6576--6591},
year={2024},
publisher={IEEE}
}

@article{wang2024momentum,
title={Momentum-based distributed gradient tracking algorithms for distributed aggregative optimization over unbalanced directed graphs},
author={Wang, Zhu and Wang, Dong and Lian, Jie and Ge, Hongwei and Wang, Wei},
journal={Automatica},
volume={164},
pages={111596},
year={2024},
publisher={Elsevier}
}

@article{cai2024distributed,
title={Distributed event-triggered aggregative optimization with applications to price-based energy management},
author={Cai, Xin and Xiao, Feng and Wei, Bo and Wang, Aiping},
journal={Automatica},
volume={161},
pages={111508},
year={2024},
publisher={Elsevier}
}

@article{mengmin,
title={Distributed aggregative optimization with affine coupling constraints},
author={Du, Kaixin and Meng, Min},
journal={Neural Netw.},
volume={184},
pages={107085},
year={2025},
publisher={Elsevier}
}

@article{zhou2025distributed,
title={Distributed aggregative optimization over directed networks with column-stochasticity},
author={Zhou, Qixing and Zhang, Keke and Zhou, Hao and L{\"u}, Qingguo and Liao, Xiaofeng and Li, Huaqing},
journal={J. Frankl. Inst.},
volume={362},
number={2},
pages={107492},
year={2025},
publisher={Elsevier}
}

@article{zhang2025distributed,
title={Distributed continuous-time algorithm for nonsmooth aggregative optimization over weight-unbalanced digraphs},
author={Zhang, Zheng and Yang, GuangHong},
journal={Neurocomputing},
volume={617},
pages={129022},
year={2025},
publisher={Elsevier}
}

@article{li2025distributed,
title={Distributed approximate aggregative optimization of multiple {Euler--Lagrange} systems using only sampling measurements},
author={Li, Cong and Wang, Qingling},
journal={Neurocomputing},
volume={636},
pages={130000},
year={2025},
publisher={Elsevier}
}

@article{huang2025distributed,
	title={Distributed Aggregative Optimization Algorithm for Solving Multi-Robot Formation Problem},
	author={Huang, Jingyi and Yang, Chuanhai and Wu, Shuang and Liu, Qingshan},
	journal={IEEE Trans. Control Netw. Syst.},
	volume={12},
	number={3},
	pages={2102--2114},
	year={2025},
	publisher={IEEE}
}

@article{yang2024class,
title={A Class of Distributed Online Aggregative Optimization in Unknown Dynamic Environment},
author={Yang, Chengqian and Wang, Shuang and Zhang, Shuang and Lin, Shiwei and Huang, Bomin},
journal={Mathematics},
volume={12},
number={16},
pages={2460},
year={2024},
publisher={MDPI}
}

@ARTICLE{onlinestochatsic1,
	author={Cao, Xuanyu and Zhang, Junshan and Poor, H. Vincent},
	title={Online Stochastic Optimization With Time-Varying Distributions},
	journal={IEEE Trans. Autom. Control},  
	year={2021},
	volume={66},
	number={4},
	pages={1840-1847}}

@article{onlinestochatsic2,
	title={An improved convergence analysis for decentralized online stochastic non-convex optimization},
	author={Xin, Ran and Khan, Usman A and Kar, Soummya},
	journal={IEEE Trans. Signal Process.},
	volume={69},
	pages={1842--1858},
	year={2021},
	publisher={IEEE}
}

@article{onlinestochatsic3,
	title={Online distributed stochastic gradient algorithm for nonconvex optimization with compressed communication},
	author={Li, Jueyou and Li, Chaojie and Fan, Jing and Huang, Tingwen},
	journal={IEEE Trans. Autom. Control},
	volume={69},
	number={2},
	pages={936--951},
	year={2023},
	publisher={IEEE}
}

@article{control,
	title={An accelerated distributed stochastic gradient method with momentum},
	author={Huang, Kun and Pu, Shi and Nedi{\'c}, Angelia},
	journal={Math. Program.},
	pages={1--44},
	year={2025},
	publisher={Springer}
}

@article{lixiuxian2,
	title={Distributed online convex optimization with an aggregative variable},
	author={Li, Xiuxian and Yi, Xinlei and Xie, Lihua},
	journal={IEEE Trans. Control Netw. Syst.},
	volume={9},
	number={1},
	pages={438--449},
	year={2021},
	publisher={IEEE}
}

@inproceedings{wenguanghui1,
	title={Distributed Frank-Wolfe Algorithm for Stochastic Aggregative Optimization},
	author={Chen, Liyuan and Wen, Guanghui},
	booktitle={Proc. IEEE 49th Annu. Conf. IEEE Ind. Electron. Soc.},
	pages={1--6},
	year={2023},
	organization={IEEE}
}

@inproceedings{leakage1,
	title={Machine learning models that remember too much},
	author={Song, Congzheng and Ristenpart, Thomas and Shmatikov, Vitaly},
	booktitle={Proc. ACM SIGSAC Conf. Comput. Commun. Secur.},
	pages={587--601},
	year={2017}
}

@article{leakage2,
	title={Deep leakage from gradients},
	author={Zhu, Ligeng and Liu, Zhijian and Han, Song},
	journal={Adv. Neural Inf. Process. Syst.},
	pages={14774--14784},
	volume={32},
	year={2019}
}

@article{leakage3,
	title={A comprehensive survey on poisoning attacks and countermeasures in machine learning},
	author={Tian, Zhiyi and Cui, Lei and Liang, Jie and Yu, Shui},
	journal={ACM Comput. Surv.},
	volume={55},
	number={8},
	pages={1--35},
	year={2022},
	publisher={ACM New York, NY}
}

@article{homomorphic1,
	title={Enabling privacy-preservation in decentralized optimization},
	author={Zhang, Chunlei and Wang, Yongqiang},
	journal={IEEE Trans. Control Netw. Syst.},
	volume={6},
	number={2},
	pages={679--689},
	year={2018},
	publisher={IEEE}
}

@article{homomorphic2,
title = {Privacy preserving distributed optimization using homomorphic encryption},
author = {Yang Lu and Minghui Zhu},
journal = {Automatica},
volume = {96},
pages = {314-325},
year = {2018}
}

@ARTICLE{homomorphic3,
	author={Alexandru, Andreea B. and Gatsis, Konstantinos and Shoukry, Yasser and Seshia, Sanjit A. and Tabuada, Paulo and Pappas, George J.},
	journal={IEEE Trans. Autom. Control}, 
	title={Cloud-Based Quadratic Optimization With Partially Homomorphic Encryption}, 
	year={2021},
	volume={66},
	number={5},
	pages={2357-2364}}

@ARTICLE{noisesinject1,
	author={Yan, Feng and Sundaram, Shreyas and Vishwanathan, S.V.N. and Qi, Yuan},
	journal={IEEE Trans. Knowl. Data Eng.}, 
	title={Distributed Autonomous Online Learning: Regrets and Intrinsic Privacy-Preserving Properties}, 
	year={2013},
	volume={25},
	number={11},
	pages={2483-2493}}

@ARTICLE{noisesinject2,
	author={Lou, Youcheng and Yu, Lean and Wang, Shouyang and Yi, Peng},
	journal={IEEE Trans. Cybern.}, 
	title={Privacy Preservation in Distributed Subgradient Optimization Algorithms}, 
	year={2018},
	volume={48},
	number={7},
	pages={2154-2165}}

@article{noisesinject3,
	title = {Dynamics based privacy preservation in decentralized optimization},
	author = {Huan Gao and Yongqiang Wang and Angelia Nedi{\'c}},
	journal = {Automatica},
	volume = {151},
	pages = {110878},
	year = {2023}
}

@ARTICLE{stepsizes,
	author={Wang, Yongqiang and Nedi{\'c}, Angelia},
	journal={IEEE Trans. Autom. Control}, 
	title={Decentralized Gradient Methods With Time-Varying Uncoordinated Stepsizes: Convergence Analysis and Privacy Design}, 
	year={2024},
	volume={69},
	number={8},
	pages={5352-5367}}

@ARTICLE{statedecomposition,
	author={Cheng, Huqiang and Liao, Xiaofeng and Li, Huaqing and L{\"u}, Qingguo and Zhao, You},
	journal={IEEE Trans. Signal Inf. Process. Netw.}, 
	title={Privacy-Preserving Push-Pull Method for Decentralized Optimization via State Decomposition}, 
	year={2024},
	volume={10},
	pages={513-526}}

@article{quantization,
	title = {Communication efficient privacy-preserving distributed optimization using adaptive differential quantization},
	author = {Qiongxiu Li and Richard Heusdens and Mads Græsbøll Christensen},
	journal = {Signal Process.},
	volume = {194},
	pages = {108456},
	year = {2022}}

@article{dwork2014,
	title={The algorithmic foundations of differential privacy},
	author={Dwork, Cynthia and Roth, Aaron and others},
	journal={Found. Trends Theor. Comput. Sci.,},
	volume={9},
	number={3--4},
	pages={211--407},
	year={2014},
	publisher={Now Publishers, Inc.}
}

@inproceedings{huang,
	title={Differentially private distributed optimization},
	author={Huang, Zhenqi and Mitra, Sayan and Vaidya, Nitin},
	booktitle={Proc. 16th Int. Conf. Distrib. Comput. Netw.},
	pages={1--10},
	year={2015}
}

@article{nozari2016,
	title={Differentially private distributed convex optimization via functional perturbation},
	author={Nozari, Erfan and Tallapragada, Pavankumar and Cort{\'e}s, Jorge},
	journal={IEEE Trans. Control Netw. Syst.},
	volume={5},
	number={1},
	pages={395--408},
	year={2016},
	publisher={IEEE}
}

@article{dingtie,
	title={Differentially private distributed optimization via state and direction perturbation in multiagent systems},
	author={Ding, Tie and Zhu, Shanying and He, Jianping and Chen, Cailian and Guan, Xinping},
	journal={IEEE Trans. Autom. Control},
	volume={67},
	number={2},
	pages={722--737},
	year={2021},
	publisher={IEEE}
}

@article{Tailoring,
	title={Tailoring gradient methods for differentially private distributed optimization},
	author={Wang, Yongqiang and Nedi{\'c}, Angelia},
	journal={IEEE Trans. Autom. Control},
	volume={69},
	number={2},
	pages={872--887},
	year={2023},
	publisher={IEEE}
}

@article{Assumption3topology,
	title={Robust Constrained Consensus and Inequality-Constrained Distributed Optimization With Guaranteed Differential Privacy and Accurate Convergence},
	author={Wang, Yongqiang and Nedi{\'c}, Angelia},
	journal={IEEE Trans. Autom. Control},
	volume={69},
	number={11},
	pages={7463--7478},
	year={2024},
	publisher={IEEE}
}

@article{shilin,
	title={Differential privacy in distributed optimization with gradient tracking},
	author={Huang, Lingying and Wu, Junfeng and Shi, Dawei and Dey, Subhrakanti and Shi, Ling},
	journal={IEEE Trans. Autom. Control},
	volume={69},
	number={9},
	pages={5727--5742},
	year={2024},
	publisher={IEEE}
}

@article{xie2025,
	title={Differentially private and communication-efficient distributed nonconvex optimization algorithms},
	author={Xie, Antai and Yi, Xinlei and Wang, Xiaofan and Cao, Ming and Ren, Xiaoqiang},
	journal={Automatica},
	volume={177},
	pages={112338},
	year={2025},
	publisher={Elsevier}
}

@article{he2025,
	title={Differentially Private Distributed Optimization Over Time-Varying Unbalanced Networks With Linear Convergence Rates},
	author={Yang, Zhen and He, Wangli and Yang, Shaofu},
	journal={IEEE Trans. Signal Process.},
	volume={73},
	pages={1138-1152},
	year={2025},
	publisher={IEEE}
}

@ARTICLE{DOLA,
	author={Li, Chencheng and Zhou, Pan and Xiong, Li and Wang, Qian and Wang, Ting},
	journal={IEEE Trans. Knowl. Data Eng.}, 
	title={Differentially Private Distributed Online Learning}, 
	year={2018},
	volume={30},
	number={8},
	pages={1440-1453}
}

@article{xiong,
	title={Privacy-preserving distributed online optimization over unbalanced digraphs via subgradient rescaling},
	author={Xiong, Yongyang and Xu, Jinming and You, Keyou and Liu, Jianxing and Wu, Ligang},
	journal={IEEE Trans. Control Netw. Syst.},
	volume={7},
	number={3},
	pages={1366--1378},
	year={2020},
	publisher={IEEE}
}

@article{lu2020,
	title={Privacy masking stochastic subgradient-push algorithm for distributed online optimization},
	author={L{\"u}, Qingguo and Liao, Xiaofeng and Xiang, Tao and Li, Huaqing and Huang, Tingwen},
	journal={IEEE Trans. Cybern.},
	volume={51},
	number={6},
	pages={3224--3237},
	year={2020},
	publisher={IEEE}
}

@article{zijiGT,
	title={Locally differentially private gradient tracking for distributed online learning over directed graphs},
	author={Chen, Ziqin and Wang, Yongqiang},
	journal={IEEE Trans. Autom. Control},
	year={2025},
	volume={70},
	number={5},
	pages={3040-3055}}

@ARTICLE{zhangjifeng1,
author={Wang, Jimin and Zhang, JiFeng},
journal={IEEE Trans. Autom. Control}, 
title={Differentially Private Distributed Stochastic Optimization with Time-Varying Sample Sizes}, 
year={2024},
volume={69},
number={9},
pages={6341-6348}}

@article{zhangjifeng2,
title={Differentially Private Gradient-Tracking-Based Distributed Stochastic Optimization over Directed Graphs},
author={Chen, Jialong and Wang, Jimin and Zhang, JiFeng},
journal={arXiv preprint arXiv:2501.06793},
year={2025}
}

@article{kasiviswanathan2011,
	title={What can we learn privately?},
	author={Kasiviswanathan, Shiva Prasad and Lee, Homin K and Nissim, Kobbi and Raskhodnikova, Sofya and Smith, Adam},
	journal={SIAM J. Comput.},
	volume={40},
	number={3},
	pages={793--826},
	year={2011},
	publisher={SIAM}
}

@article{zhangjiaojiaoLDP,
	title={Locally Differentially Private Online Federated Learning With Correlated Noise},
	author={Zhang, Jiaojiao and Zhu, Linglingzhi and Fay, Dominik and Johansson, Mikael},
	journal={IEEE Trans. Signal Process.},
	volume={73},
	pages={1518--1531},
	year={2025},
	publisher={IEEE}
}

@book{largelaw,
	title={Lectures on stochastic programming: modeling and theory},
	author={Shapiro, Alexander and Dentcheva, Darinka and Ruszczynski, Andrzej},
	year={2009},
	volume={9},
	publisher={Philadelphia, PA, USA:
	SIAM}
}

@article{smooth2,
	title={Personalized federated learning with differential privacy and convergence guarantee},
	author={Wei, Kang and Li, Jun and Ma, Chuan and Ding, Ming and Chen, Wen and Wu, Jun and Tao, Meixia and Poor, H Vincent},
	journal={IEEE Trans. Inf. Forensics Secur.},
	volume={18},
	pages={4488--4503},
	year={2023},
	publisher={IEEE}
}

@article{smooth3,
	title={Differential privacy in personalized pricing with nonparametric demand models},
	author={Chen, Xi and Miao, Sentao and Wang, Yining},
	journal={Oper. Res.},
	volume={71},
	number={2},
	pages={581--602},
	year={2023},
	publisher={INFORMS}
}

@article{projectionlemma,
title={Constrained consensus and optimization in multi-agent networks},
author={Nedic, Angelia and Ozdaglar, Asuman and Parrilo, Pablo A},
journal={IEEE Trans. Autom. Control},
volume={55},
number={4},
pages={922--938},
year={2010},
publisher={IEEE}
}

@article{personalized2,
	title={Federated learning of a mixture of global and local models},
	author={Hanzely, Filip and Richt{\'a}rik, Peter},
	journal={arXiv preprint arXiv:2002.05516},
	year={2020}
}

@article{personalized3,
title={Personalized federated learning with mixture of models for adaptive prediction and model fine-tuning},
author={M Ghari, Pouya and Shen, Yanning},
journal={Adv. Neural Inf. Process. Syst.},
volume={37},
pages={92155--92183},
year={2024}
}

@article{MNIST,
	title={The {MNIST} database of handwritten digits},
	author={Y. LeCun and C. Cortes  and C. Burges},
	year={1994},
	note={http://yann.lecun.com/exdb/mnist/}
}

@article{CIFAR10,
	title={Learning multiple layers of features from tiny images},
	journal={Technical Report},
	author={Krizhevsky, Alex and Hinton, Geoffrey and others},
	year={2009},
	publisher={Toronto, ON, Canada}
}

@article{LDPziji1,
	title={Locally differentially private distributed online learning with guaranteed optimality},
	author={Chen, Ziqin and Wang, Yongqiang},
	journal={IEEE Trans. Autom. Control},
	volume={70},
	number={4},
	pages={2521--2536},
	year={2025},
	publisher={IEEE}
}

@article{robust,
	title={Robustness and generalization},
	author={Xu, Huan and Mannor, Shie},
	journal={Mach. Learn.},
	volume={86},
	number={3},
	pages={391--423},
	year={2012},
	publisher={Springer}
}

@inproceedings{PLip1,
	title={Orthogonalizing Convolutional Layers with the Cayley Transform},
	author={Trockman, Asher and Kolter, J Zico},
	booktitle={Int. Conf. Learn. Represent.},
	pages={1--21},
	year={2021}
}

@article{PLip2,
	title={Lipschitz regularity of deep neural networks: analysis and efficient estimation},
	author={Virmaux, Aladin and Scaman, Kevin},
	journal={Adv. Neural Inf. Process. Syst.},
	volume={31},
	pages={3835--3844},
	year={2018}
}

@article{PLip3,
	title={Improved training of wasserstein gans},
	author={Gulrajani, Ishaan and Ahmed, Faruk and Arjovsky, Martin and Dumoulin, Vincent and Courville, Aaron C},
	journal={Adv. Neural Inf. Process. Syst.},
	volume={30},
	pages={5767--5777},
	year={2017}
}

@article{RoBoss,
	title={{RoBoSS:} A robust, bounded, sparse, and smooth loss function for supervised learning},
	author={Akhtar, Mushir and Tanveer, M and Arshad, Mohd},
	journal={IEEE Trans. Pattern Anal. Mach. Intell.},
	volume={47},
	number={1},
	pages={149--160},
	year={2024},
	publisher={IEEE}
}

@inproceedings{Smooth,
	title={Smooth Loss Functions for Deep Top-k Classification},
	author={Berrada, Leonard and Zisserman, Andrew and Kumar, M Pawan},
	booktitle={Int. Conf. Learn. Represent.},
	pages={1--25},
	year={2018}
}

@inproceedings{direct1,
	title={Gossip-based computation of aggregate information},
	author={Kempe, David and Dobra, Alin and Gehrke, Johannes},
	booktitle={Proc. IEEE Symp. Found. Comput. Sci.},
	pages={482--491},
	year={2003},
	organization={IEEE}
}

@inproceedings{direct2,
	title={Push-sum distributed dual averaging for convex optimization},
	author={Tsianos, Konstantinos I and Lawlor, Sean and Rabbat, Michael G},
	booktitle={Proc. IEEE 51st Conf. Decis. Control (CDC)},
	pages={5453--5458},
	year={2012},
	organization={IEEE}
}

@article{direct3,
	title={Distributed optimization over time-varying directed graphs},
	author={Nedi{\'c}, Angelia and Olshevsky, Alex},
	journal={IEEE Trans. Autom. Control},
	volume={60},
	number={3},
	pages={601--615},
	year={2014},
	publisher={IEEE}
}

@article{applicationDC,
	title={Aggregation Optimization-Based Secondary Control for {DC} Microgrids},
	author={Yuan, QiFan and Wang, YanWu and Liu, XiaoKang and Li, Yunwei},
	journal={IEEE Trans. Sustain. Energy},
	volume={17},
	number={1},
	pages={207--216},
	year={2025},
	publisher={IEEE}
}
\vspace{-4em}

\end{document}